\documentclass[11pt]{article}
\usepackage[utf8]{inputenc}
\usepackage[
	persons={June, Nima, Fred, Vishesh, Huy},
	authors=blocks,
	bibliosources={refs.bib}
]{pomegranate}

\usepackage{appendix}

\DeclareDocumentMathCommand{\corMat}{}{\Psi^{\operatorname{cor}}}
\DeclareDocumentMathCommand{\projBlockMarkov}{}{P^{\operatorname{proj}}}
\DeclareDocumentMathCommand{\fieldMarkov}{}{P^{\operatorname{FD}}}

\DeclareDocumentMathCommand{\DKL}{}{\D_{\operatorname{KL}}}
\DeclareOperator{\diag}
\DeclareOperator{\TV}
\DeclareDelimiter{\DTV}[\mathnormal{d_{\TV}}]{\lparen}{\rparen}
\DeclareOperator{\KL}
\DeclareOperator{\OP}
\DeclareOperator{\mix}
\DeclareOperator{\Ent}

\title{Entropic Independence II: Optimal Sampling and Concentration via Restricted Modified Log-Sobolev Inequalities}
\author[1]{Nima Anari}
\author[1]{Vishesh Jain}
\author[2]{Frederic Koehler}
\author[1]{Huy Tuan Pham}
\author[1]{Thuy-Duong Vuong}
\affil[1]{Stanford University, \url{{anari,visheshj,huypham,tdvuong}@stanford.edu}}
\affil[2]{Simons Institute, UC Berkeley, \url{fkoehler@berkeley.edu}}
\date{}

\begin{document}
    \maketitle
        
    \begin{abstract}
        We introduce a framework for obtaining tight mixing time bounds for Markov chains based on what we call restricted modified log-Sobolev inequalities. Modified log-Sobolev inequalities quantify the rate of relative entropy contraction for the Markov operator, and are notoriously difficult to establish. However, for distributions infinitesimally close to stationarity, entropy contraction becomes equivalent to variance contraction, a.k.a.\ a Poincare inequality, which is significantly easier to establish through, for example, spectral analysis. Motivated by this observation, we study restricted modified log-Sobolev inequalities that guarantee entropy contraction not for all starting distributions, but for those in a large neighborhood of the stationary distribution.
        
        We use our framework to show that we can sample from the hardcore and Ising models on $n$-node graphs that have a constant  $\delta$ relative gap to the tree-uniqueness threshold, in nearly-linear time $\widetilde O_{\delta}(n)$. Notably, our bound does not depend on the maximum degree $\Delta$ of the graph, and is therefore optimal even for high-degree graphs. Our work improves on prior mixing time bounds of $\widetilde O_{\delta, \Delta}(n)$ and $\widetilde O_{\delta}(n^2)$, established via (non-restricted) modified log-Sobolev and Poincare inequalities respectively. As an additional corollary of our results we show that optimal concentration inequalities can still be achieved from the restricted form of modified log-Sobolev inequalities. To establish restricted entropy contraction for these distributions, we extend the entropic independence framework of Anari, Jain, Koehler, Pham, and Vuong to distributions that satisfy spectral independence under a restricted set of external fields. We also develop an orthogonal trick that might be of independent interest: utilizing Bernoulli factories we show how to implement Glauber dynamics updates on high-degree graphs in $O(1)$ time, assuming the graph is represented so that one can sample random neighbors of any desired node in $O(1)$ time. Put together, we obtain the perhaps surprising result that we can sample from tree-unique hardcore and Ising models in time $\widetilde O_{\delta}(n)$, i.e., without even necessarily having enough time to read all edges of the graph.
    \end{abstract}
    \clearpage

\section{Introduction}

The Ising and hardcore models are two classes of probability distributions on graphs which have long been the subject of intense study in statistical physics, probability theory, and theoretical computer science. These models are some of the most fundamental examples of probability distributions with interacting and correlated random variables, and also serve as key examples in the study of critical phenomena in lattice and mean-field spin systems. Concretely, the Ising model on a graph $G=(V,E)$ is the probability measure over spins $\sigma \in \set{\pm 1}^V$ with probability mass function
\[ \mu(\sigma) \propto \lambda^{\card{\set{ i \given \sigma_i = +1 }}} \prod_{\set{i,j} \in E} \beta^{\1(\sigma_i = \sigma_j)}  \]
parameterized by \emph{external field} $\lambda$ and \emph{edge activity} $\beta$, where the parameter $\beta$ controls whether the spins like to align with ($\beta > 1$, \emph{ferromagnetic}) or opposite to ($\beta < 1$, \emph{antiferromagnetic}) their neighbors on the graph. Similarly, the hardcore model is the probability measure on \emph{independent sets} $\sigma \subset 2^{V}$ given by
\[ \mu(\sigma) \propto \lambda^{\card{\sigma}}, \]
i.e., weighted by the size of the independent set, where the parameter $\lambda$ is conventionally referred to as the \emph{fugacity}. The study of these two models has been closely linked; in some cases the hardcore model can even be recovered as a limit of the Ising model with strongly antiferromagnetic interactions ($\beta \to 0, \lambda \to 0$).

It has long been understood that there are close connections between several questions of interest concerning these models: in particular, between the \emph{uniqueness of Gibbs measures} on infinite graphs, rapid mixing of the \emph{Glauber dynamics} (or Gibbs sampler) for the corresponding model on finite graphs, and \emph{concentration of measure} estimates and \emph{functional inequalities} for the Gibbs measure; see e.g. \cite{dobrushin1968problem,stroock1992equivalence,dobrushin1985constructive,marton2004measure,zegarlinski1992dobrushin,wu2006poincare} for some examples. In many cases, there is a single ``high temperature'' regime for the parameters of the model under which all of these properties of the measure hold at once. In the case of antiferromagnetic models, the ``high temperature'' regime also characterizes the exact regime of parameters where these models are algorithmically tractable \cite{sly2012computational,galanis2016inapproximability}.

Recently, there has been an intense interest in understanding the sharp thresholds for mixing in the hardcore and Ising models on the class of graphs of maximum degree $\Delta$ based on connections to the \emph{uniqueness threshold} on the infinite $\Delta$-regular tree. After a long line of work including \cite{dyer2002counting,weitz2006counting,mossel2009hardness,sly2010computational,sly2012computational,mossel2013exact,sinclair2014approximation,galanis2016inapproximability,chen2021optimal,chen2021rapid,anari2021spectral} we know that in the particular case of the hardcore model, sampling is computationally hard above the uniqueness threshold on graphs of degree at most $\Delta$ \cite{sly2010computational}, and below the uniqueness threshold sampling can be done in $\tilde{O}(n^2)$ steps of the Glauber dynamics \cite{chen2021rapid} or $\widetilde O(C_{\Delta} n)$ steps if the maximum degree $\Delta$ is fixed \cite{chen2021optimal, blanca2021mixing}; $C_\Delta$ is an exponentially large function of the maximum degree $\Delta$. The picture is similar for the Ising model, with $C_\Delta$ being instead a polynomially large function of $\Delta$ whose exponent depends on the gap to the uniqueness threshold \cite{chen2021optimal,blanca2021mixing}. Nevertheless, it has been generally expected that the mixing time of Glauber dynamics is always $O(n \log n)$ within the uniqueness regime, regardless of the degree of the graph. \Textcite{chen2021optimal} raised the challenge of proving even a weaker bound with near-linear dependence on $n$ and polynomial dependence on $\Delta$, i.e., $\widetilde O(n \cdot \poly(\Delta))$, on the mixing time of Glauber dynamics for the tree-unique hardcore model. \Textcite{chen2021rapid} also raised the open problem of proving $\widetilde O(n)$ mixing time in the high-degree setting.

In this work, we develop new methods for analyzing the mixing time of Markov chains and as an application, finally prove that a slight variation of Glauber dynamics, which we dub balanced Glauber dynamics, mixes in the optimal $O(n \log n)$ many steps up to the appropriate uniqueness thresholds for both the hardcore and Ising models. Our results build upon tools from recent works  
which analyze Markov chains through high-dimensional expansion properties  \cite[e.g., ][]{anari2021spectral,chen2021optimal,chen2021rapid} and in particular, techniques for establishing \emph{entropic independence} introduced by \textcite{anari2021entropic}. To surpass the limitations of previous methods, we introduce a number of new ideas for proving and using functional inequalities restricted to certain classes of well-behaved functionals.
As a consequence of our theory, we not only recover sharp mixing time estimates but also prove strong concentration of measure estimates which generalize classical results from the product measure setting.

Going beyond mixing times, we show that the \emph{total running time} of our sampling algorithms for both the Ising model and the hardcore model can be bounded by $\widetilde O(n)$ as well, assuming suitable access to the underlying graph; remarkably, this is sublinear in the size of the \emph{input}, for not-too-sparse graphs, and is nearly-linear in the size of the output, which is optimal. We show this by introducing a technique which might be of independent interest: we employ Bernoulli factories to implement Glauber dynamics updates in amortized $O(1)$ time per update instead of the na\"ive $O(\Delta)$ time, under the standard assumption 
that the graph is given to the algorithm in a way that (uniformly) random neighbors of any desired node $v$ can be efficiently sampled.
\subsection{Our results}
In this paper, we focus on the behavior of the hardcore model and Ising model up to the tree uniqueness threshold. For the hardcore model on a graph of maximum degree $\Delta \ge 3$, the uniqueness threshold is explicitly given by the critical fugacity
\begin{equation*} 
\lambda_{\Delta} := \frac{(\Delta - 1)^{\Delta - 1}}{(\Delta - 2)^{\Delta}}
\end{equation*}
and we say the hardcore model with fugacity $\lambda$ is $\delta$-unique if $\lambda \le (1 - \delta) \lambda_{\Delta}$. Similarly, for the Ising model we say it is $\delta$-unique (with worst case external field) if 
\[ \beta \in \bracks*{\frac{\Delta - 2 + \delta}{\Delta - \delta}, \frac{\Delta - \delta}{\Delta - 2 + \delta}}. \]

\paragraph{Optimal mixing of the (balanced) Glauber dynamics.} 
We now state our main mixing time results. For sampling from the hardcore model, we consider a variant of Glauber dynamics which we call the \emph{balanced Glauber dynamics}. In this variation, the update site is chosen in a slightly non-uniform fashion, effectively introducing a small number of additional ``balancing'' updates into the usual Glauber chain --- see \cref{sec:balanced-gd} for a precise description of this process. We explain the motivation for this small modification in Techniques below.
\begin{theorem}[See \cref{thm:hardcore-detailed} for a more precise statement]
\label{thm:main hardcore}
Suppose $\mu$ is the $\delta$-unique hardcore model on $G = (V,E)$ with $\abs{V}= n$. The balanced Glauber dynamics with $O_{\delta}( n \log n)$ many steps approximately samples from $\mu$.
\end{theorem}
\begin{theorem} \label{thm:main ising}
Suppose $\mu$ is the $\delta$-unique Ising model on $G = (V,E)$ with $\abs{V}= n$. The (standard) Glauber dynamics with $O_{\delta}(n\log n)$ many steps approximately samples from $\mu$. 
\end{theorem}
\begin{remark}
The uniqueness region defined for the Ising model has a strange peculiarity: \cref{thm:main ising} in its exact form, that is nearly-linear mixing for all degrees, can actually be derived without appealing to the techniques of this work or even the prior work of \textcite{chen2021rapid} who studied spectral gap for high-degree regimes. The reason is that $\delta$-unique Ising models with high enough degree $\Delta>\Delta_0(\delta)$ satisfy the Dobrushin uniqueness condition! And complementing that, all small degrees $\Delta=O(1)$ are covered by the earlier work of \textcite{chen2021optimal}.

To see why high-degree cases fall under Dobrushin's regime, note that each entry of the Dobrushin influence matrix (see \cite{hayes2006simple} for definition) can be bounded by $\abs{\log(\beta)}/2$. Be aware that the notation $\beta$ is not consistent between our work and \cite{hayes2006simple}; one can translate $\beta \mapsto \log(\beta)/2$ to go from our notation to that of \cite{hayes2006simple}. This means the norm of the Dobrushin influence matrix is bounded by $\Delta\cdot \abs{\log(\beta)}/2$ which is asymptotically
\[ \frac{\Delta}{2}\cdot \parens*{\frac{2-\Theta(\delta)}{\Delta}+O(1/\Delta^2)}=1-\Theta(\delta)+O(1/\Delta). \] 
Fixing $\delta$, for large enough $\Delta$, this norm gets smaller than a constant $<1$, which entails Dobrushin uniqueness and hence nearly-linear mixing time \cite{hayes2006simple,levin2017markov}.

We remark that ``asymptotic Dorbushin uniqueness'' does not happen for the hardcore model or slight variants of the tree-unique Ising model (see \cref{rem:modified-ising} for details).
\end{remark}
\begin{remark}\label{rem:modified-ising}
Note that in the limit $\Delta \to \infty$, the endpoints of the uniqueness region for the Ising model are essentially of the form $\beta  = 1 \pm 2/\Delta$.
For antiferromagnetic Ising models, even if $\beta$ is outside of this region (e.g. $\beta \approx 1 - \alpha/\Delta$ for $\alpha > 2$), there is a critical external field $\lambda_c \in (0,1)$ below which the Ising model is in the tree uniqueness region \cite{sinclair2014approximation} (see also \cref{sec:ising-antiferro}): we also establish an analogous result covering this two-dimensional uniqueness region in $(\beta,\lambda)$ space (see \cref{thm:ising-anti-detailed}) using the balanced Glauber dynamics. The sublinear time sampling result (\cref{thm:general-ising-ss}) and concentration result (\cref{prop:ising-lipschitz}) below also extend to this setting, see \cref{prop:ising-anti-concentration}.
\end{remark}

\paragraph{Sublinear time sampling algorithms.}
Supposing that the graph $G$ is represented using the standard data structure of adjacency arrays, i.e., each vertex has an array of neighbors so that sampling a random neighbor of a vertex can be performed in $O(1)$ time. We show how to sample from both the hardcore model and the Ising model in runtime nearly linear in the output length $n$, and so in sublinear time for graphs of at least polylogarithmic average degree. 

\begin{theorem}[See \cref{thm:hardcore-systematic} for a more precise statement]
Suppose $\mu$ is the $\delta$-unique hardcore model on $G = (V,E)$ with $\abs{V}= n$, and $G$ is represented by adjacency arrays. Then, there is a randomized algorithm to approximately sample from $\mu$ which can be implemented in 
expected time $O_{\delta}(n \log^2(n))$. 
\end{theorem}

\begin{theorem}[See \cref{thm:general-ising-ss} for a more precise statement]
Suppose $\mu$ is the $\delta$-unique Ising model on $G = (V,E)$ with $\abs{V}= n$ and $G$ is represented by adjacency arrays. 
Then a step of the Glauber dynamics can be implemented by a randomized algorithm with expected running time $O(1)$. Combined with \cref{thm:main ising}, this implies that approximate sampling can be performed in expected runtime $O_{\delta}(n \log n)$.
\end{theorem}


\paragraph{Sharp concentration of measure and transport-entropy inequalities.}
Using our restricted modified log-Sobolev inequalities, we show via the Herbst argument that sub-Gaussian concentration bounds hold for all Lipschitz functions in both the hardcore and Ising model in the uniqueness region. By the celebrated result of \textcite{bobkov1999exponential}, concentration of Lipschitz functions is equivalent to a $W_1$ transport-entropy inequality, i.e. $W_1(\nu,\mu)^2 \le C \DKL{\nu \river \mu}$ for all measures $\nu$ where $W_1$ denotes the Wasserstein-1 distance with the Hamming metric. 
\begin{theorem}[See \cref{lem:hardcore-lipschitz}]
Suppose $\mu$ is the $\delta$-unique hardcore model on a graph with $n$ vertices, and let $f$ be so that $\abs{f(\sigma_+)-f(\sigma_-)} \le \kappa$ for all adjacent states $(\sigma_-,\sigma_+)$, i.e. $f$ is $\kappa$-Lipschitz with respect to the Hamming metric. For all $t \ge 0$, we have 
\[ \P_{\mu}{f -  \E_{\mu} f > t} \le 
e^{-c t^2/\kappa^2 n} \]
for some $c = c(\delta) > 0$.
\end{theorem}

\begin{remark}
In the hardcore model with small fugacity, sites are much more likely to be unoccupied than occupied, which can lead to even better concentration. To reflect this, we establish (see \cref{lem:hardcore-monotone}) a more precise two-level Bernstein-type inequality for monotone functionals, such as the number of occupied sites in the hardcore model.
\end{remark}

\begin{theorem}[See \cref{prop:ising-lipschitz}]
Suppose $\mu$ is the $\delta$-unique Ising model on a graph with $n$ vertices, and let $f$ be so that $\abs{f(\sigma_+)-f(\sigma_-)} \le \kappa$ for all adjacent states $(\sigma_-,\sigma_+)$, i.e. $f$ is $\kappa$-Lipschitz with respect to the Hamming metric. For all $t \ge 0$ we have 
\[ \P_{\mu}{f -  \E_{\mu} f > t} \le 
e^{-c t^2/\kappa^2 n} \]
for some $c = c(\delta) > 0$.
\end{theorem}

\subsection{Techniques}

A classic approach to the analysis of Markov chain mixing times consists of establishing functional inequalities, where roughly speaking, one shows that a measure of distance to the stationary measure $\mu$ multiplicatively contracts at every step. Two popular measures of distance to stationarity for a distribution $\nu$ are the $\chi^2$-divergence, a.k.a.\ the variance of $\nu$'s density w.r.t.\ $\mu$:
\[ \E*_{\mu}{\parens*{\frac{d\nu}{d\mu}-1}^2}, \]
and the relative entropy, a.k.a.\ the Kullback-Leibler divergence:
\[ \E*_{\mu}{\frac{d\nu}{d\mu}\cdot \log \frac{d\nu}{d\mu}}. \]
Contraction of these ``divergences'' are related to Poincare and modified log-Sobolev inequalities respectively \cite[see, e.g.,][]{bobkov2006modified}. Contraction of variance is often easier to establish, because of its relation to the spectral gap of the Markov chain which enables a host of techniques for spectral analysis, but often it leads to a suboptimal (with a polynomial factor loss) bound on the mixing time; in contrast, modified log-Sobolev inequalities are notoriously difficult to establish, especially since there is no equivalent spectral connection, but they can lead to optimal mixing time bounds.

It is well-known that entropy contraction is strictly stronger than variance contraction \cite{bobkov2006modified}. The reason for this is that for distributions $\nu$ that are infinitesimally close to $\mu$, entropy contraction and variance contraction become \emph{equivalent}. Roughly speaking, this is because the functional $x\mapsto x\log x$ can be approximated by its quadratic Taylor expansion near $x=1$, with the second degree term giving us the variance.

\paragraph{Restricted modified log-Sobolev inequalities.} Motivated by the observation that entropy contraction and variance contraction are equivalent in infinitesmially small neighborhoods of the stationary distribution, we propose studying an intermediate form of functional inequality that we call a \emph{restricted modified log-Sobolev inequality}. Roughly speaking, this is an inequality which guarantees entropy contraction in one step of the Markov chain for a \emph{restricted class} of distribtuions $\nu$. Intuitively, one should think of this as entropy contraction in a \emph{large neighborhood} of the stationary distribution. Our work shows that, in well-studied settings, restricted modified log-Sobolev inequalities can be considerably easier to establish than (full) modified log-Sobolev inequalities, while at the same time, yielding essentially the same consequences for mixing times and concentration of measure. 

\paragraph{Restricted entropic independence.} In order to establish restricted modified log-Sobolev inequalities we use a generalization of techniques developed by earlier work of \textcite{anari2021entropic}. Roughly speaking they showed that spectral independence \cite{ALO20}, a form of variance contraction, for not just the distribution $\mu$, but rather all external fields applied to $\mu$, automatically entails entropic independence, a form of entropy contraction. The main barrier in applying this framework to the distributions studied in this work, especially the hardcore model, is that arbitrary external fields can easily take us outside the uniqueness regime where there is no hope of mixing, let alone spectral independence; this is because an external field can change the parameter $\lambda$ (the fugacity) to an arbitrarily large positive number. Nevertheless, we employ the fact that a \emph{restricted} class of external fields keep the distribution in the spectral independence regime \cite{chen2021rapid}, and generalize the entropic independence machinery to show entropy contraction for a \emph{restricted} class of distributions $\nu$, which includes all of the distributions necessary for analyzing the mixing time of Markov chains and concentration of Lipschitz functions. 

\paragraph{Boosting contraction results using field dynamics.} We follow the footsteps of the prior work of \textcite{chen2021rapid} who invented a new Markov chain called field dynamics, and showed its utility in establishing a spectral gap, both for the field dynamics itself, and by a comparison argument, for the Glauber dynamics. Field dynamics allows one to combine a loose bound on variance contraction near the uniqueness threshold together with an optimal bound for variance contraction far away from the threshold, to get a boosted optimal bound on variance contraction near the threshold. Our arguments follow the same high-level plan but with variance replaced with entropy. That is, we establish restricted modified log-Sobolev inequalities for the field dynamics first, and use optimal entropy contraction inequalities far away from the uniqueness threshold, to get a boosted optimal entropy contraction near the threshold. We then use comparison arguments to translate the results to a variant of Glauber dynamics.

A challenging part of using restricted modified log-Sobolev inequalities to establish mixing times is that a priori there is no reason that the evolution of the Markov chain will keep the distribution in the restricted class where we have entropy contraction, even if we initially start from a distribution within this class. We show that in the case of tree-unique hardcore and Ising models, simple modifications of the well-studied Glauber dynamics Markov chain and the field dynamics guarantee that the distribution at time $t$ never escapes the restricted class of distributions.

\paragraph{Balanced Glauber dynamics, and field dynamics interleaved with systematic scans.} As noted above, in our analysis we consider a variant of Glauber dynamics. In this variation, the update site is chosen in a slightly non-uniform fashion, effectively introducing a small number of additional ``balancing'' updates into the usual Glauber chain. Similarly, for another Markov chain called field dynamics that was introduced by \textcite{chen2021rapid}, we sometimes add an additional interleaving systematic scan step to keep the distribution within the restricted region of entropy contraction.

The introduction of these additional steps is very analogous to the use of \emph{projections} in optimization algorithms such as projected gradient descent. In our case, these steps serve as projection operators in the following sense: they guarantee that the density of the resulting distribution lies in the class of $C$-completely bounded measures (see \cref{def:C-bounded}), where we have contraction of entropy, while ensuring that the projection itself does not increase the relative entropy. 
The projection step enables us to show that the Glauber/field dynamics step  makes a large amount of progress. In the optimization literature, such projection steps are sometimes crucial: the Iterative Hard Thresholding algorithm \cite{blumensath2009iterative} alternates between a projection onto the set of sparse vectors and a gradient step on the squared loss, where the sparsity generated by the projection step is needed to argue that the gradient step makes progress (enabling appeal to the ``Restricted Isometry Property''). Somewhat similarly, the Nash-Moser iteration (see, e.g., \cite{secchi2016nash}) combines the Newton step with a step which improves regularity.

We leave it as an interesting open question to investigate whether for the hardcore model, the balancing steps added to Glauber  dynamics are actually needed. Stated differently, does vanilla Glauber dynamics (potentially started from a judicious choice of starting point) automatically remain in the $C$-bounded region of entropy contraction?

\paragraph{Concentration inequalities.}
Modified log-Sobolev inequalities have other applications beyond mixing time of Markov chains; for example, they can be used to establish concentration inequalities using a technique known as the Herbst argument \cite[see, e.g.,][]{goel2004modified}. We show that for the Ising model and the hardcore model in the uniqueness region, restricted modified log-Sobolev inequalities are enough to establish the same optimal concentration inequalities (as would be obtained by conjectured modified log-Sobolev inequalities), by demonstrating that the Herbst argument essentially only needs entropy contraction for functionals within the ``good restricted class'' of $C$-bounded measures.

\paragraph{Sublinear time sampling algorithms.} Our results improve the mixing time bounds for the high-degree regime of the Ising and hardcore models. One concern might be that mixing time could be a misleading indicator of algorithmic tractability; after all, it is easy to construct Markov chains that mix in one step, but whose steps take exponential time to implement. This concern is moot for Glauber dynamics in bounded-degree graphical models, as the steps of Glauber dynamics can be easily implemented in constant time. We show that this concern is moot even for the high-degree regime, by introducing new tricks to implement Glauber dynamics updates of the tree-unique Ising and hardcore models in amortized $O(1)$ time per step, improving on the na\"ive implementation which takes $O(\Delta)$ time per update. For the Ising model, we assume the ability to sample uniformly random neighbors of any desired node in the graph, and show that a trick based on Bernoulli factories can achieve the desired $O(1)$ update time. As far as we know, this trick has not been studied before, and it might be of independent interest. 

\subsection{Further related work}

Previous work of \cite{anari2021entropic} introduced the notion of \emph{entropic independence} and new tools for proving entropic independence which we build upon in this work. In that work, one of the main results was also a $O(n \log n)$ time mixing estimate for the Ising model, established under an incomparable condition (bounded spectral norm of the interaction matrix, previously studied in \cite{bauerschmidt2019very,eldan2021spectral}).
In short, the bounded spectral norm condition is more powerful for studying frustrated models such as the Sherrington-Kirkpatrick model from spin glass theory \cite{mezard1987spin}, where the analogous tree threshold is the broadcasting/reconstruction threshold rather than the uniqueness threshold; in contrast, the results in this paper based on the uniqueness threshold are more powerful for studying the behavior of models with strong biases (such as the hardcore model) and for recovering the sharp uniqueness threshold constants in terms of the degree $\Delta$.
 In certain cases (e.g. the complete graph Curie-Weiss model) both of these results are applicable and we recover similar results using either approach.
 
 The hardcore model is one of the most well-studied distributions in statistical physics and more recently the area of sampling and counting in computer science, where the main challenge has been establishing near-linear time mixing for fugacity $\lambda$ all the way up to the uniqueness threshold. This was first achieved under structural assumptions on the graph: large degree and large girth \cite{hayes2006coupling}, amenable neighborhoods \cite{weitz2006counting}, and just large girth \cite{efthymiou2019convergence}. A sequence of results following the relatively new approach of high-dimensional-expanders-based Markov chain analysis, recently achieved the near-linear mixing time goal with no structural assumption other than $O(1)$-bounded-degrees \cite{ALO20, chen2020rapid, chen2021optimal, blanca2021mixing}. The remaining case of graphs with large maximum degree remained largely open prior to our work; a very recent work of \textcite{chen2021rapid} made substantial progress on this case by establishing an optimal spectral gap, and as a corollary, a (suboptimal) quadratic mixing time (we note that they established this optimal spectral gap for the more general class of all anti-ferromagnetic two-spin systems).
 
 \subsection{Organization}
 
 We first state our results in the general setting of two-spin models in \cref{sec:entropy-contraction-field,sec:comparison}. In \cref{sec:entropy-contraction-field}, we develop general techniques to translate the spectral independence of a distribution under a family of external fields to restricted entropic independence and restricted modified log-Sobolev inequalities for the field dynamics. Then in \cref{sec:comparison}, we show how to translate these notions from field dynamics to Glauber dynamics.
 
 In \cref{sec:sampling-hardcore,sec:ising} we specialize to the problem of sampling from the hardcore and Ising models respectively. In each section, we show how to modify field dynamics and Glauber dynamics to keep the evolving distribution over the run of the Markov chain within the well-behaved class of $C$-bounded distributions. In each section, we also show how to implement a sampling algorithm that has runtime near-linear in the output size.
 
 Finally, in \cref{sec:concentration}, we show how restricted modified log-Sobolev inequalities yield optimal concentration inequalities for both the hardcore and the Ising model.
 
 \subsection{Acknowledgements}
 Nima Anari and Thuy-Duong Vuong are supported by NSF CAREER award CCF-2045354, a Sloan Research Fellowship, and a Google Faculty Research Award. Huy Tuan Pham is supported by a Two Sigma Fellowship. Frederic Koehler was supported in part by E.\ Mossel's Vannevar Bush Fellowship ONR-N00014-20-1-2826.
 
\section{Preliminaries}

We let $[n]$ denote the set $\set{1,\dots, n}$ and $\N$ denote the set of natural numbers $\set{1,2,\dots}$.

\subsection{Distributions and generating polynomials}
We will frequently view distributions on $\set{\pm 1}^{n}$ as distributions on $2^{[n]}$ by identifying $\sigma \in \set{\pm 1}^{n}$ with the subset $\set{i \in [n]\given \sigma_i = 1}$. For a distribution $\mu$ on $2^{[n]}$ and a subset $\Lambda \subseteq [n]$, we say that $\sigma_{\Lambda} \in 2^{\Lambda}$ is a (valid) partial configuration if there exists some $S \in 2^{[n]\setminus \Lambda}$ with $\mu(S \cup \sigma_{\Lambda}) > 0$. For such a valid partial configuration, we define the \emph{$\sigma_{\Lambda}$-pinned} probability measure $\mu^{\sigma_{\Lambda}}$ on $2^{[n]\setminus \Lambda}$ by 
\[\mu^{\sigma_{\Lambda}}(T) \propto \mu(\sigma_{\Lambda} \cup T) \quad \forall T \in 2^{[n]\setminus \Lambda} \]
where the notation $\propto$ denotes equality up to a normalizing constant (so that the measure $\mu^{\sigma_{\Lambda}}$ is a valid probability measure). Also, for any subset $\Lambda \subseteq [n]$, we define 
\[\mu(\Lambda) = \sum_{S \supseteq \Lambda}\mu(S).\]
For notational convenience, we will denote $\mu(\set{i})$ simply as $\mu(i)$. 

\begin{definition}
The multivariate generating polynomial $g_{\mu} \in \R[z_1,\dots, z_n]$ associated to a density $\mu\colon 2^{[n]} \to \R_{\geq 0}$ is given by
\[g_{\mu}(z_1,\dots, z_n) := \sum_{S}\mu(S)\prod_{i\in S}z_i = \sum_{S}\mu(S)z^{S}.\]
\end{definition}
Here we have used the standard notation that for $S \subseteq [n]$, $z^S = \prod_{i\in S}z_i$; note that this is a multilinear polynomial.

\begin{definition}[Measure tilted by external field]\label{def:measure-tilt}
For a distribution $\mu$ on $2^{[n]}$ and vector $\lambda = (\lambda_1,\dots, \lambda_n) \in \R^{n}_{>0}$, which we refer to as the \emph{external field}, we denote the measure $\mu$ tilted by external field $\lambda$ by the notation $\lambda \ast \mu$, formally defined as
\[\P_{\lambda \ast \mu}{S} = \frac{1}{Z_{\lambda}} \mu(S)\cdot \prod_{i \in S}\lambda_i \]
where the normalizing constant $Z_{\lambda}$ is defined so that $\lambda \ast \mu$ is a probability measure.
Note that for any $(z_1,\dots, z_n) \in \R^{n}_{\geq 0}$, 
\[g_{\lambda \ast \mu}(z_1,\dots, z_n) \propto g_{\mu}(\lambda_1 z_1,\dots, \lambda_n z_n).\]
We also use the following shorthand: for $\lambda \in \R_{> 0}$, the notation $\lambda \ast \mu: = (\lambda, \ldots, \lambda) \ast \mu$ denotes the measure $\mu$ tilted by uniform external field $\lambda$.
\end{definition}

Sometimes we use the notation $\bar{x}$ to denote a natural involution applied to some object $x$. Suppose we have a set of indices $\Omega = \set{i_1,\dots, i_n}$, and define the set $\bar{\Omega} = \set{\bar{i_1},\dots, \bar{i_n}}$, which is disjoint from $\Omega$, and each of whose elements is naturally paired with an element of $\Omega$. Note that since we assumed the sets are disjoint, we can write $\Omega \cup \bar{\Omega} = \set{i_1, \bar i_1, \ldots,i_n,\bar i_n }$. This lets us represent a vector $\sigma \in \set{\pm 1}^{\Omega}$ as a subset $\sigma^{\hom}$ of $\Omega \cup \bar{\Omega}$ in the obvious way, which we define formally below. 

\begin{definition}[Homogenization]
For $\sigma \in \set*{\pm 1}^{\Omega},$ let $\sigma^{\hom} \in \binom{\Omega \cup \bar{\Omega} }{\card{\Omega}}$ be the set $$\set*{i \in \Omega \given \sigma_i = 1} \cup\set*{\bar{i} \in \bar{\Omega} \given \sigma_i = -1}.$$ For a distribution $\mu$ over $\set*{\pm 1}^{\Omega}$, let the {homogenization} of $\mu$, denoted by $\mu^{\hom},$ be the distribution supported on $\set*{\sigma^{\hom} \given \sigma \in \set*{\pm 1}^{\Omega}} $ defined by $\mu^{\hom}(\sigma^{\hom}) \propto \mu(\sigma). $ The completely analogous definition is also made for $\sigma \in 2^{\Omega}$. 
\end{definition}
In terms of generating polynomials, we can equivalently write
\begin{equation}\label{eqn:generating-polynomial-hom}
g_{\mu^{\hom}}(z_1,\ldots,z_n,z_{\bar 1},\ldots z_{\bar n}) 
= g_{\mu}(z_1/z_{\bar 1},\ldots,z_n/z_{\bar n}) z_{\bar 1} \cdots z_{\bar n},
\end{equation}
and we note that the generating polynomial $g_{\mu^{\hom}}$ is a homogeneous polynomial of degree $n$. 

Finally, we introduce the notion of a blow-up of a probability distribution \cite{chen2021rapid}, which plays an important role in analyzing the field dynamics (see \cref{def:projected block dynamics} and \cref{lem:approx-field-block} later).

\begin{definition}[$\vec k$-blow-up distribution]
For a distribution $\mu: 2^{[n]}\to \R_{\geq 0}$ and $\vec{k}\in \N^{n},$ let $\mu_{\vec{k}}$ be the distribution on $2^{[k_1 + \dots + k_n]}$ with multivariate generating polynomial
\[g_{\mu}\parens*{\frac{z_{1, 1} +\dots z_{1,k_1} }{k_1}, \dots, \frac{z_{n, 1} +\dots z_{n,k_n} }{k_n}}.\]
Note that this polynomial is multilinear with nonnegative coefficients summing to one, so this definition indeed describes a probability measure. 
\end{definition}

We remark that the blow-up of a homogeneous $\mu$ (only supported on sets of a particular size) has also been called subdivision \cite{anari2020isotropy}, but we use the term blow-up instead to emphasize the asymmetric nature of the operation in the absence of homogeneity.

We will identify $[k_1 + \dots + k_n]$ with 
the set of pairs $\set{(i,j) \given i \in [n], j \in [k_i]}$ and so
denote the elements of $[k_1+\dots +k_n]$ by pairs $(i,j)$.
Accordingly, for any such $(i,j)$, $\mu_{\vec{k}}(\cdot \mid (i,j))$ (respectively, $\mu_{\vec{k}}(\cdot \mid \overline{(i,j)})$) denotes the distribution of $\mu_{\vec{k}}$ conditioned on $(i,j)$ being included (respectively, excluded).

There is an equivalent description of $\mu_{\vec{k}}$ which is, perhaps, more intuitive. With notation as in the previous paragraph, $\mu_{\vec{k}}$ is the distribution of a random variable $Y$, taking values in $2^{[k_1 + \dots + k_n]}$ constructed as follows: first, sample $X \sim \mu$. If $i \notin X$, then $(i,j)\notin Y$ for all $j \in [k_i]$. For each $i \in X$, sample $j_i^* \in [k_i]$ uniformly and independent of all other choices. Let $(i,j_i^*) \in Y$ and $(i,j) \notin Y$ for all $j \neq j_i^*$. In particular, note that under the blow-up-distribution, two distinct elements $(i,j)$ and $(i,j')$ for $j \ne j'$ will never simultaneously be present in a sample $\sigma_{\vec k} \sim \mu_{\vec k}$. 

There is an obvious projection map $P_{\vec k}(i,j) = i$ which goes from sites for the blow-up distribution (i.e. elements of $[k_1 + \dots + k_n]$) to sites for the original distribution (elements of $[n]$). Similarly, we can define a natural projection from configurations in the blow-up distribution to configurations on the original space, as follows. 
\begin{definition}[Projection]
\label{def:projection}
Let $\mu$ be a distribution over $2^{[n]}$, let $\vec{k} \in \N^{n}$, and let $\mu_{\vec{k}}$ denote the corresponding blow-up distribution on $2^{[k_1 + \dots + k_n]}$. 
For any 
$\sigma \in 2^{[k_1 + \cdots + k_n]}$, its projection $\sigma^\ast \in 2^{[n]}$ is defined by $i \in \sigma^\ast$ iff there exists some $j \in [k_i]$ such that $(i, j) \in \sigma$.
\end{definition}

The following lemma shows that the conditionals of $\mu_{\vec{k}}$ are blow-ups of conditionals of the tilted measure $\lambda \ast \mu$ for some external field $\lambda \in [0,1]^{n}$. 

\begin{lemma} \label{lem:conditional of blow up}
Let $\mu$ be a distribution on $2^{[n]}$, let $\vec{k} \in \N^{n}$, and let $\mu_{\vec{k}}$ denote the blow-up distribution. 
For $i \in [n]$ and $j \in [k_i]$,
\begin{equation*}
    \begin{split}
        \mu_{\vec{k}} (\cdot  \mid (i,j) ) &\equiv \mu(\cdot \mid i)_{k_1, \dots, k_{i-1}, k_{i+1}, \dots, k_n}.\\
        \mu_{\vec{k}} (\cdot  \mid \overline{(i,j)} ) &\equiv (\lambda \ast \mu)_{k_1, \dots, k_{i-1}, k_i -1, k_{i+1}, \dots, k_n} \text{ if } k_i > 1, \text{ with } \lambda_i = \frac{k_i - 1}{k_i} \text{ and } \lambda_j = 1 \text{ for } j \neq i.\\
         \mu_{\vec{k}} (\cdot  \mid \overline{(i,j)} ) &\equiv  \mu(\cdot \mid \bar{i})_{k_1, \dots, k_{i-1}, k_{i+1}, \dots, k_n} \text{ if } k_i = 1.
    \end{split}
\end{equation*}
\end{lemma}
\begin{proof}
We prove only the second part; the proofs of the other parts are similar (and easier). We identify $[k_1 + \dots + (k_i-1) + \dots + k_n]$ with $\set{(i,j)\given i\in [n], j \in [k_i]} \setminus \set{(i,j)}$. Then, for any $S' \in 2^{[k_1+\dots + k_n]}$ such that $(i,j)\notin S'$, we have that
\[\mu_{\vec{k}}(\cdot \mid \overline{(i,j)}) \propto \mu(S)\prod_{\ell\in S}k_{\ell}^{-1},\]
where $S = \set{i \in [n]\given (i,j)\in S' \text{ for some }j\in [k_i]}$. Moreover,
\begin{align*}
    (\lambda \ast \mu)_{k_1,\dots, k_{i-1}, k_i-1,k_{i+1},\dots, k_n}(S') 
    &\propto \mu(S)\prod_{\ell \in S\setminus i}k_{\ell}^{-1}\cdot (\mathds{1}_{i\notin S} + \mathds{1}_{i\in S}\lambda (k_i-1)^{-1})\\
    &= \mu(S)\prod_{\ell \in S}k_{\ell}^{-1},
\end{align*}
which gives the desired conclusion. 
\end{proof}

\subsection{Restricted fractional log-concavity and spectral domination}

\begin{definition}[Restricted fractional log-concavity]
	
	For $\alpha \in (0,1]$ and for a convex subset $\mathcal{R} \subseteq \R^{n}_{\geq 0}$, a distribution $\mu$ on $\binom{[n]}{\ell}$ is said to be $\alpha$-fractionally log-concave on $\mathcal{R}$ (abbreviated as $\alpha$-FLC on $\mathcal{R}$) if $\log g_{\mu}(z_1^{\alpha},\dots, z_{n}^{\alpha})$ is concave, viewed as a function on $\mathcal{R}$. 
\end{definition}

We note that when $\mathcal{R} = \R^{n}_{\geq 0}$, then this is simply the notion of $\alpha$-fractional-log-concavity introduced in \cite{alimohammadi2021fractionally}. We also note that $1$-FLC on $\R^{n}_{\geq 0}$ is equivalent to complete/strong log-concavity \cite{ALOV19,BH19}.

\begin{definition}[Correlation matrix] \label{def:corr}
	Let $\mu$ be a probability distribution over $2^{[n]}$.  
Its \textit{correlation matrix} $\corMat_{\mu} \in \R^{n\times n}$  is defined by
\[\corMat_{\mu} (i,j) = \begin{cases} 1 -\P{i} &\text{ if } j=i, \\ \P{j \given i} - \P{j} &\text{ otherwise.}\end{cases} \]
\end{definition}

\begin{definition}[Spectral domination]
For $\eta, \epsilon \geq 0$, a distribution 
$\mu: 2^{[n]} \to \R_{\geq 0} $ is said to be $(\eta,\epsilon)$-spectrally dominated if
\[\lambda_{\max}(\corMat_{\lambda \ast \mu}) \leq \eta \quad \forall \lambda \in (0,1+\epsilon]^{n}.\] 
\end{definition}


$(\eta,\epsilon)$-spectral domination, as defined above, implies fractional log-concavity of the homogenization over a certain convex cone in the positive orthant. As a special case, if we take $\epsilon = \infty$ then we recover fractional log-concavity over the whole positive orthant, as studied in \cite{alimohammadi2021fractionally}.
\begin{proposition}
\label{prop:spectralind-to-flc}
Let $\eta > 1/2$ and $\epsilon \geq 0$. If the distribution
$\mu: 2^{[n]} \to \R_{\geq 0}$ is $(\eta,\epsilon)$-spectrally dominated, 
then for all $0< \alpha \leq 1/2\eta$, its homogenization $\mu^{\hom}$ is $\alpha$-FLC on the region 
\begin{equation*} 
\Lambda_{\alpha, \epsilon} := \set*{(z_1,\dots,z_n, z_{\bar{1}},\dots, z_{\bar{n}}): 0\leq z_i \leq z_{\bar{i}}(1+\epsilon)^{1/\alpha} \quad\forall i \in [n]}. 
\end{equation*}
\end{proposition}
\begin{proof}
The proof closely follows the proofs of \cite[Lemmas~69 and 71]{alimohammadi2021fractionally}. Let  $0 < \alpha \leq 1/2\eta$, $f := g_{\mu^{\hom}}$, and $\vec{v} = (v_1,\dots, v_n, v_{\bar{1}},\dots, v_{\bar{n}}) \in \Lambda_{\alpha, \epsilon}$. We wish to show that the Hessian of $\log f (\set{z_i^\alpha, z_{\bar{i}}^\alpha})$ is negative semi-definite at $\vec{v}$. Let $\vec{1}$ denote the all-ones vector and let $D_{\vec{v}}$ denote the $2n\times 2n$ diagonal matrix whose $j^{th}$ entry, for $j \in [n]\cup \overline{[n]}$, is $v_j^{-1}$. Then, by homogeneity of $f$ we have the identity 
\[\eval{\nabla^2 \log f (\set{z_i^\alpha, z_{\bar{i}}^\alpha})}_{\vec{v}} = D_{\vec{v}} \parens*{\eval{\nabla^2 \log f^{\vec{v}} (\set{z_i^\alpha, z_{\bar{i}}^\alpha})}_{\vec{1}}} D_{\vec{v}},\]
where $f^{\vec{v}}(z_i, z_{\bar{i}}) = f(v_i^\alpha z_i, v_{\bar{i}}^\alpha z_{\bar{i}})$, so it suffices to show that $\eval{\nabla^2 \log f^{\vec{v}} (\set{z_i^\alpha, z_{\bar{i}}^\alpha})}_{\vec{1}}$ is negative semi-definite. Moreover, since
\[f^{\vec{v}} \propto g_{\vec{v}^{\alpha}\ast \mu^{\hom}},\]
it suffices to show that $H := \eval{\nabla^2 \log g_{\vec{v}^{\alpha}\ast \mu^{\hom}} (\set{z_i^\alpha, z_{\bar{i}}^\alpha})}_{\vec{1}}$ is negative semi-definite. From the proof of \cite[Lemma~69]{alimohammadi2021fractionally}, this holds provided that $\lambda_{\max}(\corMat_{\vec{v}^{\alpha}\ast \mu^{\hom}}) \leq 1/\alpha$. Finally, from the  proof of \cite[Lemma~71]{alimohammadi2021fractionally}, the  $(\eta,\epsilon)$-spectral domination of $\mu$ implies that $\lambda_{\max}(\corMat_{\vec{v}^{\alpha}\ast \mu^{\hom}}) \leq 2\eta \leq 1/\alpha$, since from \cref{eqn:generating-polynomial-hom} we have
\[ \vec{v}^{\alpha}\ast \mu^{\hom} \propto (\vec{\lambda} \ast \mu)^{\hom} \]
for $\vec \lambda = (\lambda_1, \ldots, \lambda_n)$ with
$\lambda_i := {v_i^{\alpha}}/{v_{\bar{i}}^{\alpha}} \leq (1+\epsilon)$. 
\end{proof}

\subsection{Markov chains and functional inequalities}\label{sec:prelim-markov}
Let $\mu$ and $\nu$ be probability measures on a finite set $\Omega$. The Kullback-Leibler divergence (or relative entropy) between $\nu$ and $\mu$ is given by
\[\DKL{\nu \river \mu} = \sum_{x \in \Omega}\nu(x)\log\parens*{\frac{\nu(x)}{\mu(x)}},\]
with the convention that this is $\infty$ if $\nu$ is not absolutely continuous with respect to $\mu$. By Jensen's inequality, $\DKL{\nu \river \mu} \geq 0$ for any probability measures $\mu, \nu$. For later use, we record the following simple lemma. 

\begin{lemma}
\label{lem:blow up preserves KL}
Let $\mu, \nu$ be probability distributions on $2^{[n]}$, let $\vec{k} \in \N^{n}$, and let $\mu_{\vec{k}}, \nu_{\vec{k}}$ denote the corresponding blow-up distributions. Then, 
\[\DKL{\nu \river \mu} = \DKL{\nu_{\vec{k}} \river \mu_{\vec{k}}}.\]
\end{lemma}
\begin{proof}
For a $\nu$-valid configuration $\sigma \in 2^{[n]}$, let $\Sigma_{\sigma} \in 2^{[k_1 + \dots + k_n]}$ denote the collection of $\nu_{\vec{k}}$-valid configurations corresponding to $\sigma$. It is readily seen that
\begin{align*}
\sum_{\sigma' \in \Sigma_{\sigma}}\nu_{\vec{k}}(\sigma')\log\parens*{\frac{\nu_{\vec{k}}(\sigma')}{\mu_{\vec{k}}(\sigma')}} 
= \sum_{\sigma' \in \Sigma_{\sigma}}\nu_{\vec{k}}(\sigma')\log\parens*{\frac{\nu(\sigma)}{\mu(\sigma)}} = \nu(\sigma)\log\parens*{\frac{\nu(\sigma)}{\mu(\sigma)}}.  
\end{align*}
Summing over $\sigma$ gives the desired conclusion. 
\end{proof}
For a probability measure $\mu$ and a nonnegative random variable $f$, we define the (relative) entropy functional
\[ \Ent_{\mu}[f] = \E_{\mu}{f \log f} - \E_{\mu}{f} \log \E_{\mu}{f}, \]
noting that for a probability measure $\nu$, $\Ent_{\mu}\bracks*{\frac{d\nu}{d\mu}} = \DKL{\nu \river \mu}$. 

The total variation distance between $\mu$ and $\nu$ is given by
\[\DTV{\mu, \nu} = \frac{1}{2}\sum_{x \in \Omega}\abs{\mu(x) - \nu(x)}.\]

A Markov chain (see \cite{levin2017markov} for a detailed introduction) on $\Omega$ is specified by a row-stochastic non-negative transition matrix $P \in \R^{\Omega \times \Omega}$. We will view probability distributions on $\Omega$ as row vectors. A transition matrix $P$ is said to be reversible with respect to a distribution $\mu$ if for all $x,y \in \Omega$, $\mu(x)P(x,y) = \mu(y)P(y,x)$. In this case, it follows immediately that $\mu$ is a stationary distribution for $P$, i.e., $\mu P = \mu$. If $P$ is further assumed to be ergodic, then $\mu$ is its unique stationary distribution, and for any probability distribution $\nu$ on $\Omega$, $\DTV{\nu P^{t}, \mu} \to 0$ as $t \to \infty$. 

\begin{definition}[Mixing time]
Let $P$ be an ergodic Markov chain on a finite state space $\Omega$ and let $\mu$ denote its (unique) stationary distribution. For any probability distribution $\nu$ on $\Omega$ and $\epsilon \in (0,1)$, we define
\[t_\mix(P, \nu, \epsilon) = \min\set{t\geq 0 \given \DTV{\nu P^{t}, \mu}\leq \epsilon},\]
and
\[t_\mix(P,\epsilon) = \max\set*{t_\mix(P, \1_x, \epsilon)\given x\in \Omega},\]
where $\1_{x}$ is the point mass distribution supported on $x$. 

By replacing $\DTV{\cdot, \mu}$ with $\DKL{\cdot \river \mu}$ in the above definitions, we can analogously define $t^{\KL}_\mix(P,\nu, \epsilon)$ and $t_{\mix}^{\KL}(P,\epsilon)$. 
\end{definition}

We will drop $P$ and $\nu$ if they are clear from context. Moreover, if we do not specify $\epsilon$, then it is set to $1/4$. This is because the growth of $t_{\operatorname{mix}}(P,\epsilon)$ is at most logarithmic in $1/\epsilon$ (see \cite{levin2017markov}). Pinsker's inequality allows us to translate mixing in KL-divergence to mixing in total variation distance, i.e.,
\begin{equation*}
    t_\mix(P,\nu, \epsilon) \leq t_\mix^{\KL}(P, \nu, 2\epsilon^2).
\end{equation*}

We use the following standard notation for the \emph{Dirichlet form} of a Markov chain $P$ with stationary measure $\mu$:
\[
\mathcal{E}_P(f,g)=\angles{f, (I-P)g}_{L_2(\mu)} = \frac{1}{2} \sum_{x,y} \mu(x)P(x,y)(f(x)-f(y))(g(x)-g(y)).
\]

\begin{definition}[Restricted entropy contraction]
Let $P$ denote the transition matrix of an ergodic, reversible Markov chain on $\Omega$ with stationary distribution $\mu$. Given $\mathcal{V}$, a family of probability distributions on $\Omega$, the restricted entropy contraction constant of $P$ (with respect to $\mathcal{V}$) is defined to be
\[\alpha_0(P, \mathcal{V}) = \sup\set*{\alpha \in \R\given \DKL{\nu P\river \mu P} \leq (1-\alpha)\DKL{\nu\river \mu}\text{ for all }\nu \in \mathcal{V}}.\]

When $\mathcal{V}$ is the family of all probability distributions on $\Omega$, we will refer to $\alpha_0(P, \mathcal{V})$ as the entropy contraction constant of $P$ and denote it simply by $\alpha_0(P)$. 
\end{definition}
Note that by the data processing inequality, $\alpha_0(P, \mathcal{V}) \geq 0$. Moreover, the following relationship between the modified log-Sobolev constant of $P$ (see \cite{bobkov2006modified} for an introduction), defined by 
\[\rho_0(P) := \inf\set*{\frac{\mathcal{E}_P(f, \log f)}{2\Ent_{\mu}[f]}: f: \Omega \to \R_{>0}, \Ent_{\mu}[f] > 0} ,\] 
and the entropy contraction constant is well-known (see, e.g., \cite{anari2021entropic}):
\begin{equation*}
    \rho_0(P) \geq 2\alpha_0(P).
\end{equation*}

\begin{lemma}\label{lem:MLSI to mixing}
Let $P$ denote the transition matrix of an ergodic, reversible Markov chain on $\Omega$ with stationary distribution $\mu$. Let $\mathcal{V}$ denote a family of probability distributions on $\Omega$ which is closed under $P$. Then, for any $\nu \in \mathcal{V}$ and for any $\epsilon \in (0,1)$,
\[t_\mix^{\KL}(P,\nu, \epsilon) \leq  \ceil*{\alpha_0(P, \mathcal{V})^{-1}\cdot \parens*{\log\log\max\set*{\parens*{\frac{\nu(x)}{\mu(x)}}\given x\in \Omega} + \log\parens*{\frac{1}{\epsilon}}} }.\]
\end{lemma}

\begin{proof}
By the definition of the restricted entropy contraction constant and the closure of $\mathcal{V}$ under $P$, it follows that for any non-negative integer $t\geq 0$,
\begin{align*}
    \DKL{\nu P^{t+1}\river \mu P^{t+1}} = \DKL{(\nu P^{t})P\river \mu P} \leq (1-\alpha_0(P,\mathcal{V}))\DKL{\nu P^{t}\river \mu},
\end{align*}
from which it follows immediately that for all non-negative integers $t$,
\[\DKL{\nu P^{t}\river \mu} \leq (1-\alpha_0(P,\mathcal{V}))^t \DKL{\nu\river \mu}. \]
The claimed bound now follows by using the numerical inequality $1-x \leq e^{-x}$ and noting that $\DKL{\nu\river \mu} \leq \log\parens*{\max_{x\in \Omega}\frac{\nu(x)}{\mu(x)}}$.  
\end{proof}
We will also need the notion of \emph{approximate tensorization of entropy} as defined in, e.g., 
\cite{caputo2015approximate}. 
\begin{definition}[Approximate tensorization of entropy]
\label{def:approx-tensorization}
Let $\mu$ be a probability measure on the product space $\Omega = \Omega_1 \times \cdots \times \Omega_n$. We say that $\mu$ satisfies approximate tensorization of entropy with constant $C$ 
if for any positive measurable function $f$,
\[ \Ent_{\mu}(f) \le C \sum_{v = 1}^n \E_{\mu}{\Ent_{v}(f)} \]
where 
\[ \Ent_k(f) := \Ent_{\mu(\sigma_k = \cdot \mid \sigma_{\sim k})}(f) \]
is the entropy functional with respect to the conditional measure of $\sigma_k \in \Omega_k$ given $\sigma_j \in \Omega_j$ for all $j \ne k$ and for $\sigma \sim \mu$. We also say the pair $(\mu,f)$ satisfy approximate tensorization of entropy if the inequality holds for this particular function $f$. 
\end{definition}

It is known that if $\mu$ satisfies approximate tensorization of entropy with constant $C$, then $\mu$ satisfies the modified log-Sobolev inequality with constant $1/2Cn$ (see, e.g. \cite[Fact~3.5]{chen2021optimal}), where the factor of $2$ 
comes from the different conventions we are using for $\rho_0(P)$. We remark that approximate tensorization of entropy is equivalent to entropy contraction of the down operator (see \cref{def:down-operator}), the latter notion being defined for more general distributions beyond those supported on product spaces \cite{anari2021entropic}. 



We introduce some more notation for later use: in the setting of \cref{def:approx-tensorization}, for $S \subset [n]$, we will use the notation $\Ent_S(f)$ for the entropy functional with respect to the conditional measure of $\sigma_S$ given $\sigma_{[n] \setminus S}$ (in particular, $\Ent_S(f)$ is a function of $\sigma_{[n]\setminus S})$; this is useful when considering ``block versions'' of approximate tensorization.


\subsection{Complete spectral independence}

We record a variant of the definition of complete spectral independence that was instrumental in obtaining Poincare inequalities in the prior work of \textcite{chen2021rapid}.

\begin{definition}
For $\eta, \epsilon \geq 0$, a distribution $\mu: 2^{[n]} \to \R_{\geq 0}$ is said to be completely $(\eta, \epsilon)$-spectrally independent if for all $\Lambda \subseteq [n]$ with $\card{\Lambda} \leq n-2$, and for all valid partial configurations $\sigma_{\Lambda} \in 2^{\Lambda}$, the conditional distribution $\mu^{\sigma_{\Lambda}}$ on $2^{[n]\setminus \Lambda}$ is $(\eta, \epsilon)$-spectrally dominated. 
\end{definition}

\begin{proposition}
\label{prop:blowup preserve complete spectral ind}
For $\eta\geq 1$ and $\epsilon \geq 0$, if the distribution $\mu: 2^{[n]} \to \R_{\geq 0}$ is completely $(\eta,\epsilon)$-spectrally independent, then so is the distribution $\mu_{\vec{k}} : 2^{[k_1 + \dots + k_n]} \to \R_{\geq 0}$ for any $\vec{k} \in \N^n.$
\end{proposition}
\begin{proof}
Let $\mu' : = \mu_{\vec{k}}$ and let $\lambda' \in \R_{\geq 0}^{k_1+\dots+k_n}$. 
We will first show that \begin{equation} \label{ineq:blowup eigenvalues}
\lambda_{\max} \parens{\corMat_{\lambda' \ast \mu'}} \leq  \max \parens*{1, \lambda_{\max} \parens{\corMat_{\lambda \ast \mu}}}
\end{equation} where $ \lambda_i = \frac{1}{k_i}\sum_{j=1}^{k_i} \lambda'_{i,j}.$ In particular, for $\lambda' \in (0,1+\epsilon]^{\sum_{i=1}^n k_i},$ we have that $\lambda \in (0,1+\epsilon]^n,$ so that by the $(\eta, \epsilon)$-spectral domination of $\mu$ and the assumption that $\eta \geq 1$, we see that   $\lambda_{\max} \parens{\corMat_{\lambda' \ast \mu'}} \leq \eta$ for all such $\lambda'$. 

Since $\mu_{k_1,\dots, k_n} = (\mu_{k_1,1,\dots,1})_{1,k_2,\dots, k_n}$ and so on, it suffices to prove \cref{ineq:blowup eigenvalues} for $k_2=\dots = k_n = 1$. Let $\Psi$ and $\Psi'$ be the correlation matrix of $\lambda \ast \mu$ and $\lambda' \ast \mu'.$ 
Note that $\forall i\neq 1$ and distinct $j_1, j_2 \in [k_1]:$
\begin{align*}
    &\P_{\lambda' \ast \mu'}{(1,j_1)} = \frac{\lambda'_{1,j_1}}{k_1 \lambda_1}\P_{\lambda\ast\mu}{1}    \text{ and } \P_{\lambda' \ast \mu'}{(1,j_1) \given (1,j_2)} = 0;\\
    & \P_{\lambda' \ast \mu'}{(1,j_1) \given i} = 
    \frac{\lambda'_{1,j_1}}{k_1 \lambda_1 }\P_{\lambda \ast \mu} { 1 \given i} \text{ and }
    \P_{\lambda' \ast \mu'}{i } = \P_{\lambda \ast \mu}{i} \text{ and } \P_{\lambda' \ast \mu'}{i \given (1,j_1)} = \P_{\lambda\ast \mu}{i \given 1}. 
\end{align*}
By definition, $\Psi = \diag^{-1}\parens*{\sqrt{\P_{\lambda \ast \mu}{i}}} \cdot A\cdot  \diag\parens*{\sqrt{\P_{\lambda \ast \mu}{i}}}$ for a symmetric matrix $A$. 
Thus, $\Psi$ has an orthogonal basis  with respect to the inner product $\langle \cdot, \diag(\P_{\lambda \ast 
\mu}{i}) \cdot \rangle$ of eigenvectors $\vec{v}^1, \dots, \vec{v}^n$ 
with corresponding eigenvalues $\rho_1 \geq \rho_2 \geq \dots \geq \rho_n.$ 
It is straightforward to check that for each such $\vec{v}^t = (v^t_1,\dots, v^t_n)$, the vector \[\vec{w}^t : =(\underbrace{v^t_1, \dots, v^t_1}_{k_1 \text{ times}}, v^t_2, \dots, v^t_n)\]
is an eigenvector of $\Psi'$ with eigenvalue $\rho_t.$ Indeed, for any $j \in [k_1]$ and $i \geq 2,$ we have:
\begin{align*}
  \Psi'[(1,j), \cdot] \vec{w}^t &= v_1^t - \sum_{j'=1}^{k_1} \P_{\lambda' \ast \mu'}{(1,j')}  v_1^t + \sum_{i\neq 1} \Psi'_{(1,j),i} v^t_i = v_1^t \parens*{1-  \frac{\sum_{j'}\lambda'_{1,j'}}{k_1 \lambda_1}\P_{\lambda\ast\mu}{1}} +  \sum_{i\neq 1} \Psi_{1,i} v^t_i \\
  &= \sum_{i=1}^n \Psi_{1,i} v^t_i = \rho_t v^t_1;\\
  \Psi'[i, \cdot] \vec{w}^t &= \sum_{j'=1}^{k_1} \Psi'_{i,(1,j')}  v_1^t + \sum_{i'\neq 1} \Psi'_{i,i
 '} v^t_{i'} = v_1^t  \frac{\sum_{j'}\lambda'_{1,j'}}{k_1 \lambda_1}( \P_{\lambda\ast\mu}{1 \given i} - \P_{\lambda\ast\mu}{1}) +  \sum_{i'\neq 1} \Psi_{i,i'} v^t_{i'} \\
  &= \sum_{i'=1}^n \Psi_{i,i'} v^t_{i'} = \rho_t v^t_i.\\
\end{align*}
In addition, it is easily checked that $\Psi'$ has eigenvectors $\vec{w}^{n+t}:= (u^t_{1}, \dots, u^t_{k_1},\underbrace{0, \dots, 0}_{n-1 \text{ times}} ) $ with eigenvalue $ 1,$ where $\set*{\vec{u}^t}_{t=1}^{k_1-1}$ forms an orthogonal basis with respect to the inner product $\langle \cdot, \diag(\lambda'_{1,j})_{j=1}^{k_1} \cdot \rangle$ of the vector space $\set*{\vec{u} \in \R^{k_1} \given\vec{u} \perp (\lambda'_{1,j})_{j=1}^{k_1}}.$  Finally, observe that $\set*{\vec{w}^t}_{t=1}^{n+k_1-1}$ forms an orthogonal basis of eigenvectors of $\Psi'$ with respect to the inner product $\langle \cdot, \diag(\P_{\lambda \ast \mu}{i}\frac{\lambda'_{i,j}}{k_i \lambda_i}) \cdot \rangle$. Therefore, the spectrum of $\Psi'$ is
\[\set*{\rho_1, \dots, \rho_n} \cup \set*{1^{(k_1-1)}},\]
from which it follows that $\lambda_{\max}(\Psi') \leq \max(\rho_1, 1)$, as desired. 

So far, we have shown that the $(\eta, \epsilon)$-spectral domination of $\mu$ implies that $\mu_{\vec{k}}$ is also $(\eta, \epsilon)$-spectrally dominated for all $\vec{k} \in \N^{n}_{\geq 1}$. To upgrade this to complete spectral independence, we simply apply this result to the conditionals of $\set{\lambda \ast \mu\given \lambda \in (0,1]^{n}}$  and recall that, by \cref{lem:conditional of blow up}, the conditionals of a blow-up are blow-ups of conditionals of $\lambda \ast \mu$ for some $\lambda \in (0,1]^{n}$. 
\end{proof}



\subsection{Down-up walks and field dynamics}
\begin{definition}[Down operator]\label{def:down-operator}
	For a ground set $\Omega$, and  $\card{\Omega} \geq k\geq \l$,  the down operator $D_{k\to \l}\in \R^{\binom{\Omega}{k}\times \binom{\Omega}{\l}}$ is defined to be
	\[ 
		D_{k\to \l}(S, T)=\begin{cases}
			\frac{1}{\binom{k}{\l}}&\text{ if }T\subseteq S,\\
			0&\text{ otherwise}.\\
		\end{cases}
	\]
\end{definition}
Note that $D_{k\to \l}D_{\l\to m}=D_{k\to m}$. 

\begin{definition}[Up operator]
	For a ground set $\Omega$, $\card{\Omega}\geq k\geq \l$, and density $\mu:\binom{\Omega}{k}\to \R_{\geq 0}$, the up operator $U_{\l \to k}\in \R^{\binom{\Omega}{\l}\times \binom{\Omega}{k}}$ is defined to be
	\[ 
		U_{\l\to k}(T, S)=\begin{cases}
			\frac{\mu(S)}{\sum_{S'\supseteq T}\mu(S')}&\text{ if }T\subseteq S,\\
			0&\text{ otherwise}.\\
		\end{cases}
	\]
\end{definition}

\begin{remark}
Note that, unlike the down operator, the up operator depends on the underlying density $\mu$. 
\end{remark}

If we define $\mu_k=\mu$ and more generally let $\mu_\l$ be $\mu_k D_{k\to \l}$, then the down and up operators satisfy
\[ \mu_k(S)D_{k\to \l}(S, T)=\mu_\l(T)U_{\l \to k}(T, S), \]
In other words, $U_{\l \to k}$ is the time-reversal of $D_{k\to \l}$ with respect to the $\mu_k$ and $\mu_\l$ measures. This property ensures that the composition of the down and up operators have the appropriate $\mu_{i}$ as a stationary distribution, are reversible, and have nonnegative real eigenvalues.

\begin{definition}[Down-up walk]
	For a ground set $\Omega$, $\card{\Omega}\geq k\geq \l$, and density $\mu:\binom{\Omega}{k}\to \R_{\geq 0}$, the $k\leftrightarrow \l$ down-up walk is defined by the row-stochastic matrix $D_{k\to \l} U_{\l \to k}$.
\end{definition}

\begin{proposition}[{\cite[see, e.g.,][]{KO18,AL20,ALO20}}]
	The operators $D_{k\to \l}U_{\l\to k}$ and $U_{\l\to k}D_{k\to \l}$ both define Markov chains that are time-reversible and have nonnegative eigenvalues. Moreover $\mu_k$ and $\mu_\l$ are respectively their stationary distributions.
\end{proposition}

The following novel Markov chain, called the field dynamics, was introduced in recent work of \textcite{chen2021rapid}.

\begin{definition}[Field dynamics]
\label{def:field-dynamics}
Let $\mu$ be a probability measure over $2^{[n]}$ and let $\theta \in (0,1)$. The field dynamics with respect to $\mu$ at parameter $\theta$ is a Markov chain on $2^{[n]}$ with transition matrix $\fieldMarkov_{\theta}$ defined as follows: given the current state $\sigma \in 2^{[n]}$, generate a random subset $S\subseteq [n]$ by independently including each $i \in [n]$ into $S$ with probability $\theta$ if $i \in \sigma$ and probability $1$ if $i \notin \sigma$. Generate $\sigma'_S \in 2^{S}$ from the distribution $(\theta \ast \mu)^{\sigma_{[n]\setminus S}}$. Move to the new state $\sigma' = \sigma'_S \cup \sigma_{[n]\setminus S} = \sigma'_S \cup ([n]\setminus S)$.  
\end{definition}

It was shown in \cite{chen2021rapid} that just like the Glauber dynamics, the field dynamics is ergodic and reversible with respect to $\mu$. 

\begin{lemma}[{\cite[Proposition~2.2]{chen2021rapid}}]
Let $\mu$ be a probability measure over $2^{[n]}$ and let $\theta \in (0,1)$. Then, $\fieldMarkov_{\theta}$ is irreducible, aperiodic, and reversible with respect to $\mu$. 
\end{lemma}

The analysis of the field dynamics in \cite{chen2021rapid}, as well as in the present work, hinges on its characterization as the limit of the projection of appropriate down-up walks for a sequence of blow-ups of the distribution $\mu$. The following definition and lemma formalize this connection.

\begin{definition}[Projected block dynamics, {\cite[Proposition~5.3]{chen2021rapid}}]
\label{def:projected block dynamics}
Let $\mu$ be a probability measure over $2^{[n]}$, let $k \in \N$, and let $\mu_{k} := \mu_{(k,k,\dots,k)}$ denote the blow-up distribution on $\Omega = [k + \dots + k]$. Let $1 \leq \ell \leq kn$ be an integer and let $(X_t)_{t\geq 0}$ denote the Markov chain on $\binom{\Omega \cup \bar{\Omega}}{kn}$ starting from $X_0$ and with transition matrix corresponding to the $kn \leftrightarrow kn - \ell$ down-up walk on $\binom{\Omega \cup \bar{\Omega}}{kn}$. We define the $(k,\ell)$-projected-block dynamics to be the stochastic process $(X_t^\ast)_{t\geq 0}$ on $2^{[n]}$, where $X_t^\ast$ denotes the projection of $X_t$. 

The $(k,\ell)$-projected block dynamics is a well-defined Markov chain on $2^{[n]}$ which is reversible with respect to $\mu$. We will denote its transition matrix by $\projBlockMarkov_{k,\ell}$.  
\end{definition}

\begin{lemma}[{\cite[Lemma~5.5]{chen2021rapid}}]
\label{lem:approx-field-block}
For all probability distributions $\mu$ over $2^{[n]}$, for all $\theta \in (0,1)$, and for all $\epsilon > 0$, there exists $K = K(\mu, \theta, \epsilon) \geq 1$ such that for all integers $k\geq K$ and for all $\sigma, \sigma' \in 2^{[n]}$,
\[\abs*{\projBlockMarkov_{k, \ceil{\theta k n}}(\sigma, \sigma') - \fieldMarkov_{\theta}(\sigma, \sigma')} \leq \epsilon.\]
\end{lemma}

\section{Restricted entropy contraction for field dynamics}
\label{sec:entropy-contraction-field}

The goal of this section is to prove \cref{thm:restricted entropy field dynamics}, which shows that if the probability distribution $\mu$ on $2^{[n]}$ is completely $(\eta, \epsilon)$-spectrally independent, then for any $\theta \in (0,1)$, the corresponding field dynamics has a strictly positive restricted entropy contraction constant with respect to a natural class of ``$C$-completely bounded probability distributions'', as captured by the next definition. Crucially, the lower bound on the restricted entropy contraction constant depends only on $\eta, \epsilon, \theta$ and $C$, all of which will be $\Theta(1)$ in our applications.  

\begin{definition}[(Completely) bounded distributions]
\label{def:C-bounded}
Let $\mu$ be a probability distribution over $2^{[n]}$ and let $C \geq 1$. Define the class of $C$-bounded distributions with respect to $\mu$ by  
\[\mathcal{V}(C,\mu) := \set*{\nu \in \operatorname{AC}_{\mu}\given \frac{\nu(i) (1- \mu(i)) }{\mu(i) (1-\nu(i))} \leq C \quad \forall i\in [n]},\]
where $\operatorname{AC}_{\mu}$ denotes the set of all probability measures over $2^{[n]}$ which are absolutely continuous with respect to $\mu$.


Moreover, we define the class of $C$-completely bounded distributions with respect to $\mu$, denoted by $\mathcal{V}^{c}(C,\mu)$, as those distributions $\nu$ such that for all valid partial assignments $\sigma_{\Lambda} \in 2^{\Lambda}$ and all external fields $\lambda \in (0,1]^n,$ $(\lambda \ast \nu)^{\sigma_{\Lambda}} $ is $C$-bounded with respect to $(\lambda \ast \mu)^{\sigma_{\Lambda}}$.  
\end{definition}

\begin{theorem} 
\label{thm:restricted entropy field dynamics}
Let $\eta \geq 1$, $\epsilon > 0$, and $C\geq 1$. Suppose that $\mu$ is a probability distribution on $2^{[n]}$ which is completely $(\eta, \epsilon)$-spectrally independent. Then, for any $\nu \in \mathcal{V}^{c}(C,\mu)$, $\theta \in (0,1)$, $k \in \N$, and $P \in \set*{\projBlockMarkov_{k,\lceil nk \theta \rceil},\fieldMarkov_{\theta}}$, we have
\begin{equation*}
    \DKL {\nu P \river \mu P} \leq  (1- \kappa) \DKL {\nu \river \mu},
\end{equation*}
where $\kappa = (\theta/3)^{\eta'}$ with $\eta' =  \max \set*{2\eta , \sqrt{\log(C)/\log(1+\epsilon)}} .$
\end{theorem}

The proof of this theorem requires a few intermediate steps. We begin with the following lemma, which shows that the blow-up operation preserves (complete) $C$-boundedness,

\begin{lemma} \label{prop:blowup preserve bounded}
Let $\mu$ be a probability distribution over $2^{[n]}$, let $C\geq 1$, and let $\vec{k} \in \N^{n}$. If $\nu \in \mathcal{V}(C,\mu)$, then $\nu_{\vec{k}} \in \mathcal{V}(C, \mu_{\vec{k}})$. Moreover, if $\nu \in \mathcal{V}^c(C,\mu)$, then $\nu_{\vec{k}} \in \mathcal{V}^c(C, \mu_{\vec{k}})$.
\end{lemma}
\begin{proof}
Both assertions are a consequence of the following statement and \cref{lem:conditional of blow up}: fix $\lambda' \in (0,1]^{k_1+\dots + k_n}$ and define $\lambda \in (0,1]^{n}$ by $\lambda_i = \frac{1}{k_i} \sum_{j=1}^{k_i} \lambda'_{i,j}.$ Let $\tilde{\mu} = \lambda \ast \mu$ and $\hat{\mu}_{\vec{k}} = \lambda' \ast \mu_{\vec{k}}$, and define $\tilde{\nu}, \hat{\nu}_{\vec{k}}$ similarly. We claim that if $\tilde{\nu} \in \mathcal{V}(C, \tilde{\mu})$, then $\hat{\nu}_{\vec{k}} \in \mathcal{V}(C,\hat{\mu}_{\vec{k}}).$

To see this, note that for any $i\in [n]$ and $j \in [k_i]$, 
\[\hat{\mu}_{\vec{k}} ((i,j)) = \tilde{\mu}(i)\frac{\lambda'_{i,j}}{k_i\lambda_i }, \quad \hat{\mu}_{\vec{k}} (\overline{(i,j)}) = 1- \tilde{\mu}(i)\frac{\lambda'_{i,j}}{k_i\lambda_i },\] 
and similarly for $\hat{\nu}_{\vec{k}} $ and $\tilde{\nu}.$ If $\tilde{\nu}(i) \leq \tilde{\mu}(i)$, then
\[\frac{\hat{\nu}_{\vec{k}} (i,j) (1-\hat{\mu}_{\vec{k}} (i,j) ) }{\hat{\mu}_{\vec{k}}(i,j) (1 - \hat{\nu}_{\vec{k}} (i,j) )}  = \frac{\tilde{\nu}(i) \lambda'_{i,j}/(k_i \lambda_i)}{1- \tilde{\nu}(i)\lambda'_{i,j}/(k_i \lambda_i)} \cdot \parens*{\frac{\tilde{\mu}(i) \lambda'_{i,j}/(k_i \lambda_i)} {1- \tilde{\mu}(i)\lambda'_{i,j}/(k_i \lambda_i)}}^{-1} \leq 1\]
On the other hand, if $\tilde{\nu}(i) \geq \tilde{\mu}(i)$, then
\[\frac{\hat{\nu}_{\vec{k}} (i,j) (1-\hat{\mu}_{\vec{k}} (i,j) ) }{\hat{\mu}_{\vec{k}}(i,j) (1 - \hat{\nu}_{\vec{k}} (i,j) )}  = \frac{\tilde{\nu}(i) (1- \tilde{\mu}(i) \lambda'_{i,j}/(k_i \lambda_i) )}{\tilde{\mu}(i) (1- \tilde{\nu}(i) \lambda'_{i,j}/(k_i \lambda_i) )} \leq \frac{\tilde{\nu}(i) (1- \tilde{\mu}(i))}{\tilde{\mu}(i) (1-\tilde{\nu}(i))} \leq C.\]
Thus, in either case, we have $\frac{\nu_{\vec{k}}(i,j) (1-\mu_{\vec{k}}(i,j)) }{\mu_{\vec{k}}(i,j) (1 - \nu_{\vec{k}}(i,j))}\leq C$, as desired.  
\end{proof}

The next proposition shows that spectral domination of $\mu$ implies a restricted form of entropic independence (see \cite{anari2021entropic}), for probability distributions $\nu$ which are $O(1)$-bounded with respect to $\mu$. In the special case that we can take $\epsilon = \infty$, we do not need the $O(1)$-boundedness condition, and recover the fact that fractional log concavity implies entropic independence from \cite{anari2021entropic}.
\begin{proposition}
\label{prop:entropy-contraction-n-to-1}
Let $\eta \geq 1$, $\epsilon > 0$, and $C\geq 1$. Suppose that $\mu$ is a probability distribution on $2^{[n]}$ which is $(\eta, \epsilon)$-spectrally dominated. Then, for any $\nu \in \mathcal{V}(C, \mu)$, we have that
\[\DKL {\nu^{\hom} D_{n \to 1}  \river \mu^{\hom} D_{n \to 1} } \leq \frac{\eta'}{n}\DKL{\nu^{\hom} \river \mu^{\hom}}, \]
where $\eta' =  \max \set*{2\eta , \sqrt{\log(C)/\log(1+\epsilon)}} .$ 
\end{proposition}

The proof of this proposition is similar to the proof of the aforementioned statement in \cite{anari2021entropic}, and requires a couple of preliminary lemmas. The first one is folklore and the second is a direct consequence of convex duality.

\begin{lemma}[Folklore]
\label{lem:lc-equiv}
Let $\mathcal{C} \subseteq \R_{\geq 0}^{n}$ denote a convex cone.  
		For a $d$-homogeneous function $f:\mathcal{C}\to\R_{\geq 0}$ the following are all equivalent:
		\begin{enumerate}
			\item $f$ is quasi-concave.
			\item $f$ is log-concave.
			\item $f$ is $d$-th-root-concave, i.e., $f^{1/d}$ is concave.
		\end{enumerate}  
\end{lemma}
\begin{lemma}[Lemma 26 of \cite{anari2021entropic}, see also {\cite[Appendix B]{singh2014entropy}}]
\label{lem:KL-dual}
Consider a homogeneous distribution $\mu: \binom{[n]}{k} \to \R_{\geq 0}$ and let  $g_{\mu}(z_1, \dots, z_n)$ be its multivariate generating polynomial. Then, for any $q \in \R^{n}_{\geq 0}$ with $\sum_{i=q}^{n}q_i = 1$, we have
\[ \inf\set*{\DKL{\nu \river \mu} \given \nu D_{k\to 1}=q } = -\log\parens*{\inf_{z_1,\dots, z_n>0} \frac{g_\mu(z_1,\dots,z_n)}{z_1^{kq_1}\cdots z_n^{kq_n}} }. \]
\end{lemma}

\begin{proof}[Proof of \cref{prop:entropy-contraction-n-to-1}]
For $i \in [n]$, let $p_i = \P_{S\sim \mu}{i \in S}$ and $q_i = \P_{S\sim \nu}{i \in S}$. Note that $\mu^{\hom}D_{n\to 1} = (p_1,\dots, p_n, 1-p_1,\dots, 1-p_n)$ and similarly for $\nu^{\hom}D_{n\to 1}$.  
Let \[g = g_{\mu^{\hom}}(z_1,\dots, z_n, z_{\bar{1}},\dots, z_{\bar{n}})\] denote the multivariate generating polynomial of $\mu^{\hom}$. By \cref{prop:spectralind-to-flc}, for all $0<\alpha \leq 1/2\eta$, $\log g(\set{z_i^{\alpha}, z_{\bar{i}}^{\alpha}})$ is concave on the convex cone
\[\Lambda_{\alpha, \epsilon} := \set{(z_1,\dots,z_n, z_{\bar{1}},\dots, z_{\bar{n}})\given 0\leq z_i \leq z_{\bar{i}}(1+\epsilon)^{1/\alpha} \quad\forall i \in [n]}.\]
In particular, it follows from \cref{lem:lc-equiv} (cf. \cite{anari2021entropic}) that 
\[f(z_1,\dots, z_n, z_{\bar{1}},\dots, z_{\bar{n}}) := g(z_1^{\alpha},\dots, z_n^{\alpha}, z_{\bar{1}}^{\alpha},\dots, z_{\bar{n}}^{\alpha})^{1/\alpha n}\] 
is concave on $\Lambda_{\alpha, \epsilon}$, and hence, for $(z_1,\dots,z_n,z_{\bar{1}},\dots, z_{\bar{n}}) \in \Lambda_{\alpha, \epsilon}$,
\begin{equation*}
     g(z_1, \dots, z_n, z_{\bar{1}}, \ldots, z_{\bar{n}}) \leq \parens*{\sum_{i=1} p_i z_i^{1/\alpha} + (1-p_i) z_{\bar{i}}^{1/\alpha} }^{\alpha n}.
\end{equation*}
Let $\alpha  = 1/\eta'$ and let $z^* \in \R^{2n}_{\geq 0}$ be defined by $z^*_i = (q_i/p_i)^{\alpha}$ and $z^*_{\bar{i}} = ((1-q_i)/(1-p_i))^{\alpha}$. 
Since $\nu \in \mathcal{V}(C,\mu)$, it 
is seen by expanding the definition of $\mathcal{V}(C,\mu)$ that $z^* \in \Lambda_{\alpha, \epsilon}$. Therefore, by \cref{lem:KL-dual} (cf. \cite{anari2021entropic}), we have that
\begin{align*}
    \DKL{\nu^{\hom}\river \mu^{\hom}} &\geq -\log \parens*{\frac{g(z_1^*,\dots, z_n^*, z_{\bar{1}}^*,\dots, z_{\bar{n}}^*)}{(z_1^*)^{nq_1}\dots (z_{n}^*)^{nq_n}(z_{\bar{1}}^*)^{n(1-q_1)}\dots (z_{\bar{n}}^*)^{n(1-q_n)}}}\\
    &\geq -\log\parens*{\frac{1}{(z_1^*)^{nq_1}\dots (z_{n}^*)^{nq_n}(z_{\bar{1}}^*)^{n(1-q_1)}\dots (z_{\bar{n}}^*)^{n(1-q_n)}}}\\
    &\geq \alpha n \DKL{\nu^{\hom}D_{n\to 1}\river \mu^{\hom}D_{n\to 1}},
\end{align*}
as desired. \qedhere

\end{proof}

\begin{corollary} \label{cor:restricted entropy contraction}
Let $\eta \geq 1$, $\epsilon > 0$, and $C\geq 1$. Suppose that $\mu$ is a probability distribution on $2^{[n]}$ which is completely $(\eta, \epsilon)$-spectrally independent. Then, for any $\nu \in \mathcal{V}^{c}(C,\mu)$, $\theta \in (0,1)$, and $\ell = \lceil \theta n \rceil$, we have that
\begin{equation*}
    \DKL {\nu^{\hom} D_{n \to (n-\ell)} \river \mu^{\hom} D_{n \to (n-\ell)}} \leq  (1- \kappa) \DKL {\nu^{\hom} \river \mu^{\hom}},
\end{equation*}
where $\kappa = (\theta/3)^{\eta'}$ with $\eta' =  \max \set*{2\eta , \sqrt{\log(C)/\log(1+\epsilon)}} .$   
\end{corollary}
\begin{proof}
This follows immediately by applying \cref{prop:entropy-contraction-n-to-1} to the conditionals of $\mu$ and $\nu$ and using exactly the same local-to-global argument as in \cite{anari2021entropic} (cf.~\cite{alimohammadi2021fractionally}).  
\end{proof}

Finally, we are in a position to state and prove the main result of this section.

\begin{proof}[Proof of \cref{thm:restricted entropy field dynamics}]
First, let $P = \projBlockMarkov_{k, \lceil n k \theta \rceil}$, for $k \in \N$. Let $\vec{k} = (k,\dots, k) \in \N^{n}$. By \cref{prop:blowup preserve complete spectral ind}, $\mu_{\vec{k}}$ is completely $(\eta, \epsilon)$-spectrally independent and by \cref{prop:blowup preserve bounded}, $\nu_{\vec{k}} \in \mathcal{V}^{c}(C, \mu_{\vec{k}})$.    
Therefore, with $\ell = \lceil n k \theta \rceil$, we have
\begin{align*}
     \DKL {\nu P \river \mu P}  &\leq \DKL {\nu_{\vec{k}}^{\hom} D_{nk \to (nk-\ell)} \river \mu_{\vec{k}}^{\hom} D_{nk \to (nk-\ell)}} & \textrm{(data processing inequality)}\\
     &\leq (1- \kappa) \DKL {\nu_{\vec{k}}^{\hom} \river \mu_{\vec{k}}^{\hom}} & \textrm{\cref{cor:restricted entropy contraction}}\\
     &=  (1- \kappa) \DKL{\nu_{\vec{k}} \river \mu_{\vec{k}}}\\
     &= (1-\kappa) \DKL{\nu \river \mu}. & \textrm{\cref{lem:blow up preserves KL}}
\end{align*}
The result for $P = \fieldMarkov_{\theta}$ follows from the above and \cref{lem:approx-field-block}. 
\end{proof}

\section{Comparison of Glauber dynamics and field dynamics}\label{sec:comparison}
In the previous section, we established restricted entropy contraction for the field dynamics, with respect to the class of C-completely bounded measures. In this section, we will establish such a result for the usual Glauber dynamics, in the (stronger) form of approximate tensorization of entropy. More precisely, we will show the following.  

\begin{theorem}
\label{thm:comparison}
 Let $\mu$ be a probability distribution on $2^{[n]}$ which is completely $(\eta, \epsilon)$-spectrally independent. Let $\nu \in \mathcal{V}^{c}(C,\mu)$ and $\theta \in (0,1)$, and let $f := d\nu/d\mu$. Suppose there exists $\kappa > 0$ such that for $\pi := \theta \ast \mu := (\theta,\dots, \theta) \ast \mu$ and for all $R \subseteq [n]$ with $\pi(R) > 0$, 
\begin{equation}\label{eqn:mlsi-easy}
\kappa \Ent_{\pi^{{\bf 1}_R}}(f) \le \mathcal{E}_{P_{\theta}}(f,\log f),
\end{equation}
where $P_{\theta}$ is the transition matrix for the Glauber dynamics with respect to $\pi^{{\bf 1}_R}$ and $\mathcal{E}_{P_{\theta}}$ is its corresponding Dirichlet form. Then, we have
 \[ \Ent_{\mu}[f] \le C' \sum_{v \in [n]} \E_{\mu}{\Ent_v[f]},\]
 where $C' = \frac{C}{\kappa n}\times \frac{1}{\Omega(\theta)^{O(\eta')}}$, with $\eta' =  \max \set*{2\eta , \sqrt{\log(C)/\log(1+\epsilon)}}$. 
\end{theorem}

\begin{remark}
\cref{eqn:mlsi-easy} amounts to an MLSI for $P_{\theta}$ with constant $\Theta(\kappa)$. In our applications, we will have $\kappa = \Theta(1/n)$, so that $C'$ will be a constant depending only on $C, \eta, \epsilon, \theta$ (and crucially, not on $n$). 
\end{remark}

Intuitively, the above theorem uses \cref{cor:restricted entropy contraction} (which is stronger than restricted entropy contraction of the field dynamics) 
to transfer a modified log-Sobolev inequality for the Glauber dynamics with respect to the ``easier measure'' $\theta \ast \mu$ to an approximate entropy tensorization statement for the true measure $\mu$. Such a scheme was also employed in \cite{chen2021rapid} at the level of the spectral gap. However, here, considerably more care is required due to two reasons: (i) we only have a \emph{restricted} version of entropy contraction for $D_{n \to (n - \theta n)}$; (ii) in contrast to the spectral gap, modified log-Sobolev inequalities and approximate tensorization statements are significantly more challenging to establish for the ``easier measure''. Indeed, what we would really like to have is approximate tensorization of entropy for the ``easier measure'', but as this does not seem to be available in the literature for, e.g., the hard-core model, we show how to make do with only a modified log-Sobolev inequality (see \cref{lem:restricted mlsi to tensorization}).

To begin with, we show how to obtain a version of \cref{thm:comparison} if (\cref{eqn:mlsi-easy}) is replaced by the stronger condition that approximate tensorization of entropy is available for the ``easier measure'' (\cref{eq:entfac-easy} below).

\begin{proposition}\label{lem:ent-tensorization}
Let $V$ be a finite set. Let $\mu$ be a probability measure on $2^V$, $\theta \in (0,1)$,
and let $\pi := \theta * \mu := (\theta,\ldots,\theta) * \mu$ be the tilt of $\mu$ under uniform external field $\theta$. 
Suppose $f$ is a positive function on $2^V$ such that:
\begin{enumerate}
    \item For every $R\subseteq V$ with 
    $\pi(R) > 0$, we have 
\begin{equation}\label{eq:entfac-easy}
\Ent_{\pi^{{\bf 1}_R}}(f) \le C_1 \sum_{v\in V} \pi^{{\bf 1}_R}[\Ent_v(f)].
\end{equation}
    \item For $R$ a random subset of $V$, where each element $v\in V$ is included independently with probability $1-\theta$, we have
\begin{equation}\label{eq:entfac-field}
\Ent_{\mu}(f) \le C_2\frac{Z_\theta}{\theta^{\card{V}}} \E_R { \pi(R) \Ent_{\pi^{{\bf 1}_R}}(f)}, 
\end{equation}
    where $Z_{\theta}$ is the normalizing constant of $\pi = \theta * \mu$ as in \cref{def:measure-tilt}.
\end{enumerate}
Then, we have that 
\[
\Ent_{\mu}(f) \le C \sum_{v\in V} \mu[\Ent_v(f)]
\]
for $C = C_1 C_2 \theta^{-1}$.
\end{proposition}
\begin{proof}
The proof of this lemma is essentially the entropic version of \cite[Lemma~4.1]{chen2021rapid}. Concretely, applying the assumptions \cref{eq:entfac-field} and then \cref{eq:entfac-easy}, we have 
\begin{align*}
\Ent_{\mu}(f) &\le C_2\frac{Z_\theta}{\theta^{\card{V}}} \E_R {\sum_{v\in V} \pi(R) \Ent_{\pi^{{\bf 1}_R}}(f)}\\
&\le {C_2C_1} \frac{Z_\theta}{\theta^{\card{V}}} \E_R {\pi(R) \sum_{v\in V} \pi^{{\bf 1}_R}(\Ent_v(f))}.
\end{align*}
For $\sigma \in 2^{V}$, we denote by $\card{\sigma}$ the number of elements in $\sigma$. Following the proof of Equation (22) in  \cite[Lemma~4.1]{chen2021rapid}, we have 
\begin{equation}\label{eq:exp-over-R}
\E_R { \pi(R)\pi^{{\bf 1}_R} (\Ent_{v}(f))} = \E_{\sigma \sim \pi_{V-v}}{\theta^{\card{V}-\card{\sigma}} \Ent_{\pi^\sigma}(f)},
\end{equation}
where $\pi_{V-v}$ denotes the marginalization of $\pi$ to $2^{V\setminus \set{v}}$. 
Next, we define the notation $\pi^\sigma_+ = \pi^{\sigma}(\set{v})$ and $\pi^\sigma_- = \pi^{\sigma}(\emptyset)$, and similarly, define $\mu^{\sigma}_+$, $\mu^{\sigma}_-$, $f^\sigma_+$, and $f^\sigma_-$. Also, for $\sigma \in 2^{V\setminus \set{v}}$, we let $\sigma_+ := \sigma \cup \set{v}$ and $\sigma_- := \sigma$ (but now viewed as a subset of $V$). 
Since $\pi^{\sigma}_-  = \mu(\sigma_-) \theta^{\card{\sigma}}/\pi(\sigma) Z_\theta$ and similarly for $\pi^{\sigma}_+$, we have that  
\begin{equation}\label{eqn:pi-to-mu-site}
\frac{\pi_+^\sigma}{\pi_-^\sigma} =  \frac{\mu(\sigma_+)\theta^{\card{\sigma}+1}}{\mu(\sigma_-)\theta^{\card{\sigma}}} = \theta \frac{\mu^\sigma_+}{\mu^\sigma_-}.
\end{equation}
Since $\theta \leq 1$, this shows that $\pi_-^{\sigma} \ge \mu^{\sigma}_-$.

We need the following numerical inequality, whose proof is deferred to \cref{sec:appendix-numerical}. 
\begin{lemma}\label{lem:ent-compare}
For $a,b,x,\eta>0$, we have 
\begin{align*}
&\frac{1}{x+1}a\log a+ \frac{x}{x+1}b\log b - \frac{a+xb}{x+1}\log\frac{a+xb}{x+1} \\
&\le \max(\eta,\eta^{-1}) \parens*{\frac{1}{\eta x+1}a\log a+ \frac{\eta x}{\eta x+1}b\log b - \frac{a+\eta xb}{\eta x+1}\log\frac{a+\eta xb}{\eta x+1}}.
\end{align*}
Equivalently, for any $f : \set{0,1} \to \R_{> 0}$, $\Ent_{\operatorname{Ber}(1/(x + 1))}[f] \le \max(\eta,\eta^{-1}) \Ent_{\operatorname{Ber}(1/(\eta x + 1))}[f]$.
\end{lemma}

Applying \cref{lem:ent-compare} with $\eta = \theta^{-1}, x = \frac{\pi^{\sigma}_+}{\pi^{\sigma}_-}, a = f^{\sigma}_-, b = f^{\sigma}_+$ and using \cref{eqn:pi-to-mu-site}, we can bound the RHS of \cref{eq:exp-over-R} as follows:
\begin{align*}
\E_{\sigma \sim \pi_{V-v}}{\theta^{\card{V}-\card{\sigma}} \Ent_{\pi^\sigma}(f)}
&\le \theta^{-1} \E_{\sigma \sim \pi_{V-v}}{\theta^{\card{V}-\card{\sigma}} \Ent_{\mu^\sigma}(f)} \\
&= \theta^{-1} \sum_{\sigma} \theta^{\card{V}-\card{\sigma}} \pi_{V-v}(\sigma) \Ent_{\mu^\sigma}(f) \\
&\le \theta^{-1} \sum_{\sigma} \theta^{\card{V}-\card{\sigma}} \frac{\pi_{V}(\sigma_-)}{\mu^{\sigma}_-} \Ent_{\mu^\sigma}(f)\\
&= \theta^{-1} \frac{\theta^{\card{V}}}{Z_\theta}\sum_{\sigma \in \Omega(\pi_{V-v})} \mu_{V-v}(\sigma)  \Ent_{\mu^\sigma}(f),
\end{align*}
where we have used $\pi_{V - v}(\sigma) = \frac{\pi_{V}(\sigma_-)}{\pi^{\sigma}_-} \le \frac{\pi_{V}(\sigma_-)}{\mu^{\sigma}_-}$ in the third line and
$\pi_V(\sigma_-) = \mu(\sigma_-) \theta^{\card{\sigma}}/Z_\theta = \mu_{V - v}(\sigma) \mu^{\sigma}_- \theta^{\card{\sigma}}/Z_\theta$ in the last line.
Combining with the previous bounds, we obtain the conclusion of \cref{lem:ent-tensorization}. 
\end{proof}

\subsection{Approximate tensorization of entropy for bounded functions from MLSI}
It is well-known that approximate tensorization of entropy implies the modified Log-Sobolev inequality (MLSI) with the same constant, due to Jensen's inequality (see e.g. \cite{caputo2015approximate}). In the following lemma, we show that the reverse implication also holds for suitably nice functions $f$.

\begin{lemma}\label{lem:restricted mlsi to tensorization}
Suppose that $P$ is the transition matrix of a Markov chain on $2^V$ with stationary distribution $\mu$ such that all transitions of $P$ consist of adding or removing at most one element from the current set. Let $f \colon 2^{V} \to \R_{\geq 0}$ and suppose there exists $C > 0$ such that the following holds: for all $v \in V$ and $\sigma \in 2^{V \setminus \set{v}}$, 
\[\frac{f(\sigma_+)}{f(\sigma_-)} \in [C^{-1}, C],\]
where $\sigma_+ = \sigma \cup \set{v}$ and $\sigma_- = \sigma$ (viewed as a subset of $V$). 
Then,
\[
\mathcal{E}_P(f,\log f) \le \frac{C}{n} \sum_{v\in V} \mu[\Ent_v(f)].
\]
In particular, if there also exists $\kappa > 0$ such that $\kappa \Ent_{\mu}(f) \le \mathcal{E}_P(f,\log f)$,
then approximate tensorization of entropy (\cref{def:approx-tensorization}) holds for $f$ with constant $\kappa^{-1}C/n$. 
\end{lemma}

We will need the following numerical lemma, whose proof is deferred to \cref{sec:appendix-numerical}.

\begin{lemma}\label{lem:dirichlet-to-entropy}
Suppose positive numbers $f_+$, $f_-$, and $C$ satisfy that $f_+/f_- \in [C^{-1},C]$ and suppose that $\mu_+, \mu_- > 0$ satisfy $\mu_++\mu_-=1$. Then we have 
\begin{align*}
&C\parens*{\mu_-f_-\log f_-+\mu_+f_+\log f_+ - (\mu_-f_-+\mu_+f_+)\log(\mu_-f_-+\mu_+f_+)} \\
&\quad \ge \mu_-f_-\log f_-+\mu_+f_+\log f_+ - (\mu_-f_-+\mu_+f_+)(\mu_+\log f_+ + \mu_-\log f_-).
\end{align*}
\end{lemma}

\begin{proof}[Proof of \cref{lem:restricted mlsi to tensorization}]
The lemma follows immediately from \cref{lem:dirichlet-to-entropy}. This is because we have (using the notation introduced in the previous subsection),
\begin{multline*}
\mathcal{E}_P(f,\log f) = \frac{1}{2} \sum_{x,y} \mu(x)P(x,y)(f(x)-f(y))(\log f(x)-\log f(y)) \\
= \frac{1}{n} \sum_{v\in V} \E_{\sigma \sim \mu_{V-v}}{\mu^{\sigma}_+ f^{\sigma}_+ \log f^{\sigma}_+ + \mu^{\sigma}_- f^{\sigma}_- \log f^{\sigma}_- - (\mu^{\sigma}_{+}f^{\sigma}_+ +\mu^{\sigma}_{-} f^{\sigma}_-)(\mu^{\sigma}_{+}\log f^{\sigma}_+ + \mu^{\sigma}_{-} \log f^{\sigma}_{-})},
\end{multline*}
while
\begin{align*}
\sum_{v\in V}\mu[\Ent_{v}{f}] &=  \sum_{v\in V} \E_{\sigma \sim \mu_{V-v}}{\mu^{\sigma}_{+}f^{\sigma}_{+}\log f^{\sigma}_{+}+\mu^{\sigma}_{-}f^{\sigma}_{-}\log f^{\sigma}_{-}  - (\mu^{\sigma}_+ f^{\sigma}_+ +\mu^{\sigma}_-f^{\sigma}_-)\log(\mu^{\sigma}_+f^{\sigma}_+ +\mu^{\sigma}_- f^{\sigma}_-)},
\end{align*}
and \cref{lem:dirichlet-to-entropy} lets us compare these two sums term by term.
\end{proof}

\subsection{Approximate entropy factorization for the field dynamics}

In this subsection, we show how to deduce \cref{eq:entfac-field} for functions $f$ with $f(\sigma_+)/f(\sigma_-)$ bounded away from $0$ and $\infty$, assuming approximate tensorization of entropy for the block dynamics on the $k$-blow-up distribution (we will use this terminology to refer to the $(k,\dots, k)$-blow-up distribution). For convenience of notation, we will denote the underlying space of the $k$-blow-up distribution by $2^{V_k}$. Recall (\cref{def:projection}) that there is a natural projection map from $2^{V_k}$ to $2^{V}$. Accordingly, for a function $f \colon 2^{V} \to \R_{\geq 0}$, we define its $k$-lift to be the function on $2^{V_k}$, given by the pulling back $f$ by the projection.    

\begin{lemma}\label{lem:field tensorization from block}
Let $\mu$ be a probability measure on $2^{V}$ and let $f \colon 2^{V} \to \R_{\geq 0}$.  Let $\theta \in (0,1)$ and $\pi = \theta \ast \mu = (\theta,\dots, \theta) \ast \mu$, with $Z_{\theta}$ as in \cref{def:measure-tilt}.  
Suppose that there exists $\overline{C} > 0$ such that for all integers $k \geq 1$, for
$\overline{f}$ the $k$-lift of $f$, and for $\ell = \theta nk$, 
\begin{equation}\label{eq:ent-k-lift}
    \Ent(\overline{f}) \le (1 + o_{k \to \infty}(1)) \frac{\overline{C}}{\binom{nk}{\ell}} \sum_{\overline{S} \subseteq \binom{V_k}{\ell}} \mu_k[\Ent_{\overline{S}}(\overline{f})].
\end{equation}
Then, letting $R$ be a random subset of $V$ where each element $v\in V$ is included independently with probability $1-\theta$, we have
\begin{equation}\label{eq:entfac-field-1}
\Ent_{\mu}(f) \le \overline{C} \frac{Z_\theta}{\theta^{\card{V}}} \E_R { \pi(R) \Ent_{\pi^{{\bf 1}_R}}(f)}. 
\end{equation}
\end{lemma}

\begin{proof}
Let $P_k:V_k\to V$ denote the natural projection map (\cref{def:projection}). Given $\overline{S} \subseteq V_k$ and a configuration $\sigma_k$ on $V_k\setminus \overline{S}$ (i.e.~$\sigma_k$ is a subset of $V_k\setminus \overline{S}$), let $S$ be the complement of the set of vertices $v$ in $V$ with $P_k^{-1}(v) \subseteq V_k\setminus \overline{S}$ or $P_k^{-1}(v) \cap ((V_k\setminus \overline{S})\cap \sigma_k) \ne \emptyset$. Let $\sigma \subseteq V \setminus S$ be the projection of $\sigma_k$ on $V\setminus S$. For each $v \in S$, let $C_v$ be the number of vertices in $\overline{S}$ whose projection is $v$. For a subset $\xi$ of $S$, let $\overline{\xi}$ denote the collection subsets of $\overline{S}$ that project to $\xi$. Then 
\[
\mu_{k,\overline{S}}^{\sigma_k}[\overline{\xi}] = \frac{\mu[\xi \mid \sigma] \prod_{v\in \xi} \frac{C_v}{k}}{\sum_{\xi}\mu[\xi \mid \sigma] \prod_{v\in \xi} \frac{C_v}{k}}.
\]
Now, let $\overline{S}$ be a uniformly random subset of $V_k$ of size $\ell := \theta nk$. For $\varepsilon>0$, let $\mathcal{E}(\varepsilon)$ denote the event that for all $v\in V$, $\card{\overline{S}\cap P_k^{-1}(v)} \in [(\theta - \varepsilon)k, (\theta + \varepsilon) k]$. 
By basic concentration estimates and the union bound, there exists a sequence $\epsilon_k$ such that $\epsilon_k \to 0$ and $\P_{\overline{S}}{\mathcal{E}(\varepsilon_k)} \to 1$ as $k \to \infty$.

In the following, we denote by $(\theta \ast \mu)[\xi \mid \sigma]$ the probability of $\xi$ conditional on $\xi$ agreeing with $\sigma$ on $V\setminus S$. Observe that for $\overline{S}$ satisfying $\mathcal{E}(\varepsilon)$, we have
\[ \mu_{k,\overline{S}}^{\sigma_k}[\overline{\xi}] \in [(1-\varepsilon/\theta)^n (\theta\ast\mu)[\xi \mid \sigma], (1+\varepsilon/\theta)^n (\theta\ast\mu)[\xi \mid \sigma]] \]
and so for $\overline{S} \in \mathcal{E} := \mathcal{E}(\varepsilon_k)$ we have
\begin{equation} \label{eq:approx block dynamics}
    \mu_{k,\overline{S}}^{\sigma_k}[\overline{\xi}] 
    = (1 + o_{k \to \infty}(1)) (\theta\ast\mu)[\xi \mid \sigma].
\end{equation}
We also have the bound for all $\overline{S}$ that (for $\sigma_k$ a sample from the marginal law on $V_k \setminus \overline{S}$)
\[
\mu_k[\Ent_{\overline{S}}(\overline{f})]\le \mu_k[\Ent_{\overline{S}}(\overline{f})] + \Ent_{\mu_k}{\E{\overline{f} \mid \sigma_k }}= \Ent(\overline{f}).
\]
    Thus, applying this bound to the right hand side of \cref{eq:ent-k-lift} for $\overline{S} \not \in \mathcal{E}$, we get 
    \[ \Ent(\overline{f}) \le (1 + o_{k \to \infty}(1)) \frac{\overline{C}}{\binom{nk}{\ell}} \sum_{\overline{S} \in \mathcal{E}} \mu_k[\Ent_{\overline{S}}(\overline{f})] +  o_{k \to \infty}(1) \cdot \Ent(\overline{f}) \]
    and eliminating $\Ent(\overline{f})$ from the right hand side gives
\[
\Ent(\overline{f})\le \frac{(1 + o_{k \to \infty}(1))
\overline{C}}{\binom{nk}{\ell}}\sum_{\overline{S}\in \mathcal{E}} \mu_k[\Ent_{\overline{S}}(\overline{f})].
\]

Let $\tau_k$ be a sample from $\mu_k$. Using \cref{eq:approx block dynamics}, we have 
\begin{align*}
\Ent_{\overline{S}}(\overline{f}) = (1+o_{k \to \infty}(1))\E_{\sigma_k = \tau_k[\overline{S}]}{ \Ent_{\theta\ast\mu_{\sigma}}(f)}.
\end{align*}
Note that if $\overline{S}$ satisfies $\mathcal{E}$, then $\sigma$ is always equal to $V\setminus S$. 
Fixing a subset $R$ of $V$, the probability that $\sigma = R$ (given $\overline{S}$ satisfying $\mathcal{E}$) is given by 
\[
\sum_{\xi \in 2^V : R\subseteq \xi} \mu[\xi] \prod_{v\in R}\frac{k-C_v}{k}\prod_{v\notin R : v\in \xi}\frac{C_v}{k} = (1+o_{k \to \infty}(1))\sum_{\xi \in 2^V : R\subseteq \xi} \mu[\xi] (1-\theta)^{\card{R}} \theta^{\card{\xi}-\card{R}}.
\]
Note that
\[
(\theta\ast\mu)[R] = \sum_{\xi \in 2^V : R\subseteq \xi} \mu[\xi] \theta^{\card{\xi}}/Z_\theta.
\]
Thus,
\begin{align*}
\Ent_{\overline{S}}(\overline{f}) 
&= (1+o_{k \to \infty}(1))\E_{\sigma_k = \tau_k[\overline{S}]}{ \Ent_{(\theta\ast\mu)^{\sigma}}(f)}\\
&=(1+o_{k \to \infty}(1))\sum_{R\subseteq V} \parens*{\sum_{\xi \in 2^V : R\subseteq \xi} \mu[\xi] (1-\theta)^{\card{R}} \theta^{\card{\xi}-\card{R}}} \Ent_{(\theta\ast\mu)^{{\bf 1}_R}}(f)\\
&=(1+o_{k \to \infty}(1))Z_\theta \sum_{R\subseteq V} \parens*{\frac{1-\theta}{\theta}}^{\card{R}} (\theta\ast\mu)[R] \Ent_{(\theta\ast\mu)^{{\bf 1}_R}}(f)\\
&=(1+o_{k \to \infty}(1))\frac{Z_\theta}{\theta^{\card{V}}} \E_R{ (\theta\ast\mu)[R] \Ent_{(\theta\ast\mu)^{{\bf 1}_R}}(f)}.
\end{align*}
Hence,
\[
\Ent(f)=\Ent(\overline{f})\le \overline{C}(1+o_{k \to \infty}(1))\frac{Z_\theta}{\theta^{\card{V}}} \E_R{ (\theta\ast\mu)[R] \Ent_{(\theta\ast\mu)^{{\bf 1}_R}}(f)}.
\]
Taking $k\to \infty$, we obtain the result.
\end{proof}

\subsection{Proof of \texorpdfstring{\cref{thm:comparison}}{Theorem \ref{thm:comparison}}}
\begin{proof}[Proof of \cref{thm:comparison}]
This follows immediately from \cref{lem:ent-tensorization}, once we verify its assumptions.

\cref{eq:entfac-easy} follows with $C_1 = \kappa^{-1}C/n$ directly from assumption \cref{eqn:mlsi-easy}
and \cref{lem:restricted mlsi to tensorization}. 

\cref{eq:entfac-field} follows with $C_2 = \Omega(\theta)^{-O(\eta')}$ from the combination of the fact that blowups preserve complete spectral independence (\cref{prop:blowup preserve complete spectral ind}) and (complete)-C-boundedness (\cref{prop:blowup preserve bounded}); \cref{cor:restricted entropy contraction} which implies approximate ``block factorization'' of entropy for the block dynamics on the blowup distribution (more explicitly, that \cref{eq:ent-k-lift} holds with $\overline{C} = \Omega(\theta)^{-O(\eta')}$, where $\eta'$ is as in \cref{cor:restricted entropy contraction}); and \cref{lem:field tensorization from block} which concludes from this \cref{eq:entfac-field-1}, which is precisely what is needed. 
\end{proof}

\section{Sampling in the hardcore model}
\label{sec:sampling-hardcore}
In this section, we use the theory developed in the previous sections to establish our main results for sampling in the hardcore model. All of the key ideas in our approach (with the exception of the trick of using Bernoulli factories) already appear in this setting. 

Recall that the \emph{hardcore model} on a graph $G = (V,E)$ with \emph{fugacity} $\lambda_v$ at site $v \in V$ is the probability distribution $\mu$ over independent sets $\sigma \subset 2^{V}$ given by
\[ \mu(\sigma) \propto \prod_{v \in \sigma} \lambda_v. \]
Here $\mu(\sigma)$ is defined to be zero if $\sigma$ is not an independent set, i.e. if there exist neighbors $u \sim v$ in $G$ such that $\sigma_u = 1 = \sigma_v$. This is a pairwise Markov Random Field on the graph $G$, and from the definition, we can check that the conditional law of $\sigma_v$, given the configuration $\sigma_u$ at all other sites $u$, is 
\begin{equation}\label{eq:hardcore-conditional}
\P_{\mu}{v \in \sigma \mid (\sigma_u)_{u : u \ne v}} = \frac{\lambda_v}{1 + \lambda_v} 1_{\sigma \cap N(v) = \emptyset}, 
\end{equation}
i.e., if one of its neighbors is occupied then the probability is zero, and otherwise it is $\lambda_v/(1 + \lambda_v)$. In particular, the probability that site $v$ is unoccupied is always lower bounded by $1/(1 + \lambda_v)$ for any fixing of its neighbors.  

For $\Delta \ge 3$, we let
\begin{equation}\label{eq:hardcore-uniqueness} \lambda_{\Delta} := \frac{(\Delta - 1)^{\Delta - 1}}{(\Delta - 2)^{\Delta}}
\end{equation}
denote the critical fugacity for the uniqueness threshold on the $\Delta$-regular tree. For $\lambda \geq \lambda_{\Delta}$, sampling from the hardcore distribution with fugacity $\lambda$ on general graphs of maximum degree $\Delta$ is not computationally tractable \cite{sly2010computational}. 

Note that, using the inequality $1 + x \le e^x$, we have
\begin{equation}\label{eq:lambda-delta-bound}
\lambda_{\Delta} = \frac{(\Delta - 1)^{\Delta - 1}}{(\Delta - 2)^{\Delta}} = \frac{1}{\Delta - 2} (1 + 1/(\Delta - 2))^{\Delta - 1} \le \frac{1}{\Delta - 2} e^{(\Delta - 1)/(\Delta - 2)}. 
\end{equation}
Since this approximation is close to tight for large $\Delta$, we have that the uniqueness threshold is roughly $e/\Delta$ for large $\Delta$. We also have for any $\Delta \ge 3$ that
$\lambda_{\Delta} \le 3e^2/\Delta$.


The following proposition collects some facts about the hardcore model with all fugacities below $\lambda_{\Delta}$.
 \begin{proposition} \label{prop:hardcore}
 Let $G = (V,E)$ be a graph on $n$ vertices with maximum degree at most $\Delta \ge 3$. Let $\mu$ denote the distribution on $2^{V}$ corresponding to the hardcore model on $G$ with fugacities $(\lambda_{v})_{v\in V}$. Then, for any $\Lambda \subseteq V$, $v\in V$, and any partial configuration $\sigma_{\Lambda} \in 2^{\Lambda}$ which leaves all neighbors of $v$ unoccupied, we have that
 \begin{equation}\label{eqn:hardcore-probs}
 \frac{\lambda_v}{1+\lambda_v}\cdot \prod_{w \in N(v)}\frac{1}{1+\lambda_w} \leq \P_{\mu}{v\in \sigma \mid \sigma_{\Lambda}} \leq \frac{\lambda_v}{1+\lambda_v}.
 \end{equation}

Moreover, if there exists $\delta \ge 0$ such that $\max_{v\in V}\lambda_{v} \leq (1-\delta)\lambda_{\Delta}$, then:
\begin{enumerate}
    \item For any site $v \in V$, 
    \[ e^{-3e^2} \le \prod_{w \in N(v)} \frac{1}{1 + \lambda_w}. \]
    \item In the same setting as \cref{eqn:hardcore-probs}, we have 
    \[ \frac{\P_{\mu}{v \in \sigma \mid \sigma_{\Lambda}}}{\P_{\mu}{v \notin \sigma \mid \sigma_{\Lambda}}} \le \lambda_v.\] 
    \item Provided $\delta > 0$, $\mu$ is $(O(1/\delta), \delta/2)$-completely spectrally independent. 
\end{enumerate}

\end{proposition}
\begin{proof}
Conclusion 2 follows directly from \cref{eq:hardcore-conditional}. 

For Conclusion 3, observe that under the assumption on $\delta$, we have $$\lambda_v (1+\delta/2) \leq \lambda_{\Delta} (1-\delta)(1+\delta/2) \leq \lambda_{\Delta}(1-\delta/2)$$
and so if $\delta > 0$, conclusion 3 follows from \cite[Lemma~8.4]{chen2021rapid}. 

Next, we prove \cref{eqn:hardcore-probs}. By \cref{eq:hardcore-conditional}, for any $\Lambda \subseteq V$ and arbitrary valid partial assignment $\sigma_{\Lambda}$ on $\Lambda$, 
\[ \P_{\mu}{v \in \sigma \given \sigma_{\Lambda}} = \P{ N(v) \cap \sigma = \emptyset \given \sigma_{\Lambda}} \frac{\lambda_v}{1+\lambda_v} \leq \frac{\lambda_v}{1+\lambda_v}, \]
which establishes the upper bound in \cref{eqn:hardcore-probs}. For the lower bound, observe from \cref{eq:hardcore-conditional}, \cref{eq:lambda-delta-bound}, and the inequality $1 + x \le e^x$ that
\[ \P_{\mu}{N(v) \cap \sigma  = \emptyset \given \sigma_{\Lambda}}  
\geq \prod_{w \in {N}(v)} \frac{1}{1 + \lambda_w} \geq  \parens*{1 + \frac{3e^2}{\Delta}}^{-\Delta} \geq e^{-3e^2},\]
where the last inequality holds under the assumption $\max_v \lambda_v \le \lambda_{\Delta}$. The first inequality in the above display equation establishes the lower bound in \cref{eqn:hardcore-probs} while the last inequality establishes conclusion 1. \qedhere

\end{proof}

For later use, we record the following result due to Erbar, Henderson, Menz, and Tetali \cite{RicciMLSIhardcore} (see also \cite{dai2013entropy}), which establishes a modified log-Sobolev inequality for the hardcore model with fugacity at most $1/2\Delta$. This constitutes \cref{eqn:mlsi-easy} for the hardcore model, provided that $\theta$ is a sufficiently small constant. 

\begin{proposition}[{\cite[Theorem~4.5 and Corollary~4.7]{RicciMLSIhardcore}}]\label{prop:hardcore-easy}
Let $\mu$ denote the hardcore model on a graph $G$ of maximum degree at most $\Delta \geq 3$, with fugacity $\lambda_v \leq 1/2\Delta$ for all sites $v$, and let $P$ denote the transition matrix of the Glauber dynamics. Then, the modified log-Sobolev constant $\rho_0(P)$ satisfies $\rho_0(P) \geq 1/4n.$ 
\end{proposition}

\subsection{Field dynamics with interleaved systematic scans}

Recall that a single pass of the \emph{systematic scan} (see, e.g., \cite{dyer2006systematic}) chain $P^{\operatorname{SS}}$ proceeds as follows: for $v$ looping through a fixed arbitrary ordering of the vertices $V$, the spin $\sigma_v$ at site $v$ is resampled conditional on the configuration at all other sites. This is a Markov chain which preserves the stationary distribution; one difference compared to the usual Glauber dynamics is that this chain is not typically reversible.  

The following proposition shows that for an arbitrary probability measure $\nu_0$ supported on independent sets, the measure $\nu_0 P^{\operatorname{SS}}$ is $O(1)$-bounded with respect to the hardcore measure $\mu$. This property of the systematic scan will be used as a subroutine in our sampling algorithm, to ensure that intermediate distributions lie in the region of applicability of our restricted entropy contraction inequality.  
\begin{proposition}\label{lem:systematic-scan-hardcore}
There exist absolute constants $C, C' > 0$ such that the following is true. 
Let $G = (V,E)$ be a graph on $n$ vertices with maximum degree at most $\Delta \ge 3$. 
Let $\mu$ be the hardcore model on $G$ with fugacities $(\lambda_v)_{v \in V}$ and suppose
$\lambda_v \le \lambda_{\Delta}$ as defined in \cref{eq:hardcore-uniqueness}.
Let $\nu_0$ be an arbitrary probability measure on $2^{V}$ supported on independent sets. Let $P^{\operatorname{SS}}$ denote the Markov
operator corresponding to a single pass of the systematic scan chain for $\mu$,
and let $\nu := \nu_0 P^{\operatorname{SS}}$. Define
\[ \gamma_{v} := \max_{\sigma^-} \frac{\nu( \sigma^+)}{\nu(\sigma^-)}, \]
where the maximum ranges over independent sets $\sigma^-$ in $V$ which do not include $v$ and where $\sigma^+$ denotes the set $\sigma^{-} \cup \set{v}$ (note that $\sigma^+$ may not be an independent set). 

Then $\gamma_v \le C \lambda_v$ for all $v \in V$. As a consequence, $\nu$ is $C'$-completely bounded with respect to $\mu$, i.e.
$\nu \in \mathcal{V}^{c}(C',\mu)$.
\end{proposition}
\begin{proof}
First we show $\gamma_v \le C \lambda_v$ for all $v \in V$.
This is equivalent to $\nu(\sigma^+) \le C \lambda_v \nu(\sigma^-)$ for all independent sets $\sigma^-$ which do not include $v$. By linearity of expectation, it suffices to prove that if $P := P^{\operatorname{SS}}$ is the systematic scan chain, that for any configuration $x$ we have $P(x \to \sigma^+) \le C \lambda_v P(x \to \sigma^-)$. 

Let $u_1,\ldots,u_k$ (for $0 \le k \le \Delta$) denote the neighbors of site $v$ which are updated after site $v$ in the systematic scan. If $P(x \to \sigma^-) = 0$, this implies that $\sigma^-$ is not an independent set and so $P(x \to \sigma^+) = 0$ as well. Otherwise, if $P(x \to \sigma^-) > 0$ then from the definition of the systematic scan, we have
\[ \frac{P(x \to \sigma^+)}{P(x \to \sigma^-)} 
= \frac{\P_{\mu}{v \in \sigma \mid \sigma^+(v)}}{\P_{\mu}{v \notin \sigma \mid \sigma^-(v)}} \prod_{i = 1}^k \frac{\P_{\mu}{{u_i} \notin \sigma \mid \sigma^+(u_i)}}{\P_{\mu}{{u_i} \notin \sigma \mid \sigma^-(u_i)}} \]
where we let the notation $\sigma^-(u_i)$ denote the configuration right before the systematic scan updates site $u_i$, given that for all previous sites it updated according to configuration $\sigma^-$, and similarly for the notation $\sigma^+(u_i)$. Note that by this definition, we have $\sigma^+(v) = \sigma^-(v)$ since $\sigma^+$ and $\sigma^-$ disagree only at site $v$. In particular, by \cref{prop:hardcore} we have
\[ \frac{\P_{\mu}{v \in \sigma \mid \sigma^+(v)}}{\P_{\mu}{v \notin \sigma \mid \sigma^-(v)}} \le \lambda_v. \]
Also, from \cref{eq:hardcore-conditional} we have $\P_{\mu}{{u_i} \notin \sigma \mid \sigma^+(u_i)} = 1$, since the neighboring site $v$ will be occupied when site $u_i$ is being resampled, and $\P_{\mu}{{u_i} \notin \sigma \mid \sigma^-(u_i)} \ge 1/(1 + \lambda_{u_i})$. From \cref{prop:hardcore} we therefore have that
\[ \prod_{i = 1}^k \frac{\P_{\mu}{{u_i} \notin \sigma \mid \sigma^+(u_i)}}{\P_{\mu}{{u_i} \notin \sigma \mid \sigma^-(u_i)}} \le \prod_{i = 1}^k (1 + \lambda_{u_i}) \le e^{3e^2} \]
and combining our estimates proves the result with $C = e^{3e^2}$. 

Now we verify that $\gamma_v \le C \lambda_v$ implies $\nu$ is $C'$-completely bounded with respect to $\mu$. Let $\Lambda \subseteq V \setminus \set{v}$ and  let $\sigma_{\Lambda} \in 2^{\Lambda}$ be an arbitrary partial assignment which does not occupy any neighbor of $v$. Let the
external field be $(\theta_v)_{v\in V}\in [0,1]^V$ and let $\lambda'_v := \theta_v \lambda .$ Letting $\mathcal{I}_{v}$ denote the collection of independent sets not containing $V$, we have
\[  \max_{\sigma^- \in \mathcal{I}_v} \frac{(\theta\ast \nu)( \sigma^+)}{(\theta \ast \nu)(\sigma^-)} = \theta_v \max_{\sigma^- \in \mathcal{I}_v} \frac{\nu( \sigma^+)}{\nu(\sigma^-)} \leq C \lambda'_v      \]
thus
\[\max_{\sigma' \in 2^{\Lambda} } \frac{\P_{\theta \ast \nu}{ v \in \sigma \given \sigma_{\Lambda} = \sigma' }}{\P_{\theta \ast \nu}{ v \notin \sigma \given \sigma_{\Lambda} = \sigma' }} = \max_{\sigma' \in 2^{\Lambda} } \frac{\P_{\theta \ast \nu}{ v \in \sigma \land \sigma_{\Lambda} = \sigma' }}{\P_{\theta \ast \nu}{ v \notin \sigma \land \sigma_{\Lambda} = \sigma' }} \leq C \lambda'_v.\]
\cref{prop:hardcore} implies that for $\sigma_{\Lambda}$ not occupying any neighbors of $v$, $\P_{\theta \ast \mu}{ v \in \sigma \given \sigma_{\Lambda} } \geq  e^{-3e^2} \frac{\lambda'_v}{1 + \lambda'_v},$ thus
\[\frac{(\theta \ast \nu)^{\sigma_{\Lambda}}(v) }{(\theta \ast \mu)^{\sigma_{\Lambda}} (v)  (1 -(\theta \ast \nu)^{\sigma_{\Lambda}}(v)) } \leq C\lambda'_v \parens*{ e^{-3e^2}  \frac{\lambda'_v}{1 + \lambda'_v}}^{-1} \leq  C e^{3e^2} (1+\lambda'_v) \leq C e^{3e^2} (1 + e^2) =: C' \]
which (using the trivial inequality $1 - (\theta \ast \nu)^{\sigma_{\Lambda}}(v) \le 1$) verifies the $C'$-completely bounded property with $C' = C\cdot e^{3e^2}(1+e^2)$. 
\end{proof}

For an integer $m \geq 1$, let $\widetilde {\fieldMarkov}_{\theta,m}$ denote the `$m$-approximate' version of $\fieldMarkov_{\theta}$ (recall \cref{def:field-dynamics}), where instead of resampling exactly from the distribution $(\theta \ast \mu)^{\sigma_{[n]\setminus S}}$, we return the configuration obtained by running the Glauber dynamics with respect to $(\theta \ast \mu)^{\sigma_{[n]\setminus S}}$ for $m$ steps (say, starting from the empty configuration). 

We can now prove that the approximate field dynamics interleaved with systematic scan mixes rapidly, which naively gives a sampling algorithm running in $O(n\Delta \log^2(n))$ time. Later we show how to shave the factor of $\Delta$ to get a sublinear (if $\Delta = \omega(\log^{2}n)$) runtime.

\begin{theorem}\label{thm:hardcore-systematic}
Suppose $\mu$ is the $\delta$-unique hardcore model on $G = (V,E)$ with $\abs{V}= n$, and define the Markov chain with transition matrix $P = P^{\operatorname{SS}} \widetilde {\fieldMarkov}_{\theta,m}$ with $\theta = 1/10$ and $m = 10n\log(n T/\epsilon)$. Let $\nu_0$ be arbitrary and define $\nu_t := \nu_{t - 1} P$ inductively. Then
\[ \DTV{\nu_{T}, \mu} \le \epsilon \]
for all $T = \Omega_{\delta}(\log[n \DKL{\nu_0 \river \mu}/\epsilon])$.
\end{theorem}

\begin{proof}[Proof of \cref{thm:main hardcore}]
Let $\nu_0$ be an arbitrary probability distribution supported on the independent sets of $G$. Let $\nu^*_0 := \nu_0$ and for $t\geq 0$, inductively define
\[ \nu^*_{t + 1} := \nu^*_t P^{\operatorname{SS}} {\fieldMarkov}_{\theta}.\]
Then, by \cref{lem:systematic-scan-hardcore}, \cref{cor:restricted entropy contraction}, and the spectral independence conclusion from \cref{prop:hardcore} we have
\[ \DKL{\nu^*_{t + 1} \river \mu} \le (1 - \kappa) \DKL{\nu^*_t P^{\operatorname{SS}} \river \mu} = (1 - \kappa) \DKL{\nu^*_t P^{\operatorname{SS}} \river \mu P^{\operatorname{SS}}} \le (1 - \kappa) \DKL{\nu^*_t \river \mu} \]
where $\kappa = e^{-O(\sqrt{1/\log(1 + \delta)})} = e^{-O(1/\sqrt{\delta})}$. Here, the equality holds because $\mu$ is the stationary measure of the systematic scan $P^{\operatorname{SS}}$, and the last inequality is the data processing inequality.


Now, let $\theta = 1/10$. By \cref{prop:hardcore-easy} (although, a simple coupling argument suffices, see e.g.~\cite[Chapter~13.1]{levin2017markov}), we see that the $\epsilon$-mixing time of the Glauber dynamics with respect to $(\theta \ast \mu)^{\sigma_{[n]\setminus S}}$ is at most $5n\log(n/\epsilon)$. Therefore, for $m \geq 5n\log(n/\epsilon')$ and letting $\nu_{t + 1} := \nu_t P^{\operatorname{SS}} \widetilde{\fieldMarkov}_{\theta,m}$, 
we have by the triangle inequality for $\TV$ distance that 
\begin{align*} 
\DTV{\nu_{t + 1}, \nu^*_{t + 1}}
&\le \DTV{\nu_t P^{\operatorname{SS}} \widetilde{\fieldMarkov}_{\theta,m}, \nu_t P^{\operatorname{SS}} \fieldMarkov_{\theta}} + \DTV{\nu_t P^{\operatorname{SS}} \fieldMarkov_{\theta}, \nu^*_t P^{\operatorname{SS}} \fieldMarkov_{\theta}} \\
&\le \DTV{\nu_t P^{\operatorname{SS}} \widetilde{\fieldMarkov}_{\theta,m}, \nu_t P^{\operatorname{SS}} \fieldMarkov_{\theta}} + \DTV{\nu_t, \nu^*_t } \\
&\le \epsilon' + \DTV{\nu_t, \nu^*_t }.
\end{align*}
and so $\DTV{\nu_t,\nu^*_t} \le t \epsilon'$ for all $t$. Using the triangle inequality and Pinsker's inequality, we have
\[ \DTV{\nu_t,\mu} \le \DTV{\nu_t,\nu^*_t} + \DTV{\nu^*_t,\mu} \le t\epsilon + \sqrt{(1/2) \DKL{\nu^*_t \river \mu}} \le  t\epsilon' + \sqrt{(1/2) (1 - \kappa)^t \DKL{\nu^*_0 \river \mu}} \]
which proves the result (upon setting $\epsilon' = \epsilon/2t$). 
\end{proof}
\subsubsection{Sublinear-time sampling}
\cref{thm:hardcore-systematic} provides an algorithm for sampling from a distribution within $\epsilon$-TV distance of $\mu$ using $T = \Theta_{\delta}(n\log^{2}(n/\epsilon))$ steps, each of which consists of resampling the spin at a vertex from a hardcore model with fugacity $\lambda = \Theta(1/\Delta)$, conditioned on a given configuration of the neighbors of the vertex. \footnote{Selecting which vertices to `pin' in a given round of the field dynamics requires time $O(n)$; since we only need to do this $\Theta_{\delta}(\log(n/\epsilon))$ times, the contribution to the overall running time is of a lower order.} Naively, implementing each such step takes time $\Theta(\Delta)$, since the vertex must learn the configuration of its neighbors. Now, we show how to implement each such step in time $O(1)$ in expectation, the idea being that in order to update the spin of a vertex, one only needs to look at its neighbors with probability roughly $\lambda = \Theta(1/\Delta)$. We proceed to the details.  Without loss of generality, we may assume that $\Delta \geq 20$, since for $\Delta \leq 20$, the naive implementation already works in time $O(\Delta) = O(1)$.



\begin{proof}[Proof of implementing updates]
Throughout, we assume $\Delta\geq 20.$  
Let $S_1, \cdots, S_{T}$ with $T =  C_{\delta} n\log^2 (n/\epsilon)$ be the updating steps that the algorithm takes. When updating a given vertex $i,$ we sample $r \sim \text{Ber} (\frac{\lambda}{\lambda+1}):$ if $r = 1$, then we include $i$ in the independent set, provided that all neighbors of $i$ are unoccupied. Else, if $r = 0$, then $i$ is not included in the independent set. 

For the update at time $t$ ($1 \leq t \leq T$), let $Z_t$ denote the indicator that $r=1$, when we sample $r \sim \text{Ber} (\frac{\lambda}{\lambda+1}).$ Since we only need to look at the neighbors of the vertex updated at time $t$ if $Z_{t} = 1$, it follows that the expected time to carry out the update at time $t$ is

\[O(\Delta)\P{Z_{t} = 1} + O(1)\P{Z_{t} = 0} = O(\Delta) \frac{\lambda }{\lambda + 1} + O(1), \frac{1}{\lambda+1} = O(1)\]
where we have used that $\lambda \leq \lambda(\Delta) \leq e/\Delta$ for $\Delta \geq 20.$ Thus, by the linearity of expectation, the expected runtime of all of the update steps is $O(T)$. 

We can upgrade this to a high probability bound on the overall running time of the algorithm. Indeed, by the discussion in the first paragraph of this subsection, we see that for $T = C_{\delta}n\log^2(n/\epsilon)$, the total runtime of the algorithm is $O(T+\Delta \sum_{t=1}^T Z_t)$, where $Z_1,\dots, Z_T$ are i.i.d. Bernoulli random variables with expected value $\leq \frac{e}{\Delta}$. Therefore, by Chernoff's bound, we get that
\[\P*{ \sum_{t=1}^T Z_t \geq 2e \frac{T}{\Delta}} \leq \exp(-\Omega(T/\Delta)) = O\parens*{\frac{1}{(n/\epsilon)^{n\log (n/\epsilon)/\Delta}}} \]
Thus, except with probability $O\parens*{\frac{1}{(n/\epsilon)^{n\log (n/\epsilon)/\Delta}}}$, our algorithm runs in time $O( T).$ \qedhere

\end{proof}
\subsection{Balanced Glauber dynamics}
\label{sec:balanced-gd}

Fix $K > 1$. We consider the following sampling algorithm, which we call the $K$-balanced Glauber dynamics for a distribution $\mu$. At each time step, we keep track of a configuration $\sigma_t$ and a tuple $(N_t(v))_{v\in V}$, described below. We initial $N_{0}(v) = 0$ for all $v \in V$. For each $t\ge 1$, sample a vertex $I_t$ uniformly at random and update $\sigma_{t-1}$ at $I_t$ according to the  distribution $\mu$ conditioned on $\sigma_{t-1,-I_t}$. Let $\sigma_{t,0}$ be the resulting configuration. We define $N_{t,0}(v)=N_{t-1}(v)+1$ for each $v$ adjacent to $I_t$ and define $N_{t,0}(I_t)=0$. For all other vertices $u$, we define $N_{t,0}(u) = N_{t-1}(u)$. Then for $j\ge 1$, as long as there is a vertex $v$ with $N_{t,j-1}(v)>K\Delta$, we choose such a vertex with the smallest index (according to a fixed, but otherwise arbitrary ordering of the vertices) and resample $\sigma_{t,j-1}$ at $v$ according to the distribution $\mu$ conditioned on $\sigma_{t,j-1,-v}$ to form $\sigma_{t,j}$. We then define $N_{t,j}$ by increasing $N_{t, j-1}$ at the neighbors of $v$ by $1$, setting $N_{t,j}(v)=0$, and for all other vertices $u$, setting $N_{t,j}(u) = N_{t-1,j}(u)$. At $j_t$, when there are no vertices $v$ with $N_{t,j_t}(v)>K\Delta$, we let $\sigma_{t}=\sigma_{t,j_t-1}$ and $N_{t} = N_{t,j_t-1}$. 

Observe that $(\sigma_t,N_t)_{t\ge 0}$ forms a Markov chain, initialized from $\sigma_0 \sim \nu_0$ and $N_0(v) = 0$ for all $v \in V$. We also define another Markov chain $(\xi_t,N_t)_{t\ge 0}$ where $\xi_0$ is sampled from the hardcore distribution $\mu$.

Note that there is a natural coupling of these two chains by using the same choice of $I_t$ in each step (observe that this leads to all the $N_{t,i}$'s being the same for both processes). We denote the coupled process by $(\sigma_{t,i}, \xi_{t,i}, N_{t,i})$. Since $(\sigma_t,\xi_t,N_t)$ has a Markovian dependence on $(\sigma_{t,0},\xi_{t,0},N_{t,0})$, it follows from the data processing inequality that for all $t\geq 1$,
\begin{align}
\label{eqn:dpi-balanced-gd}
\DKL{(\sigma_t,N_t)\river(\xi_t,N_t)} 
&\le \DKL{(\sigma_{t,0},N_{t,0})\river(\xi_{t,0},N_{t,0})}.
\end{align}

The main property of the balanced Glauber dynamics is that for every vertex $v\in V$ and at every step, the number of updates to neighbors of $v$ since $v$'s last update is at most $K\Delta$. In the following, for a vertex $v\in V$, let $P_{v}$ denote the transition matrix of the Markov chain which updates a configuration $\sigma$ by resampling the vertex $v$ according to the distribution $\mu$, conditioned on $\sigma_{-v}$. 
\begin{lemma}\label{lem:balanced-bounded}
There exists a constant $C = C(K) > 0$ such that the following is true. Let $G = (V,E)$ be a graph on $n$ vertices with maximum degree at most $\Delta \ge 3$. 
Let $\mu$ be the hardcore model on $G$ with fugacities $(\lambda_v)_{v \in V}$ and suppose
$\lambda_v \le \lambda_{\Delta}$. Let $\nu_0$ denote an arbitrary distribution supported on the independent sets of $G$.

Let $v_1,\dots, v_t \in V$ denote any sequence of vertices (possibly repeated) such that (i) the total number of appearances of neighbors of a vertex $v$ since the last appearance of $v$ is at most $K\Delta$ and (ii) every vertex appears at least once. Then, $\nu := \nu_0 P_{v_1}\dots P_{v_t}$ is $C$-completely bounded with respect to $\mu$. 
\end{lemma}
\begin{proof}
This is very similar to the proof of the analogous property for the systematic scan, \cref{lem:systematic-scan-hardcore}, and uses the fact that the sequence of update locations is fixed.

Let $\nu_{x}$ denote the resulting measure on independent sets of $G$, starting from the initial distribution $\nu_0 = \mathbb{1}_x$. As in the proof of \cref{lem:systematic-scan-hardcore}, it suffices to show that for any vertex $v \in V$ and for any independent set $\sigma^{-}$ which does not include $V$, 
\[\frac{\nu_{x}(\sigma^+)}{\nu_{x}(\sigma^-)} \leq \tilde{C}(K)\lambda_{v}.\]

Since every vertex is updated at least once, we see that if $\nu_{x}(\sigma^-) = 0$, then $\nu_{x}(\sigma^+) = 0$ as well. Moreover, by a similar argument as in the proof of \cref{lem:systematic-scan-hardcore}, we get that
\begin{align*}
    \frac{\nu_x(\sigma^+)}{\nu_x(\sigma^-)} \leq \lambda_{v} \prod_{i = 1}^k (1+\lambda_{u_i}),
\end{align*}
where $u_1,\dots, u_k$ denote the neighbors of site $v$ which are updated after the final update to site $v$. 

By assumption, $k \leq K\Delta$. Therefore, we have that
\begin{align*}
    \frac{\nu_x(\sigma^+)}{\nu_x(\sigma^-)} \leq \lambda_{v} \parens*{1+\frac{3e^2}{\Delta}}^{K\Delta} \leq \tilde{C}(K) \lambda_{v},
\end{align*}
as desired. \qedhere

\end{proof}

The next lemma allows us to go from an approximate tensorization of entropy estimate for the distributions corresponding to $(\sigma_{t})_{t\geq 0}$ to an entropy contraction statement for the chain $(\sigma_t, N_t)_{t\geq 0}$. More precisely: 

\begin{lemma}\label{lem:tensorization-to-balance}
With the notation at the start of this subsection, let $\nu$ be the distribution of $\sigma_{t - 1}$ and let $f = d\nu/d\mu$ denote the corresponding density. Then, provided that the approximate tensorization of entropy estimate
\begin{equation}
\label{eqn:assumption-tensorization}
 \DKL{\nu\river\mu} \le \frac{1}{C} \sum_{v \in V} \E_{\mu}{\Ent_{v}(f)} 
\end{equation}
holds for some $C > 0$, we have 
\[
\DKL{(\sigma_{t,0},N_{t,0})\river(\xi_{t,0},N_{t,0})} \le \parens*{1-\frac{C}{n}}\DKL{(\sigma_{t-1},N_{t-1})\river(\xi_{t-1},N_{t-1})}.
\]
Consequently,
\[
\DKL{(\sigma_t,N_t)\river(\xi_t,N_t)} 
\le \parens*{1-\frac{C}{n}} \DKL{(\sigma_{t-1},N_{t-1})\river(\xi_{t-1},N_{t-1})}.
\]
\end{lemma}
\begin{proof}
Note that $N_{t,0}$ is determined from $N_{t-1}$ once we know $I_t$, and moreover, $I_t$ is chosen uniformly at random from $V$. Furthermore, given $N_{t,0}$ and $N_{t-1}$, we can recover $I_t$ uniquely. Thus, by the chain rule for KL divergence, we have that
\[
\DKL{(\sigma_{t,0},N_{t,0})\river(\xi_{t,0},N_{t,0})} = \frac{1}{n}\sum_{v\in V} \DKL{(\sigma_{t,0}\river\xi_{t,0})\given I_t=v}.
\]
Let $\nu$ be the distribution of $\sigma_{t-1}$ and $\mu$ be the distribution of $\xi_{t-1}$. Let $P_v$ denote be the transition matrix of the Markov chain which resamples $\sigma$ at $v$ according to the hardcore distribution $\mu$, and let $D_v$ denote the projection of a configuration on $V\setminus \set{v}$. Then 
\[
\frac{1}{n}\sum_{v\in V} \DKL{(\sigma_{t,0}\river\xi_{t,0})\given I_t=v} 
= \frac{1}{n}\sum_{v\in V} \DKL{\nu P_v\river\mu P_v} \le \frac{1}{n}\sum_{v\in V} \DKL{\nu D_v\river\mu D_v}.
\]
Next, we have that 
\begin{align*}
\frac{1}{n}\sum_{v\in V} \DKL{\nu D_v\river\mu D_v} 
&= \frac{1}{n} \sum_{v\in V}\sum_{\theta \in 2^{V \setminus \set{v}}} \nu(\theta)\log \frac{\nu(\theta)}{\mu(\theta)}.
\end{align*}
On the other hand, for $f=d\nu/d\mu$,
\begin{align*}
\frac{1}{n}\sum_{v\in V} \E_{\mu}{\Ent_{v}(f)}
&= \frac{1}{n}\sum_{v\in V} \sum_{\theta \in 2^{V \setminus \set{v}}}  \parens*{\nu(\theta_+) \log \frac{\nu(\theta_+)}{\mu(\theta_+)} + \nu(\theta_-)\log \frac{\nu(\theta_-)}{\mu(\theta_-)} - \nu(\theta)\log\frac{\nu(\theta)}{\mu(\theta)} }\\
&=\DKL{\nu\river\mu} - \frac{1}{n}\sum_{v\in V} \DKL{\nu D_v\river\mu D_v},
\end{align*}
where $\theta_+$ denotes $\theta \cup \set{v}$ and $\theta_-$ denotes $\theta$, viewed as an element of $2^V$. 
Thus, under the approximate entropy tensorization assumption \cref{eqn:assumption-tensorization}, we have that
\begin{align*}
\frac{1}{n}\sum_{v\in V} \DKL{(\sigma_{t,0}\river\xi_{t,0})\given I_t=v} &\le \DKL{\nu\river\mu} - \frac{1}{n}\sum_{v\in V}\E_{\mu}{\Ent_{v}(f)} \le \parens*{1-\frac{C}{n}}\DKL{\nu\river\mu}.
\end{align*}
Therefore,
\begin{align*}
\DKL{(\sigma_{t,0},N_{t,0})\river(\xi_{t,0},N_{t,0})} &\le \parens*{1-\frac{C}{n}}\DKL{\nu\river\mu} \\
&\le \parens*{1-\frac{C}{n}} \DKL{(\sigma_{t-1},N_{t-1})\river(\xi_{t-1},N_{t-1})}.
\end{align*}
The ``consequently'' part follows by combining the above estimate with \cref{eqn:dpi-balanced-gd}.
\end{proof}


Next, for $T\geq 1$, let $J_T :=\sum_{t\le T}j_t$, where we have used the notation in the first paragraph of this subsection. Thus, the number of updates needed to obtain $\sigma_T$, starting from $\sigma_0$, is $J_T + T$. Note that $J_T$ is a random variable, depending on the choice of the update sequence $\set{I_t}_{1\leq t \leq T}$. 
\begin{lemma}
\label{lem:number-of-steps-balanced}
Deterministically, $J_T\le T/(K-1)$.
\end{lemma}
\begin{proof}
We use a potential function argument. For $t, j \geq 0$, let $\Phi_{t,j} = \sum_{v\in V}N_{t,j}(v)$. By definition, $\Phi_{t,j} \geq 0$ for all $t,j \geq 0$ and $\Phi_{0,0} = 0$. 

In each sampling step of the form $(t,0)$, where the configuration is resampled at a uniformly chosen vertex $v \in V$, the potential function increases by at most $\Delta$. 

In each sampling step of the form $(t,j)$, $j \geq 1$, where the configuration is resampled due to the presence of a vertex $v$ satisfying $N_{t,j-1} > K\Delta$, the potential function decreases by at least $K\Delta - \Delta = (K-1)\Delta$. 

Note that, up to and including the sampling step $(T,j_T)$, there are $T$ steps of the first type and $J_T$ steps of the second type. Therefore, by the non-negativity of $\Phi_{T, j_T}$ and the above step-wise bounds, we must have
\[0 \leq \Delta \cdot T - (K-1)\Delta \cdot J_T,\]
as desired. 
\end{proof}

Combining all of the above, we have proved the following result, which recovers \cref{thm:main hardcore}.

\begin{theorem}
\label{thm:hardcore-detailed}
Let $\nu_{-1}$ be an arbitrary initial distribution supported on the independent sets of $G = (V,E)$ with $\card{V} = n$ and define $\nu_0 = \nu_{-1} P^{\operatorname{SS}}$ to be the distribution obtained by running a single round of the systematic scan. 
Suppose $\mu$ is the $\delta$-unique hardcore distribution on $G$ and let $(\sigma_t, N_t)$ be defined by the balanced Glauber dynamics, as above, starting from the initial distribution $\nu_0 \times (0)_{v\in V}$. Let $\nu_t$ be the marginal law of $\sigma_t$. 
Then, for any $\epsilon > 0$, 
\[ \DTV{\nu_{T}, \mu} \le \epsilon \]
for all $T = \Omega_{\delta, K}(n \log[n \DKL{\nu_{-1} \river \mu}/\epsilon])$. 
Moreover, the total number of vertices updated to arrive at $\sigma_T$ 
is deterministically $O_K(T)$.
\end{theorem}
\begin{proof}
By \cref{lem:balanced-bounded}, the distribution $\nu_{t}$ is $C(K)$-completely bounded with respect to $\mu$ for all $t\geq 0$. 

Hence, letting $f_{t}$ denote the density $d\nu_{t}/d\mu$, it follows from \cref{thm:comparison} (by using complete spectral independence (\cref{prop:hardcore}) and the MLSI for Glauber dynamics in the ``easy regime'' from \cref{prop:hardcore-easy}) that \cref{eqn:assumption-tensorization} holds with $C = \Omega_{\delta, K}(1)$.

Therefore, with $\sigma_t,\xi_t,N_t$ defined by the balanced Glauber dynamics, as above, it follows by iterating \cref{lem:tensorization-to-balance} that for all $t\geq 0$, 
\[
\DKL{\nu_{t} \river \mu} \leq \DKL{(\sigma_t,N_t)\river(\xi_t,N_t)} 
\le \parens*{1-\frac{C}{n}}^{t} \DKL{\nu_{-1} \river \mu}, 
\]
from which the first conclusion follows. The `moreover' part follows immediately from \cref{lem:number-of-steps-balanced}. \qedhere 


\end{proof}

\section{Sampling in the Ising model}
\label{sec:ising}
In this section, we show how to apply our techniques to the Ising model. 

Recall that the \emph{Ising model} on a graph $G = (V,E)$ with \emph{edge activity} $\beta$ and \emph{external field} $\lambda$ is the probability distribution $\mu$ over $\set{\pm 1}^V$ given by
\[ \mu(\sigma) \propto \lambda^{\card{\set{ i\given \sigma_i = +1 }}} \prod_{\set{i,j} \in E} \beta^{\1(\sigma_i = \sigma_j)}.  \]
Note that if $\lambda > 1$, the model is equivalent to one with $\lambda \le 1$ by swapping the roles of $+$ and $-$ in the spin system. For this reason, we will assume without loss of generality that $\lambda \le 1$ in what follows.

From the definition of the Ising model, the conditional law of a spin given all others is a function only of its neighbors and given explicitly by
\begin{equation}\label{eqn:conditional-ising}
\mu(\sigma_v = + \mid \sigma_{\sim v}) = \frac{\lambda \beta^{\Delta_v - s_v}}{\beta^{s_v} + \lambda \beta^{\Delta_v - s_v}} 
\end{equation}
where $s_v = s_v(\sigma_{\sim v}) = \card{\set{u \given v \sim u, \sigma_u = -}}$ counts the number of $-1$-neighbors of site $v$.

\subsection{Mixing and MLSI away from the uniqueness threshold}
In order to apply our main comparison theorem (\cref{thm:comparison}), we will need to establish the MLSI for the Ising model in an ``easy regime'', away from the uniqueness threshold. It was shown in \cite{chen2021rapid} that in the ``easy regime'' there exists a contractive coupling of the Glauber dynamics; this implies rapid mixing but not a good enough MLSI.
In this section, we show how to get the MLSI for the Ising model in the appropriate ``easy regime'' by appealing to a result of \textcite{marton2015logarithmic}.

Given a probability distribution $\mu$ on $\set{ \pm 1}^V$, define its \emph{Dobrushin matrix} $R\in \R^{V\times V}$ by 
\[ R_{ij} := \max_{\sigma_{\sim i,j}} \DTV{\mu(\sigma_i = \cdot \mid \sigma_{j} = +, \sigma_{\sim i,j}), \mu(\sigma_i = \cdot \mid \sigma_{j} = -, \sigma_{\sim i,j})}. \]
Recall that an upper bound on the sum of the rows of less than $1$ on the Dobrushin matrix implies the existence of a contractive coupling for the Glauber dynamics with the same constant \cite{levin2017markov}, and hence rapid mixing. The following lemma directly bounds the entries of the Dobrushin matrix in the Ising model.
\begin{lemma}\label{lem:ising-dobrushin-1}
For the Ising model with edge activity $\beta$ and external field $\lambda \le 1$, if $i$ and $j$ are neighbors and $\Delta_i$ is the degree of node $i$, then
\[ R_{ij} \le \lambda \abs{\beta^2 - 1} \max \set{\beta^{\Delta_i-2}, \beta^{-\Delta_i}} \]
and otherwise $R_{ij} = 0$.
\end{lemma}
\begin{proof}
The case where nodes $i$ and $j$ are not neighbors follows immediately from the Markov property (or equivalently, from the conditional law \eqref{eqn:conditional-ising}). Now, we consider the case where $i$ and $j$ are neighbors and let $\Delta_i$ denote the degree of node $i$. From \eqref{eqn:conditional-ising} we see that 
\begin{align*} 
R_{ij} 
= \max_{0 \le s < \Delta_i} \abs*{\frac{\lambda \beta^{\Delta_i - 2s}}{1 + \lambda \beta^{\Delta_i - 2s}} - \frac{\lambda \beta^{\Delta_i - 2s - 2}}{1 + \lambda \beta^{\Delta_i - 2s - 2}}}
&=  \max_{0 \le s < \Delta_i} \lambda \beta^{\Delta_i - 2s - 2}\abs*{\frac{\beta^2 - 1}{(1 + \lambda \beta^{\Delta_i - 2s})(1 + \lambda \beta^{\Delta_i - 2s - 2})}} \\
&\le \lambda\abs{\beta^2 - 1} \max_{0 \le s < \Delta_i} \beta^{\Delta_i - 2s - 2} \\
&= \lambda \abs{\beta^2 - 1} \max \set{\beta^{\Delta_i-2}, \beta^{-\Delta_i}}
\end{align*}
as desired.
\end{proof}
Based on the above, we get the following general result for Ising models which we will use to cover the ``easy regime'' in our results later. 
\begin{lemma}\label{lem:ising-dobrushin}
Let $\mu$ be the Ising model with edge activity $\beta$ and external field $\lambda \le 1$ on a graph with $n$ vertices and maximum degree $\Delta$. Suppose that $\lambda \Delta \abs{\beta^2 - 1} \max \set{ \beta^{\Delta - 2}, \beta^{-\Delta} } \le 1 - \delta$ for some $\delta > 0$. Then:
\begin{enumerate}
    \item The Ising model satisfies Dobrushin's uniqueness condition with gap $\delta$, and in particular there exists a $(1-\delta/n)$-contractive coupling of the Glauber dynamics in the Hamming metric. Explicitly, if $P$ is the Markov operator corresponding to the Glauber dynamics, then for any $x,y \in \set{\pm 1}^n$ we have
    \[ W^d_1(P(x,\cdot),P(y,\cdot)) \le (1 - \delta/n) d(x,y) \]
    where $d(x,y)$ is the Hamming metric on the hypercube and $W_1$ is the Wasserstein-1 distance with respect to this metric. 
    \item The Ising model $\mu$ satisfies the $W_1$ transport-entropy inequality with constant $C = (2n/\delta)/(2 - \delta/n)$, i.e. $W_1(\nu,\mu)^2 \le C \DKL{\nu \river \mu}$ for any probability measure $\nu$ absolutely continuous with respect to $\mu$.
    \item The measure $\mu$ satisfies approximate tensorization of entropy with constant $C = 1/(\alpha \delta^2)$ where $\alpha :=  \min_i \min_{\sigma_{\sim i}} \set{ \mu(\sigma_i = + \mid \sigma_{\sim i}), \mu(\sigma_i = - \mid \sigma_{\sim i})) }$.
\end{enumerate}
\end{lemma}
\begin{proof}
Note that by \cref{lem:ising-dobrushin-1}, the assumption shows that the $\ell_1$ norms of the rows of the Dobrushin matrix are most $1 - \delta$.
The first claim then follows by a direct application of classic results in Markov chain theory, see e.g. \cite{levin2017markov}. The second claim follows from the first claim and the main result of \cite{eldan2017transport}, see also \cite{djellout2004transportation}, which shows that the existence of a contractive coupling implies the transportation-entropy inequality.

We now show the last claim. Define the square symmetric matrix $R'$ with entries 
\[ R'_{ij} = \lambda \abs{\beta^2 - 1} \max \set{ \beta^{\Delta - 2}, \beta^{-\Delta} } \]
for $i \sim j$ and $R'_{ij} = 0$ otherwise, and observe that $0 \leq R_{ij} \le R'_{ij}$ entrywise, so that the operator norm of $R$ is dominated by the operator norm of $R'$. Moreover, since $R'$ is a symmetric matrix, its operator norm is equal to its largest eigenvalue in absolute value, and by Gershgorin's circle theorem, its largest eigenvalue has absolute value at most $1 - \delta$. The result now follows directly by applying \cite[Theorem~1.14]{marton2015logarithmic}.
\end{proof}
\subsection{Sharp results for worst-case external field}

In this section, we use the term \emph{$\delta$-unique Ising model} on a graph of maximum degree $\Delta \ge 3$ to mean that
\[ \beta \in \bracks*{\frac{\Delta - 2 + \delta}{\Delta - \delta}, \frac{\Delta - \delta}{\Delta - 2 + \delta}}. \]
This is the uniqueness threshold for the worst case choice of external field (which turns out to be $\lambda = 1$). We remind the reader that this setting was also considered in \cite{chen2021rapid,chen2021optimal} where mixing time bounds of the form $\min(O_{\delta}(n^2), \text{poly}_{\delta}(\Delta)n)$ were attained ($\text{poly}_{\delta}(\Delta)$ means that the degree of the polynomial is a function of $\delta$. Specifically, it is shown in \cite{chen2021optimal} that $\Delta^{O(1/\delta)}$ suffices). In \cref{thm:main ising}, which we prove in this subsection, we get the optimal $O_{\delta}(n \log(n))$ mixing time for the standard Glauber dynamics. 

We begin with the following proposition, whose proof follows from that of \cite[Theorem~1.4]{chen2021rapid} (the necessary modifications are noted in \cref{sub:appendix-ising}). 

\begin{proposition} \label{prop:ising flc}
Suppose $\mu$ is the $\delta$-unique Ising model on a graph $G = (V,E)$ with maximum degree at most $\Delta$, where $\Delta \geq 3$. Then $\mu$ is $\Omega(\delta)$-fractionally log concave.
\end{proposition}


Next, we need the following result about the $\delta$-unique Ising model with external field in an ``easier regime''. 

\begin{lemma} \label{prop:ising marginals}
Let $\delta \geq 0$ and suppose that $\mu$ is the $\delta$-unique Ising model on $G = (V,E)$ with maximum degree $\Delta \ge 3$, $\lambda = 1/800$. For any $R \subseteq V$ and valid partial configuration $\sigma_R$ on $R,$ the conditional distribution $ \mu^{\sigma_R}$ satisfies approximate tensorization of entropy with constant $O(1)$. Consequently, the modified log-Sobolev constant for the Glauber dynamics is $\Omega(1/n).$ 
\end{lemma}
\begin{proof}
This result follows from the third conclusion of \cref{lem:ising-dobrushin}. To verify the assumption of that lemma, we can check directly using calculus that
\[ (1/800)\Delta \abs{\beta^2 - 1} \max\set{\beta^{\Delta - 2},\beta^{-\Delta}} \le (24/800) \max\set{\beta^{\Delta - 2},\beta^{-\Delta}} \le 0.9 \]
provided $\beta \in [(\Delta - 2)/\Delta, \Delta/(\Delta - 2)]$. Given this, it only remains to show that the conditional marginal lower bound $\alpha$ is lower bounded by an absolute constant.
From the definition of the Ising model (see \eqref{eqn:conditional-ising}), for each vertex $v$ and partial assignment $\sigma$ of $V \setminus v$, the conditional law is
\[\frac{\P_{\mu} {v = +1 \given \sigma} }{\P_{\mu} {v = -1 \given \sigma} } = \lambda_v \beta^{\Delta_v - 2s}\]
where $s$ is the number of $-1$’s assigned by $\sigma$ to the neighborhood $\mathcal{N}_v$ of $v.$  Now, since $\beta \in [ \frac{\Delta -2}{\Delta} , \frac{\Delta}{\Delta-2}]$ we have
\[\frac{1}{3e^2} \leq \parens*{\frac{\Delta -2}{\Delta}}^{\Delta} \leq \beta^{\Delta_v} \leq \parens*{\frac{\Delta}{\Delta-2}}^{\Delta_v} \leq \parens*{\frac{\Delta}{\Delta-2}}^{\Delta} \leq 3 e^2.\]
Note that 
\[ \beta^{\Delta_v - 2s} \in\begin{cases} [\beta^{-\Delta_v}, \beta^{\Delta_v} ] &\text{ if } \beta \geq 1\\  [\beta^{\Delta_v}, \beta^{-\Delta_v}  ] &\text{ otherwise }  \end{cases} \] 
Thus \[ \frac{1}{3e^2 \times 500 }\leq \frac{\P_{\mu}{v = +1 \given \sigma} }{\P_{\mu}{v = -1 \given \sigma} } \leq \frac{3e^2}{500 }, \]
as desired. 
\end{proof}
\begin{proof}[Proof of \cref{thm:main ising}]
First, observe that if $\lambda \leq 1/800$ then the result follows from rapid mixing under Dobrushin's condition by the first case of \cref{lem:ising-dobrushin}, as previously used also in \cite{chen2021optimal,chen2021rapid}.

It remains to show that the Glauber dynamics for a $\delta$-unique Ising model with $\lambda \geq  \lambda_0 = 1/800$
has modified log-Sobolev constant $\Omega_{\delta}(\frac{1}{n})$. The proof is similar to the proof of \cref{thm:comparison}, except that instead of using \cref{lem:restricted mlsi to tensorization}, we use \cref{prop:ising marginals}. 

In more detail, \cref{prop:ising marginals} gives \cref{eq:entfac-easy} in \cref{lem:ent-tensorization} with $\theta = \lambda_0/\lambda$ and $C_1 = O(1/\delta^2)$.

Next, as in the proof of \cref{thm:comparison}, by combining \cref{prop:ising flc} and \cref{cor:restricted entropy contraction}, we get 
\cref{eq:entfac-field} in \cref{lem:ent-tensorization} with $\theta = \lambda_0/\lambda$ and $C_2 = \Omega(\lambda_0)^{O(1/\delta)}$. 

Therefore, by \cref{lem:ent-tensorization}, we obtain approximate entropy tensorization with $C = O_{\delta}(1)$ for the Glauber dynamics with respect to $\mu$, which shows that the modified log-Sobolev constant for the Glauber dynamics is $\Omega_{\delta}(1/n)$. 

Finally, by the lower bound on the conditional marginals of $\mu$, we see that configuration $\sigma $ on $V$ has $\mu(\sigma) \geq 2^{-O(n)}$. Combining this with the lower bound on the modified log-Sobolev constant and \cref{lem:MLSI to mixing}, we get that the $\epsilon$-mixing time of the Glauber dynamics is $O_{\delta}(n \log (n/\epsilon))$. 
\end{proof}


\subsection{Sampling antiferromagnetic Ising models in general uniqueness regime}
\label{sec:ising-antiferro}
In this subsection, we consider the problem of sampling from anti-ferromagnetic Ising models in the so-called up-to-$\Delta$ uniqueness regime (see \cref{sub:facts-2-spin} for definitions). For values of the external field other than $\lambda \neq 1$, this permits efficient sampling for a wider range of $\beta$.  

Our first result concerns the mixing time of the field dynamics interleaved with systematic scan.

\begin{theorem}
\label{thm:general-ising-ss}
Fix $\alpha \ge 0$. Let $(\beta,\lambda)$ be the parameters of the $\delta$-unique antiferromagnetic Ising model $\mu$ on a graph $G = (V,E)$ (with $\card{V} = n$) of maximum degree at most $\Delta$, with
\[ \beta \in \bracks*{\frac{\Delta}{\Delta + \alpha}, 1}. \]
Define the Markov chain with transition matrix $P = P^{\operatorname{SS}} \widetilde \fieldMarkov_{\theta,m}$ with $\theta = \delta^2/64$ and $m = 10n\log(nT/\epsilon)$. Let $\nu_0$ be arbitrary and define $\nu_t := \nu_{t - 1} P$ inductively. Then
\[ \DTV{\nu_{T}, \mu} \le \epsilon \]
for all $T = \Omega_{\delta,\alpha}(\log[n \DKL{\nu_0 \river \mu}/\epsilon])$.
\end{theorem}

Our next result concerns the mixing time of the balanced Glauber dynamics. We use the same notation and definitions for this process as in \cref{sec:balanced-gd}. 

\begin{theorem}
\label{thm:ising-anti-detailed}
Fix $\alpha \geq 0$. Let $(\beta,\lambda)$ be the parameters of the $\delta$-unique antiferromagnetic Ising model $\mu$ on a graph $G = (V,E)$ (with $\card{V} = n$) of maximum degree at most $\Delta$, with
\[ \beta \in \bracks*{\frac{\Delta}{\Delta + \alpha}, 1}. \]
Let $\nu_{-1}$ be an arbitrary starting distribution and define $\nu_0 = \nu_{-1} P^{SS}$ to be the distribution obtained by running a single step of the systematic scan. 
Let $(\sigma_t, N_t)_{t\geq 0}$ denote the states of the balanced Glauber dynamics with initial distribution $\nu_0$ and let $\nu_t$ denote the marginal law of $\sigma_t$. 

Also, let $\nu'_t$ be defined to be the distribution after applying $t$ steps of the standard Glauber dynamics to $\nu_0$. Then,
\[
\begin{cases}
\DTV{\nu'_T, \mu} \leq \epsilon & \text{ if }0 \le \lambda \leq \lambda_{\alpha}\\
\DTV{\nu_{T}, \mu} \} \leq \epsilon & \text{if } \lambda_{\alpha} \leq \lambda \leq 1
\end{cases}
\]
for all $T = \Omega_{\delta,\alpha}(n \log[n \DKL{\nu_{-1} \river \mu}/\epsilon])$. Moreover, the total number of vertices updated to arrive at $\sigma_T$ in the balanced Glauber dynamics is deterministically $O_K(T)$. 
\end{theorem}

We will discuss how to efficiently update vertices using Bernoulli factories in \cref{sec:sublinear-ising}.


\subsubsection{General facts about antiferromagnetic 2-spin systems}
\label{sub:facts-2-spin}
We will need some well-known facts about the uniqueness threshold for the antiferromagnetic Ising model (see, e.g., \cite{sinclair2014approximation,LLY14,chen2021rapid}). These results are typically stated in terms of a
larger class of models called \emph{general 2-spin systems}, whose definition we now recall. 
The 2-spin system on a graph $G = (V,E)$ with parameters $(\beta, \gamma,\lambda)$ is the probability distribution $\mu$ over spins $\sigma \in \set{\pm 1}^{V}$ given by
\[ \mu(\sigma) \propto \lambda^{\card{\set{i\given \sigma_i = +}}} 
\beta^{m^+(\sigma)} \gamma^{m^-(\sigma)} \]
where $m^+(\sigma) := \card{\set{\set{i,j} \in E \given \sigma_i = + = \sigma_j}}$ and $m^-(\sigma) := \card{\set{\set{i,j} \in E \given \sigma_i = - = \sigma_j}}$ count the number of monochromatic plus and minus edges, respectively.
Note that the Ising model is the special case $\beta = \gamma$.
By flipping all signs from $+$ to $-,$ we can exchange $\beta$ and $\gamma,$ so it is assumed w.l.o.g. that $\beta \leq \gamma.$ Most of these results concern the antiferromagnetic case, when  $\beta \gamma \leq 1.$ 

From the definition, we have (see e.g.~\cite[Equation~83]{chen2021rapid}) that the conditional law of $\sigma_v$ given a valid (i.e.~nonzero probability) partial assignment of the remaining spins $\sigma_{\sim v}$ is
\begin{equation}\label{eqn:conditional-2spin}
\mu(\sigma_v = + \mid \sigma_{\sim v}) = \frac{\lambda \beta^{\Delta_v - s_v}}{\gamma^{s_v} + \lambda \beta^{\Delta_v - s_v}} 
\end{equation}
where $s_v = s_v(\sigma_{\sim v}) = \card{\set{u \given v \sim u, \sigma_u = -}}$ counts the number of $-1$-neighbors of site $v$. 

An antiferromagnetic two-spin system is said to be \emph{up-to-$\Delta$ unique with gap $\delta$} if for all $1 \le d \le \Delta - 1$, we have the following \emph{$d$-uniqueness} condition: $\abs{F_d'(\hat x_d)} \le 1 - \delta$, where 
\[ F_d(x) := \lambda \parens*{\frac{\beta x + 1}{x + \gamma}}^d \]
is the univariate tree recursion and $\hat x_d$ is the unique fixpoint of $F_d$ \cite{LLY14}.


The following proposition characterizes the uniqueness region for antiferromagnetic 2-spin systems with soft interactions (the result in the reference covers hard interactions as well). 
\begin{proposition}[{\cite[Proposition 8.6]{chen2021rapid}}]\label{prop:uniqueness}
Suppose $\beta \gamma \le 1$ and $\beta > 0$.
Let $\bar{\Delta} = \frac{1+ \sqrt{\beta \gamma} }{1 - \sqrt{\beta\gamma} }.$ For all integers $d\geq 1$,
we have that:
\begin{enumerate}
    \item If $d \leq (1-\delta)\bar{\Delta}$ then $(\beta, \gamma, \lambda)$ is $d$-unique with gap $\delta$ for all $\lambda.$
    \item If $d \geq (1-\delta) \bar{\Delta},$ let $\zeta_{\delta}(d) = d(1-\beta\gamma) - (1-\delta) (1+\beta\gamma), $ 
    \[x_{1,\delta}(d) = \frac{ \parens*{\zeta_{\delta} (d) - \sqrt{\zeta_{\delta} (d)^2 - 4(1-\delta)^2\beta \gamma }} }{2(1-\delta) \beta} \text{   and   } x_{2,\delta}(d) = \frac{ \parens*{\zeta_{\delta} (d) + \sqrt{\zeta_{\delta} (d)^2 - 4(1-\delta)^2\beta \gamma }} }{2(1-\delta) \beta}  \]
    and for $i \in \set*{1,2}$ \[\lambda_{i,\delta}(d) = x_{i,\delta}(d) \parens*{\frac{x_{i,\delta}(d) + \gamma}{\beta x_{i,\delta}(d) + 1}}^d.\]
     Note that $\lambda_{1,\delta}(d)  \lambda_{2,\delta}(d) = (\frac{\gamma}{\beta})^{d+1}$ and $\lambda_{1,\delta}(d) \leq (\frac{\gamma}{\beta})^{(d+1)/2} \leq \lambda_{2,\delta}(d).$
     
       Then $(\beta, \gamma, \lambda)$ is $d$-unique with gap $\delta$ iff $\lambda \in [0, \lambda_{1,\delta}(d)) \cup (\lambda_{2, \delta}(d), +\infty).$ 


\end{enumerate}
\end{proposition}



The next result is Claim 25 of \cite{LLY14} when $\delta = 0$; we repeat the proof in \cref{sub:appendix-ising} to get the desired
dependence on the gap parameter $\delta$.
\begin{lemma}[{\cite[Claim 25]{LLY14}}] \label{prop:lambda bound}
For $ d\geq (1-\delta)\bar{\Delta}$ and under the setup of case 2 of \cref{prop:uniqueness}, we have
\begin{equation}
    \frac{(1-\delta) \gamma^{d+1}}{\zeta_{\delta} (d) }\leq \lambda_{1,\delta}(d)\leq  \frac{2(1-\delta) \gamma^{d+1}}{\zeta_{\delta } (d)} \parens*{\frac{d+1-\delta}{d-1+\delta}}^d \leq \frac{9 \gamma^{d+1}}{\sqrt{\beta \gamma} }
    \end{equation}
    In particular, if $\beta \gamma\geq \frac{\Delta}{\Delta + \alpha}$ with $\Delta \geq d+1$ then
    \[ (1-\delta) C^{-1} \leq \beta^{d+1}\lambda_{1,\delta}(d) \leq \frac{\lambda_{1,\delta}(d)}{\gamma^{d+1}} \leq  C \]
    for some $ C > 0.$
\end{lemma}

Based on these facts, we prove some estimates on the conditional marginals of a two-spin model in the uniqueness regime. The proof of the following proposition is deferred to \cref{sub:appendix-ising}.
\begin{proposition} \label{prop:2 spin soft constraint}
 Let $G = (V,E)$ be a graph on $n$ vertices with maximum degree at most $\Delta \ge 3$. Let $\mu$ denote the distribution on $\set{\pm 1}^{V}$ corresponding to the $(\beta, \gamma,\lambda)$-model on $G$ with $\beta \gamma \le 1$ which is up-to-$\Delta$ unique with gap $\delta \ge 0$. If $\beta > 0$, let $\tau = \frac{1}{\beta \gamma}.$
 Then, for any $v\in V$, $\Lambda \subseteq V\setminus \set*{v}$, and valid partial configuration $\sigma_{\Lambda} \in \set{\pm 1}^{\Lambda}$
 \begin{enumerate}
     \item If $\beta > 0$, then for any $\theta = (\theta_v)_{v\in V} \in (0,1]^V$
     \[ \theta_v \lambda \beta^{\Delta_v} \leq  \frac{\P_{\theta \ast \mu}{\sigma_v = + \mid \sigma_{\Lambda}}}{\P_{\theta \ast \mu}{\sigma_v = - \mid \sigma_{\Lambda}}} \le \frac{\theta_v \lambda}{\gamma^{\Delta_v}}\leq \theta_v \lambda \beta^{\Delta_v} \tau^{\Delta}\]
     \item Fix an arbitrary $u\in V$, neighbor $v$ of $u$ and assignment $\sigma^- \in \set{\pm 1}^{V}$ with $\sigma_v = -$. Let $\sigma^+$ denote the same assignment with $\sigma_v$ flipped to  $+$. 
    If $\beta > 0$, then for any $\theta = (\theta_v)_{v\in V} \in (0,1]^V$
     \[\frac{1}{\tau} \leq \frac{\P_{\theta \ast \mu}{\sigma_u = + \mid \sigma^+_{\sim u}} }{\P_{\theta \ast \mu}{\sigma_u = + \mid \sigma^-_{\sim u} }} \leq 1 \quad \text{ and } \quad 1\leq \frac{\P_{\mu}{\sigma_u = - \mid \sigma^+_{\sim u}} }{\P_{\mu}{\sigma_u = - \mid \sigma^-_{\sim u} }} \leq \tau  \]
    \item Provided $\delta > 0,$ $\mu$ is $(O(1/\delta), \delta/2)$-completely spectrally independent.
 \end{enumerate}
 \end{proposition}

 \subsubsection{Analysis of mixing times}

We begin with the proof of \cref{thm:general-ising-ss}. 



As before, the utility of the systematic scan is that it always outputs a distribution which is $O_{\alpha}(1)$-completely bounded with respect to the Ising measure. 
\begin{proposition}\label{lem:systematic-scan-ising}
For arbitrary $\alpha \ge 0$, there exist $C, C' > 0$ depending only on $\alpha$  such that the following is true. 
Let $G = (V,E)$ be a graph on $n$ vertices with maximum degree at most $\Delta \ge 3$. 
Let $\mu$ be the antiferromagnetic Ising model on $G$ with up-to-$\Delta$ unique parameters $(\beta,\lambda)$, where
\[ \beta \in \bracks*{\frac{\Delta}{\Delta + \alpha}, 1}. \]
Let $\nu_0$ be an arbitrary probability measure on $\set{\pm 1}^{V}$. Let $P^{\operatorname{SS}}$ denote the Markov
operator corresponding to a single pass of the systematic scan chain for $\mu$,
and let $\nu := \nu_0 P^{\operatorname{SS}}$. Define
\[ \gamma_{v} := \max_{\sigma^-} \frac{\nu( \sigma^+)}{\nu(\sigma^-)}, \]
where the maximum ranges over all $\sigma^- \in \set{\pm 1}^V$ with $\sigma^-_v = -1$, and where $\sigma^+$ denotes $\sigma^-$ with the value of $v$ flipped to $+1$. 

Then $\lambda/C \le \gamma_v \le C \lambda$ for all $v \in V$. As a consequence, $\nu$ is $C'$-completely bounded with respect to $\mu$, i.e.
$\nu \in \mathcal{V}^{c}(C',\mu)$.
\end{proposition}
\begin{proof}
From \cref{prop:2 spin soft constraint}, we have that
     \[ \lambda \beta^{\Delta_v} \leq  \frac{\P_{\mu}{\sigma_v = + \mid \sigma_{\Lambda}}}{\P_{\mu}{\sigma_v = - \mid \sigma_{\Lambda}}}\leq \lambda \beta^{\Delta_v}/\beta^{2\Delta} \le e^{2\alpha} \lambda \beta^{\Delta_v} \]
     and also that for every node $v$ with neighbor node $u$ and every partial assignment $\sigma_{\sim u}^+,\sigma_{\sim u}^- \in \set{ \pm 1}^{V \setminus \set{u}}$ differing only in the assignment at site $v$,
     \[\frac{1}{\beta^2} \leq \frac{\P_{\mu}{\sigma_u = + \mid \sigma^+_{\sim u}} }{\P_{\mu}{\sigma_u = + \mid \sigma^-_{\sim u} }} \leq 1 \quad \text{ and } \quad 1\leq \frac{\P_{\mu}{\sigma_u = - \mid \sigma^+_{\sim u}} }{\P_{\mu}{\sigma_u = - \mid \sigma^-_{\sim u} }} \leq \beta^2.  \]
    Using this, and comparing the probability of trajectories as in the proof of \cref{lem:systematic-scan-hardcore}, we get for any initial state $x$,
    \[P^{\operatorname{SS}}(x \to \sigma^+)/P^{\operatorname{SS}}(x \to \sigma^-) \in [e^{-O(\alpha)},e^{O(\alpha)}],\] 
    which gives the desired bound on $\gamma_v$. 
    
    Given this and \cref{prop:2 spin soft constraint}, the proof of $C'$-complete boundedness is similar to \cref{lem:systematic-scan-hardcore}. In a bit more detail, as in the proof of \cref{prop:2 spin soft constraint}, the bound on $\gamma_{v}$ obtained above shows that for any valid partial configuration $\sigma_{\Lambda}$ on $\Lambda \subseteq V$ and for any $\theta \in (0,1]^{V}$,
    \[\frac{\P_{\theta \ast \nu}{\sigma_v = +1 \mid \sigma_{\Lambda}}}{\P_{\theta \ast \nu}{\sigma_v = -1 \mid \sigma_{\Lambda}}} \leq C_{\alpha}\theta_{v}\lambda.\]
    Moreover, from the first conclusion of \cref{prop:2 spin soft constraint} and the constraint on $\beta$, we have that for any valid partial configuration $\sigma_{\Lambda}$ on $\Lambda \subseteq V$ and for any $\theta \in (0,1]^{V}$,
     \[\frac{\P_{\theta \ast \mu}{\sigma_v = -1 \mid \sigma_{\Lambda}}}{\P_{\theta \ast \mu}{\sigma_v = +1 \mid \sigma_{\Lambda}}} \leq \frac{\tilde{C}_{\alpha}}{\theta_{v}\lambda}.\]
     
     Combining this with the previous estimate, we see that $\nu$ is $C' = C_{\alpha}\tilde{C}_{\alpha}$-completely bounded with respect to $\mu$, as desired. 
    \end{proof}
From this, we get rapid mixing of the field dynamics interleaved with systematic scan. 

\begin{proof}[Proof of \cref{thm:general-ising-ss}]
The proof is the same as the proof of \cref{thm:hardcore-systematic} with the appropriate replacements: instead of \cref{lem:systematic-scan-hardcore}, we use \cref{lem:systematic-scan-ising} to obtain $C$'-boundedness; instead of \cref{prop:hardcore}, we use \cref{prop:2 spin soft constraint} to obtain complete spectral independence; instead of \cref{prop:hardcore-easy}, we use \cref{lem:ising-dobrushin} to obtain an $O_{\alpha, \delta}(n\log(n/\epsilon))$ bound on the $\epsilon$-mixing time in the ``easier regime'' (which our choice of $\theta$ results in, again due to \cref{lem:ising-dobrushin}).     
\end{proof}

Finally, we prove rapid mixing of the (balanced) Glauber dynamics. 

\begin{proof}[Proof of \cref{thm:ising-anti-detailed}]
We split the analysis into two cases. 
If $\lambda \le (1/800) e^{-16\alpha} \lambda_{1,0}(\Delta)$ then by \cref{lem:ising-dobrushin}, there exists a $(1 - 1/10n)$-contractive coupling and the mixing of the standard Glauber dynamics follows immediately,
see e.g. \cite{levin2017markov}. 

Otherwise, we use by \cref{prop:lambda bound} that, due to the restriction on $\beta$, $\lambda_{1,0}$ is lower bounded by a constant depending only on $\alpha$. Then, as in the proof for the $\lambda=1$ case before, we can appeal
to the last part of \cref{lem:ising-dobrushin} 
to get approximate tensorization of entropy for the Glauber dynamics in the ``easy regime'', with a constant that depends only on $\alpha$ and $\delta$.  

From here, the proof follows the same strategy as the proof of \cref{thm:hardcore-detailed}: (1) based on essentially the same argument as in \cref{lem:balanced-bounded} (with appropriate replacements coming from \cref{prop:2 spin soft constraint}, as in the proof of \cref{lem:systematic-scan-ising}), we get that the distributions $\nu_t$ are $C'$-bounded with respect to $\mu$ for $C' = C'(\alpha,\delta) > 0$, and then (2) we have entropy contraction based on approximate tensorization of entropy in the easy regime (which we showed in the previous paragraph), complete spectral independence (\cref{prop:2 spin soft constraint}), and the comparison between field dynamics and Glauber dynamics (\cref{thm:comparison}).  
\end{proof}

\subsection{Sublinear-time sampling}
\label{sec:sublinear-ising}
In the algorithms considered above, resampling a vertex $v$ according to the distribution $\mu$ condition on $\sigma_{-v}$ naively takes $O(\Delta)$ time. Here, we show how to implement these steps in $O(t_S)$ expected time, where $t_S$ is an upper bound on the time taken to sample a uniformly random neighbor of a given vertex (note that for an adjacency array representation of the graph, $t_S = O(1)$).  The key observation is that, in $t_S$ time, we can sample a random neighbor of a vertex and observe its spin. Combining this observation with results from the literature on Bernoulli factories, we then show how to (perfectly) implement a resampling step with only $O(1)$ (in expectation) independent samples of a random neighbor of a vertex, which takes expected time $O(t_S)$.

\begin{proof}[Proof of implementing updates]
We first describe the approach in the antiferromagnetic setting $\beta < 1$. Recall that we can assume $\lambda \le 1$ without loss of generality (by flipping the role of $+$ and $-$).

Let
\[ S(x) := \frac{e^x}{e^x + 1} = \frac{1 + \tanh(x/2)}{2}  \]
be the usual sigmoid function (i.e. a recentered/rescaled version of $\tanh$). 
Observe that $\lambda^{1(\sigma_i = 1)} = e^{\sigma_i (\log \lambda)/2 + (\log \lambda)/2}$. A similar manipulation shows that the marginal law of the spin $\sigma_i$ at site $i$, conditioned on the spins at the other sites, is given by
\begin{align*}
\mu[\sigma_i = +1 \mid \sigma_{\sim i}] 
&= \frac{\lambda \beta^{\card{\set{j \given i \sim j, \sigma_j = +1}}}}{\lambda \beta^{\card{\set{j \given i \sim j, \sigma_j = +1}}} + \beta^{\card{\set{j \given i \sim j, \sigma_j = -1}}}} \\
&= \frac{\lambda \beta^{2\card{\set{j \given i \sim j, \sigma_j = +1}} - \Delta_i}}{\lambda \beta^{2\card{\set{j \given i \sim j, \sigma_j = +1}} - \Delta_i} + 1} \\
&= S\parens*{\log \lambda + (\log \beta)[2\card{\set{j \given i \sim j, \sigma_j = +1}} - \Delta_i]} \\
&= S\parens*{\log \lambda + \Delta_i(\log \beta)\bracks*{2 \frac{\card{\set{j \given i \sim j, \sigma_j = +1}}}{\Delta_i} - 1}}
\end{align*}
where $\Delta_i$ is the degree of node $i$. Next, consider a 0/1-valued random variable $Y$ which is generated in the following way: with probability $1/2$ we have $Y \sim \text{Ber}(1/2)$; otherwise, a neighbor $j$ of $i$ is selected uniformly at random and $Y = 1$ if $\sigma_j = 1$ and $Y = 0$ if $\sigma_j = -1$. Clearly the resulting law is $Y \sim \text{Ber}(p)$ with 
\[ p := \frac{1}{4} + \frac{\card{\set{j \given i \sim j, \sigma_j = +1}}}{2\Delta_i} \in \left[\frac{1}{4}, \frac{3}{4}\right].\]
In terms of this bias $p$, we can write the above expression as
\[ \mu[\sigma_i = +1 \mid \sigma_{\sim i}] = S\parens*{\log \lambda + \Delta_i(\log \beta)\bracks*{4p - 2}} = S\parens*{\log \lambda - 2 \Delta_i \log \beta + 4p \Delta_i \log \beta)}. \]
Letting $c_1 = \log \lambda - 2 \Delta_i \log \beta, c_2 = 4 \Delta_i \log \beta$, we therefore have that
\[ \mu[\sigma_i = +1 \mid \sigma_{\sim i}] = \frac{e^{c_1} e^{c_2 p}}{1 + e^{c_1} e^{c_2 p}}.\]

Next, note that since $\Delta/(\Delta + \alpha) \leq \beta < 1$, there exists an absolute constant $C_{\alpha}$ such that $e^{c_1} \leq C_{\alpha}$.
Therefore, by a result of Huber \cite[Lemma~1]{huber2017optimal}, it follows that if we can simulate $\text{Ber}(e^{c_2 p})$ in expected $O_{\alpha}(1)$ time, then we can also simulate $\text{Ber}(e^{c_1 + c_2 p}/(1+ e^{c_1 + c_2p}))$ in expected $O_{\alpha}(1)$ time. 


It remains to implement a coin of bias $e^{c_2 p}$. By Proposition 3.4 of \cite{latuszynski2011simulating} (see also the discussion after Proposition 3.4), there exists a Bernoulli factory to simulate $\text{Ber}(e^{c_2 p})$ with $O_{\alpha}(1)$-expected calls to $\text{Ber}(p)$ (here we use that because the model is antiferromagnetic, $c_2 < 0$, and that $\abs{c_2}$ is upper bounded by an constant depending only on $\alpha$ by the assumed lower bound on $\beta$) . 

For the ferromagnetic model, the approach is very similar: we simply flip the role of $-1$ and $+1$ in the neighbors of node $i$ when implementing the algorithm. (This corresponds to the fact that on a tree of depth 1, a ferromagnetic model can be trivially converted to an antiferromagnetic model using the same flip).  
\end{proof}

\section{Concentration of measure}
\label{sec:concentration}

Finally, we derive concentration of measure statements for the models concerned in the previous section using our (restricted) MLSIs. 

\begin{lemma}[Herbst lemma (see, e.g.,~\cite{van2014probability})] \label{lem:herbst}
Suppose that a probability measure $\nu$, random variable $f$, and $\Gamma,\sigma \ge 0$ are such that for all $\gamma \in [0,\Gamma]$, we have $\E_{\nu}{e^{\gamma f}} < \infty$ and
\[ \Ent_{\nu}[e^{\gamma f}] \le \frac{\gamma^2 \sigma^2}{2} \E_{\nu}{e^{\gamma f}}. \]
Then
\[ \log \E_{\nu}{e^{\gamma (f - \E_{\nu}{f})}} \le \frac{\gamma^2 \sigma^2}{2} \]
for all $\gamma \in [0,\Gamma]$. Therefore by the Chernoff bound, for all $t \ge 0$ we have
\[ \P_{\nu}{f > \E_{\nu}{f} + t} \le \exp\parens*{\inf_{\gamma \in [0,\Gamma]} \gamma^2 \sigma^2/2 - \gamma t } = 
\begin{cases}
e^{-t^2/2\sigma^2} & \text{ if $t/\sigma \le \Gamma \sigma$} \\
e^{\Gamma^2 \sigma^2/2} e^{-\Gamma t} & \text{otherwise} \\
\end{cases}
\]
\end{lemma}
\begin{proof}
The first part of the statement follows from the proof of \cite[Lemma~3.13]{van2014probability}. The second part is the Chernoff bound $\P{f - \E_{\nu}{f} > t} \le e^{-t \gamma} \E{e^{\gamma f}}$, and we used that the unconstrained minimum of the quadratic objective is achieved when $\gamma = t/\sigma^2$, at which point we have $\gamma^2 \sigma^2/2 - \gamma t = t^2/2\sigma^2 - t^2/\sigma^2 = -t^2/2\sigma^2$, so the constrained minimum is either there or at $\Gamma$.
\end{proof}
To interpret this, note that for a random variable with an exponential tail like $e^{-t/\sigma}$, the largest $\Gamma$ we can hope to apply the bound with is $\Gamma \approx 1/\sigma$, in which case the first bound applies for $t$ at most $\sigma$ and the second bound applies for $t$ larger. For other probability distribution, like $\text{Poi}(\lambda)$, the two bounds correspond as in Bernstein's inequality to the sub-Gaussian behavior of moderate deviations and the subexponential behavior of larger deviations.

The above lemma, the next two lemmas, and our restricted MSLI combined lets us get concentration bounds for functions in terms of the maximum size of their ``gradient''. 
\begin{lemma}[Proof\footnote{Note that in the statement of \cite[Lemma~5]{hermon2019modified}, there is a typo in (103) where $f(y) - f(x)$ from the definition of $v_1(f)$ is reversed compared to the proof. We have corrected this in the statement of \cref{lem:dirichlet1}. \cref{lem:dirichlet2} is the correct statement and proof  when the order of $f(y) - f(x)$ is reversed.} of Lemma 5, \cite{hermon2019modified}]\label{lem:dirichlet1}
Suppose $Q$ is the transition matrix of a reversible Markov chain on a finite set $\Omega$ with stationary measure $\nu$. Then
\[ \mathcal{E}_{Q}(e^{\gamma f}, \gamma f) \le \gamma^2 v_1(f) \E_{\nu}{e^{\gamma f}} \]
where
\[ v_1(f) := \max_{x \in \Omega} \sum_{y \in \Omega} P(x,y) [f(x) - f(y)]_+^2, \]
$[x]_+ = \max(x,0)$,
and $\mathcal{E}_P(f,g) := \E_{\omega \sim \nu}{[(I - P)f](\omega) g(\omega)}$ is the Dirichlet form. 
\end{lemma}
 The proof of this result is a variant of the proof of Lemma 5 of \cite{hermon2019modified}. We omit the details. 
\begin{lemma}\label{lem:dirichlet2}
Suppose $Q$ is the transition matrix of a reversible Markov chain on a finite set $\Omega$ with stationary measure $\nu$, and let $f$ be so that $f(y)-f(x) \le \kappa$ whenever $Q(x,y) > 0$. Then
\[ \mathcal{E}_{Q}(e^{\gamma f}, \gamma f) \le e^{\gamma \kappa} \gamma^2 v_2(f) \E_{\nu}{e^{\gamma f}} \]
where
\[ v_2(f) := \max_{x \in \Omega} \sum_{y \in \Omega} P(x,y) [f(y) - f(x)]_+^2, \]
$[x]_+ = \max(x,0)$,
and $\mathcal{E}_P(f,g) := \E_{\omega \sim \nu}{[(I - P)f](\omega) g(\omega)}$ is the Dirichlet form. 
\end{lemma}
\begin{proof}
From the definition we have
\begin{align*}
    \mathcal{E}(e^{\gamma f}, \gamma f) 
    &= \frac{\gamma}{2} \sum_{x,y} \nu(x) Q(x,y) (e^{\gamma f(x)} - e^{\gamma f(y)})(f(x) - f(y)) \\
    &= \gamma \sum_{x,y} \nu(x) Q(x,y) (e^{\gamma f(y)} - e^{\gamma f(x)})(f(y) - f(x))_+ \\
    &= \gamma \sum_{x,y} \nu(x) Q(x,y) e^{\gamma f(x)} (e^{\gamma [f(y) - f(x)]} - 1)(f(y) - f(x))_+
\end{align*}
where we used that, by reversibility, $\nu(x) Q(x,y) = \nu(y) Q(y,x)$ in order to pair up terms.  Using the assumption on $f$ 
and the inequality $1 - e^{-x} \le x$ gives
\[ (e^{\gamma [f(y) - f(x)]} - 1) \le e^{\gamma \kappa} (1 - e^{-\gamma[f(y) - f(x)]}) \le \gamma e^{\gamma \kappa} [f(y) - f(x)] \]
whenever $Q(x,y) > 0$ 
so
\begin{align*}
\MoveEqLeft \gamma \sum_{x,y} \nu(x) Q(x,y) e^{\gamma f(x)} (e^{\gamma [f(y) - f(x)]} - 1)(f(y) - f(x))_+ \\
&\le \gamma^2 e^{\gamma \kappa} \sum_{x,y} \nu(x) Q(x,y) e^{\gamma f(x)} (f(y) - f(x))^2_+ \\
&\le \gamma^2 e^{\gamma \kappa} \E_{\nu}{e^{\gamma f(X)}} \max_x \sum_y Q(x,y) (f(y) - f(x))_+^2
\end{align*}
which gives the desired conclusion.
\end{proof}
Recall that if $\nu$ satisfies an MLSI with constant $\alpha/2$, then
\[ \Ent_{\nu}[e^{\gamma f}] 
\le \frac{1}{\alpha} \mathcal{E}_P(e^{\gamma f}, \gamma f).  \]
In our applications, we only have a restricted MLSI, so that the above conclusion is not guaranteed to hold for all $\gamma$ and $f$. Therefore, since $\frac{\Ent_{\nu}[g]}{\E_{\nu}{g}} = \DKL{\frac{g}{\E_{\nu}{g}} \nu, \nu}$, when using our restricted MLSIs, we will need to check that $\frac{e^{\gamma f}}{\E_{\nu}{e^{\gamma f}}} \nu$ is in the set of ``valid'' measures for suitably large $\gamma$.



\paragraph{Concentration for the hardcore model.} 
For increasing Lipschitz functions, we get the following Bernstein-type bound, which is similar to the usual concentration bounds for sums of i.i.d. $\text{Ber}(p)$ random variables.
\begin{proposition}\label{lem:hardcore-monotone}
Let $\nu$ be the hardcore model on a graph of maximum degree at most $\Delta \ge 3$ and at fugacity $\lambda \le (1 - \delta) \lambda_{\Delta}$. There exists some constant $c = c(\delta) > 0$ such that the following is true.
Let $f$ be so that $0 \le f(\sigma_+)-f(\sigma_-) \le \kappa$ for all adjacent states $(\sigma_-,\sigma_+)$. For all $t \ge 0$ we have 
\[ \P_{\nu}{f -  \E_{\nu} f > t} \le \begin{cases}
e^{-c t^2/\lambda \kappa^2 n} & \text{ if $t = O_{\delta}( \lambda \kappa n)$} \\
e^{O_{\delta}(\lambda n)} e^{- \Omega_{\delta}(t/\kappa)} & \text{otherwise} \\
\end{cases}, \]
and
\[ \P_{\nu}{f -  \E_{\nu} f < -t} \le \begin{cases}
e^{-c t^2/\lambda \kappa^2 n} & \text{ if $t = O_{\delta}( \lambda \kappa n)$} \\
e^{O_{\delta}(\lambda n)} e^{- \Omega_{\delta}(t/\kappa)} & \text{otherwise} \\
\end{cases}. \]
\end{proposition}
\begin{proof}
We first prove the upper tail bound.
Let $\gamma$ be arbitrary and consider $g = e^{\gamma f}$. Then $g(\sigma_+)/g(\sigma_-)\le e^{\gamma \kappa}$. Therefore, from MLSI in the easier regime (\cref{prop:hardcore-easy}) and \cref{lem:restricted mlsi to tensorization}, we have that \cref{eq:entfac-easy} holds for $g$ with $C_1 = O(e^{\gamma \kappa})$. 

Moreover, similar to the proof of \cref{thm:comparison}, we see that \cref{eq:entfac-field} holds with $C_2 = \theta^{-O_\delta(1+\gamma \kappa)}$, for an absolute constant $\theta$.  

Thus, \cref{lem:ent-tensorization} implies that for $P$ the Glauber dynamics with respect to $\nu$, we have the following MLSI-type estimate 
\[
\Omega(\theta^{O_\delta(1+\gamma \kappa)}\ e^{-\gamma \kappa} n^{-1}) \Ent_\nu[g] \le \mathcal{E}_{P}(g,\log g).
\]
From the definition, we have $v_2(f) \le \lambda \kappa^2$. Thus by \cref{lem:dirichlet2}, we have
\[\mathcal{E}_P(g, \log g) \leq e^{\gamma \kappa}\gamma^{2}\lambda \kappa^{2} \E_{\nu}{g}.\]
Combining this with the previous estimate, we get that
\[
\Ent_\nu[g] \le \theta^{-O_\delta(1+\gamma \kappa)} e^{2\gamma \kappa} n \gamma^2 \lambda \kappa^2 \E_\nu{g}.
\]
Applying the Herbst argument, \cref{lem:herbst}, with $\Gamma = (1/\kappa) \Theta_{\delta}(1)$ and $\sigma^2 = \Theta(\lambda \kappa^2 n)$ gives the upper tail. 

For the lower tail, we repeat the same argument with $f$ replaced by $-f$ and \cref{lem:dirichlet2} replaced by \cref{lem:dirichlet1}.
\end{proof}
For arbitrary Lipschitz functions we get a slightly weaker sub-Gaussian bound.
\begin{proposition}\label{lem:hardcore-lipschitz}
Let $\nu$ be the hardcore model on a graph of maximum degree at most $\Delta \ge 3$ and at fugacity $\lambda \le (1 - \delta) \lambda_{\Delta}$.
Let $f$ be so that $\abs{f(\sigma_+)-f(\sigma_-)} \le \kappa$ for all adjacent states $(\sigma_-,\sigma_+)$. For all $t \ge 0$ we have 
\[ \P_{\nu}{f -  \E_{\nu} f > t} \le 
e^{-c t^2/\kappa^2 n} \]
for some $c = c(\delta) > 0$.
\end{proposition}
\begin{proof}
The proof is the same as \cref{lem:hardcore-monotone} except that we only have the bound $v_2(f) \le \kappa^2$. (Note that, because of the Lipschitz assumption, we have $\abs{f - \E_\nu f} = O(\kappa n)$ with probability $1$, so for $t$ at least $\kappa n$, the bound improves to $0$).   
\end{proof}

\paragraph{Concentration for the Ising model.} Since we proved the full Modified Log-Sobolev Inequality (MLSI) in \cref{prop:ising marginals} for the Ising model in the uniqueness regime, we directly get concentration of Lipschitz functions for the Ising model without a special argument. 
\begin{proposition}\label{prop:ising-lipschitz}
Let $\nu$ be the Ising model on a graph of maximum degree at most $\Delta \ge 3$ 
with edge activity $\beta \in \bracks*{\frac{\Delta - 2 + \delta}{\Delta - \delta}, \frac{\Delta - \delta}{\Delta - 2 + \delta}}$ and external field $\lambda > 0$. There exists $c = c(\delta) > 0$ such that the following result holds. 
Let $f$ be so that $\abs{f(\sigma_+)-f(\sigma_-)} \le \kappa$ for all adjacent states $(\sigma_-,\sigma_+)$. For all $t \ge 0$ we have 
\[ \P_{\nu}{f -  \E_{\nu} f > t} \le 
e^{-c t^2/\kappa^2 n} \]
\end{proposition}
\begin{proof}
Recall that by interchanging the role of $+$ and $-$, we can assume without loss of generality that $\lambda \le 1$. 
If $\lambda < 1/800$, this follows from the second part of \cref{lem:ising-dobrushin} by the same argument as we used for mixing in \cref{thm:main ising}, because the $W_1$ transport-entropy inequality is equivalent to sub-Gaussian concentration for Lipschitz functions with corresponding constant, see \cite{eldan2017transport,bobkov1999exponential}.

If $\lambda \ge 1/800$, then by \cref{prop:ising marginals} the Glauber dynamics satisfy a  modified log-Sobolev inequality with constant $\Omega_{\delta}(1/n)$ so by the standard application of the Herbst argument (see \cref{lem:herbst} and \cite{van2014probability}) we have the desired concentration of Lipschitz functions. 
\end{proof}
\begin{proposition}\label{prop:ising-anti-concentration}
Let $\alpha \ge 0$, let $(\beta,\lambda)$ be the parameters of the $\delta$-unique antiferromagnetic Ising model $\mu$ on graph $G = (V,E)$ where
\[ \beta \in \bracks*{\frac{\Delta}{\Delta + \alpha}, 1}. \]
There exists $c = c(\delta,\alpha) > 0$ such that the following result holds. 
Let $f$ be so that $\abs{f(\sigma_+)-f(\sigma_-)} \le \kappa$ for all adjacent states $(\sigma_-,\sigma_+)$. For all $t \ge 0$ we have 
\[ \P_{\nu}{f -  \E_{\nu} f > t} \le 
e^{-c t^2/\kappa^2 n} \]
\end{proposition}
\begin{proof}
As in the proof of \cref{thm:ising-anti-detailed}, we split into two cases. If $\lambda \le (1/800) e^{-16\alpha} \lambda_{1,0}$ then the result follows from the transport-entropy inequality from \cref{lem:ising-dobrushin} and its equivalence to subgaussian concentration of Lipschitz functions \cite{bobkov1999exponential}.
Otherwise, as in the proof of \cref{thm:ising-anti-detailed} we have a MLSI (and even entropy factorization) in the ``easy regime'' (i.e. for $\pi = ((1/800) e^{-16\alpha} \lambda_{1,0}/\lambda) * \mu$) by the last part of \cref{lem:ising-dobrushin},
combining \cref{thm:comparison} gives a restricted MLSI for $\mu$, and the result follows as in the proofs for the hardcore model (\cref{lem:hardcore-lipschitz}).
\end{proof}
\PrintBibliography

\appendix
\appendixpage
\addappheadtotoc
\section{Deferred proofs from \texorpdfstring{\cref{sec:comparison}}{Section \ref{sec:comparison}}}
\label{sec:appendix-numerical}
\begin{proof}[Proof of \cref{lem:ent-compare}]
Note that when we multiply $a$ and $b$ by a constant $c > 0,$ the inequality doesn't change.
Therefore, without loss of generality, we can assume that $b=1$. Let us consider $\eta>1$, the case $\eta\le 1$ is similar. 

Let 
\[
f(a; x,\eta) = \frac{1}{x+1}a\log a- \frac{a+x}{x+1}\log(\frac{a+x}{x+1}) - \max(\eta,\eta^{-1}) \parens*{\frac{1}{\eta x+1}a\log a - \frac{a+\eta x}{\eta x+1}\log(\frac{a+\eta x}{\eta x+1})}.
\]
Note that 
\begin{align*}
f'(a; x,\eta) &= \frac{\log a+1}{x+1} - \frac{1}{x+1}\left(1+\log(\frac{a+x}{x+1})\right) - \max(\eta,\eta^{-1})\parens*{\frac{\log a+1}{\eta x+1} - \frac{1}{\eta x+1}(1+\log(\frac{a+\eta x}{\eta x+1})) } \\
&= (\log a)\left(\frac{1}{x+1}-\frac{\eta}{\eta x+1}\right) - \left(\frac{1}{x+1}\log(\frac{a+x}{x+1}) - \frac{\eta}{\eta x+1}\log(\frac{a+\eta x}{\eta x+1})\right),
\end{align*}
and
\begin{align*}
f''(a; x,\eta)
&= \frac{1}{a}\left(\frac{1}{x+1}-\frac{\eta}{\eta x+1}\right) - \left(\frac{1}{(a+x)(x+1)} - \frac{\eta}{(a+\eta x)(\eta x+1)}\right)\\
&=\frac{1-\eta}{(x+1)(\eta x+1)}\parens*{\frac{1}{a} - \frac{a-\eta x^2}{(a+x)(a+\eta x)}}\\
&=\frac{1-\eta}{(x+1)(\eta x+1)} \frac{1}{a(a+x)(a+\eta x)}\parens*{(a+x)(a+\eta x)-a(a-\eta x^2)}.
\end{align*}
Then $f''(a;x,\eta) < 0$ for $\eta > 1$ and $a>0$. Furthermore, $f'(1;x,\eta)=0$. Hence, $f(a;x,\eta)$ achieves its maximum at $a=1$, where $f(1;x,\eta) = 0$. 
\end{proof}
\begin{proof}[Proof of \cref{lem:dirichlet-to-entropy}]
Without loss of generality, we can assume that $f_+\ge f_-=1$. Let $a=f_+ \in [1,C]$ and $x=\mu_-/\mu_+$, the inequality can be rewritten as 
\[
C\parens*{\frac{1}{x+1}a\log a - \frac{x+a}{x+1}\log(\frac{x+a}{x+1})} \ge \frac{1}{x+1}a\log a - \frac{x+a}{x+1} \frac{\log a}{x+1}.
\]
Let 
\[
F(a;x) = C\parens*{\frac{1}{x+1}a\log a - \frac{x+a}{x+1} \frac{\log a}{x+1}} - \parens*{\frac{1}{x+1}a\log a - \frac{x+a}{x+1}\log(\frac{x+a}{x+1})}.
\]
Differentiating with respect to $a$, we have 
\[
F'(a;x) = \frac{1}{(x+1)^2 a} \parens*{Cx(a-1+a\log a)+ (x+1)a\log(\frac{x+a}{a(x+1)})}.
\]
For 
\[
g(a;x)=Cx(a-1+a\log a)+ (x+1)a\log(\frac{x+a}{a(x+1)}),
\]
we have 
\begin{align*}
g'(a;x) &= C x (\log a + 2) - (1+x) \parens*{\log(\frac{ax+a}{a+x}) + \frac{x}{a+x}}\\
&\ge Cx(2+\log a) - (1+x) \cdot \frac{xa}{x+a} \\
&\ge Cx(2+\log a) - xa.
\end{align*}
We have $C(2+\log t)-t \ge 0$ for all $t\in [1,C]$, we have that $g(a;x) \ge 0$ for all $a\in [1,C]$. Thus, $F'(a;x) \ge 0$ and $F(a;x) \ge F(1;x) = 0$ for all $a\in [1,C]$, which proves the result. 
\end{proof}
\section{Deferred proofs from \texorpdfstring{\cref{sec:ising}}{Section \ref{sec:ising}}}
\label{sub:appendix-ising}

\begin{proof}[Proof of \cref{prop:ising flc}]
By continuity and the connection between fractional log-concavity and spectral independence (see \cite{alimohammadi2021fractionally}), it is enough to verify that $\lambda \ast \mu$ is $O(1/\delta) $-spectrally independent for all $\lambda = (\lambda_v)_{v\in V} \in \R_{> 0}^{V}.$ This follows from the proof of \cite[Theorem~1.4]{chen2021rapid}. There, it is shown that $\lambda \ast \mu$ is $O(1/\delta)$-spectrally independent for all $\lambda \in [0,1]^{\abs{V}}.$ We simply note that their proof doesn't use the condition $\lambda_v \leq 1$. Indeed, the only place in the proof where $\lambda_v$ appears is in definition of the interval $J_{\lambda_v, d_v},$ which is the domain of $y_v.$ But in the proof, \cite{chen2021rapid} obtain a uniform bound for all $y_v \in (-\infty, +\infty)$, so that the proof applies to all external field $\lambda \in \R^{V}_{>0}.$ 
\end{proof}

\begin{proof}[Proof of \cref{prop:lambda bound}]
We repeat the proof of Claim 25 of \cite{LLY14}, now with the gap parameter $\delta$. First observe
\begin{align*}
 x_{1,\delta}(d) &= \frac{ \zeta_{\delta} (d) - \sqrt{\zeta_{\delta} (d)^2 - 4(1-\delta)^2\beta \gamma } }{2(1-\delta) \beta} \\
&= \frac{4\beta \gamma (1-\delta)^2 }{2(1-\delta) \beta\parens*{\zeta_{\delta} (d) + \sqrt{\zeta_{\delta} (d)^2 - 4(1-\delta)^2\beta \gamma }}} \\
&= \frac{2 \gamma (1-\delta) }{\zeta_{\delta} (d) + \sqrt{\zeta_{\delta} (d)^2 - 4(1-\delta)^2\beta \gamma }}   \\
\end{align*}
thus 
\[ \frac{ \gamma (1-\delta) }{\zeta_{\delta} (d)}\leq x_{1,\delta}(d)\leq\frac{2 \gamma (1-\delta) }{\zeta_{\delta} (d)} \]
Note that
$f(x):= (\frac{x + \gamma}{\beta x + 1})^d = \parens*{\frac{1}{\beta}(1 - \frac{(1- \gamma \beta)}{\beta x+ 1})}^d$ is increasing in $[0,\infty),$ since $\beta \gamma \leq 1.$ In particular, $f(x) \geq f(0) = \gamma^d,$ so
\[\frac{ \gamma^{d+1} (1-\delta) }{\zeta_{\delta} (d)}\leq x_{1,\delta}(d)\gamma^d \leq \lambda_{1,\delta}(d) \leq \frac{2 \gamma (1-\delta) }{\zeta_{\delta} (d)} f\parens*{\frac{2 \gamma (1-\delta) }{\zeta_{\delta} (d)}} = \frac{2\gamma^{d+1}(1-\delta) }{\zeta_{\delta } (d)} \parens*{\frac{d+1-\delta}{d-1+\delta}}^d \]
which proves the result. 

Note that for $\beta \gamma \geq \frac{\Delta }{\Delta + \alpha}$
\[\zeta_{\delta} (d) = d(1-\beta\gamma) -(1-\delta)(1+\beta\gamma)  \leq d(1-\beta\gamma) \leq \alpha \]
Using the fact that $d \geq (1-\delta) \bar{\Delta}$ we have
\begin{align*}
  \zeta_{\delta}(d) &\geq (1-\delta) \parens*{ \bar{\Delta} (1-\beta\gamma) - (1+\beta\gamma)} = (1-\delta) \parens*{(1-\beta\gamma)\frac{1+\sqrt{\beta\gamma} }{1-\sqrt{\beta \gamma} } - (1+\beta\gamma) }  \\
  &= (1-\delta) \parens*{(1+\sqrt{\beta \gamma})^2 - (1+\beta\gamma)} = 2(1-\delta) \sqrt{\beta \gamma }  
\end{align*}
thus 
\[ \frac{(1-\delta) \gamma^{d+1} }{\alpha } \leq \lambda_{1,\delta}(d) \leq \frac{2 \gamma^{d+1} (1-\delta) }{2(1-\delta) \sqrt{\beta \gamma} } (\frac{d+1 -\delta}{d-1+\delta})^d \leq \frac{9\gamma^{d+1}}{\sqrt{\beta \gamma} }\]
and 
\[ \frac{(1-\delta) }{\alpha e^{\alpha}}\leq \frac{(1-\delta) (\beta \gamma)^{d+1} }{\alpha } \leq \beta^{d+1}\lambda_{1,\delta}(d) \leq \frac{\lambda_{1,\delta}(d)}{\gamma^{d+1}}\leq  \frac{9}{\sqrt{\beta \gamma }} \leq 9\sqrt{1+\alpha}. \qedhere\]
\end{proof}

\begin{proof}[Proof of \cref{prop:2 spin soft constraint}]

Note that we only need to prove Conclusion 1 for $\Lambda = V \setminus v.$ We view $\theta\ast \mu$ as an Ising model with per-site external field $0 \leq \lambda_v \leq \lambda$.
By \cref{eqn:conditional-2spin}, 
for each vertex $v$ and valid assignment $\sigma_{\sim v}$ of $V \setminus v:$ 
\[\lambda_v \beta^{\Delta_v} \leq \frac{\P_{\mu} {\sigma_v = + \given \sigma_{\sim v}} }{\P_{\mu} {\sigma_v = - \given \sigma_{\sim v}} } = \frac{\lambda_v \beta^{\Delta_v - s}}{\gamma^s} = \lambda_v \beta^{\Delta_v} \tau^s \leq\lambda_v \beta^{\Delta_v}\tau^{\Delta} \]
where $s$ is the number of $-1$’s assigned by $\sigma$ to the neighborhood $\mathcal{N}_v$ of $v,$ and the inequalities follows by using $\tau \geq 1$ and $0\leq s\leq \Delta_v \leq \Delta.$ 

For Conclusion 2, we again use \cref{eqn:conditional-2spin} to obtain 
\begin{align*}
    \P_{\mu} {\sigma_u = - \given \sigma^+_{\sim u}}  &= \frac{\gamma^s} {\lambda_u \beta^{\Delta_v - s} +\gamma^s }\\
    \P_{\mu} {\sigma_u = - \given \sigma^-_{\sim u}}  &= \frac{\gamma^{s+1}} {\lambda_u \beta^{\Delta_v - (s+1)} +\gamma^{s+1} }
\end{align*}
where $s$ is the number of $-1$’s assigned by $\sigma^+_{\sim u}$ to the neighborhood of $u.$ Thus
\begin{align*}
\frac{1}{\tau}  \leq \frac{\P_{\mu} {\sigma_u = + \given \sigma^+_{\sim u}} }{\P_{\mu} {\sigma_u = + \given \sigma^-_{\sim u}}} &= \frac{\beta^{\Delta_v - s}} {\lambda_u \beta^{\Delta_v - s} +\gamma^s }  (\frac{\lambda_u \beta^{\Delta_v - (s+1)} } {\lambda_u \beta^{\Delta_v - (s+1)} +\gamma^{s+1} })^{-1} = \beta \cdot \frac{\lambda_u \beta^{\Delta_v - (s+1)} +\gamma^{s+1}}{\lambda_u \beta^{\Delta_v - s} +\gamma^s } \leq 1\\
1\leq \frac{\P_{\mu} {\sigma_u = - \given \sigma^+_{\sim u}} }{\P_{\mu} {\sigma_u = - \given \sigma^-_{\sim u}}} &= \frac{\gamma^s} {\lambda_u \beta^{\Delta_v - s} +\gamma^s }  (\frac{\gamma^{s+1}} {\lambda_u \beta^{\Delta_v - (s+1)} +\gamma^{s+1} })^{-1} = \frac{1}{\gamma} \cdot \frac{\lambda_u \beta^{\Delta_v - (s+1)} +\gamma^{s+1}}{\lambda_u \beta^{\Delta_v - s} +\gamma^s } \leq \tau
\end{align*}
where the inequalities follow from
\begin{align*}
  \beta(\lambda_u \beta^{\Delta_v - (s+1)} +\gamma^{s+1}) &= \lambda_u \beta^{\Delta_v-s} + \gamma^s (\beta \gamma) \leq \lambda_u \beta^{\Delta_v-s} + \gamma^s\\
  \lambda_u \beta^{\Delta_v - (s+1)} +\gamma^{s+1} &= \gamma( \lambda_u \beta^{\Delta_v - s} \tau +\gamma^{s})\geq \lambda_u \beta^{\Delta_v - s} +\gamma^{s}
\end{align*}

For Conclusion 3, we only need to show that for all $d$
\[(1+\delta/2) \lambda_{1,\delta}(d) \leq \lambda_{1,\delta/2}(d). \]
First, note that $f(x):= (\frac{x + \gamma}{\beta x + 1})^d = \parens*{\frac{1}{\beta}(1 - \frac{(1- \gamma \beta)}{\beta x+ 1})}^d$ is increasing in $[0,\infty),$ since $\beta \gamma \leq 1.$ Next, we show that
\[(1+\delta/2) x_{1,\delta}(d) \leq x_1^{(\delta/2)}(d) \]
We can rewrite
\begin{align*}
(1+ \frac{\delta}{2}) x_{1,\delta}(d) &= \frac{(1+\delta/2) \parens*{\zeta_{\delta} (d) - \sqrt{\zeta_{\delta} (d)^2 - 4(1-\delta)^2\beta \gamma }} }{2(1-\delta) \beta} \\
&= \frac{4(1+\delta/2)(1-\delta)^2\beta \gamma  }{2(1-\delta) \beta\parens*{\zeta_{\delta} (d) + \sqrt{\zeta_{\delta} (d)^2 - 4(1-\delta)^2\beta \gamma }}}   \\
&\leq_{(1)} \frac{4(1-\delta/2)^2\beta \gamma  }{2(1-\delta/2) \beta\parens*{\zeta_{\delta} (d) + \sqrt{\zeta_{\delta} (d)^2 - 4(1-\delta)^2\beta \gamma }}}\\
&\leq_{(2)} \frac{4(1-\delta/2)^2\beta \gamma  }{2(1-\delta/2) \beta\parens*{\zeta_{\delta/2} (d) + \sqrt{\zeta_{\delta/2} (d)^2 - 4(1-\delta/2)^2\beta \gamma }}}\\
&= x_{1,\delta/2}(d) 
\end{align*}
where in (1) we use $(1-\delta)(1+\delta/2) \leq (1-\delta/2)$ and in (2) we use the observation that for $d\geq 1 $ 
\[g(y) =  \zeta_{y} (d) + \sqrt{\zeta_{y} (d)^2 - 4(1-y)^2\beta \gamma } \]
is strictly increasing in $[0,1].$ Indeed, \[\zeta_y(d) = d(1-\beta\gamma) - (1-y) (1+\beta\gamma)\] is strictly increasing in $y,$ and
\[h(y) := \zeta_y(d)^2 -4(1-y)^2\beta \gamma = d^2 (1-\beta\gamma)^2 - 2(1-y)d (1-(\beta\gamma)^2) + (1-y)^2 (1-\beta\gamma)^2 \]
has non-negative derivative in $y$
\[h'(y) : = 2(1-\beta\gamma) (d (1+\beta\gamma) - (1-\beta\gamma) (1-y) ) \geq 0\]
Finally,
\[(1+\delta/2) \lambda_{1,\delta}(d) = (1+\delta/2) x_{1,\delta}(d) \, f(x_{1,\delta}(d) ) \leq x_{1,\delta/2}(d) \, f(x_{1,\delta/2}(d) ) =\lambda_{1,\delta/2}(d). \qedhere \]

\end{proof}

\end{document}